%% file: ft-hd-wrap.tex
\documentclass[11pt]{article}
\usepackage{fullpage}
\usepackage{times}
\usepackage{subfigure}
\usepackage{color}
\usepackage{url}
\usepackage{graphicx}
\usepackage{amsmath}
\usepackage{amsthm}
\usepackage{amssymb}
\usepackage{algorithm}
\usepackage{algpseudocode}
\usepackage{mathtools}
\usepackage{tikz}
\usetikzlibrary{positioning, calc, chains}

{\makeatletter
 \gdef\xxxmark{%
   \expandafter\ifx\csname @mpargs\endcsname\relax % in minipage?
     \expandafter\ifx\csname @captype\endcsname\relax % in figure/caption?
       \marginpar{xxx}% not in a caption or minipage, can use marginpar
     \else
       xxx % notice trailing space
     \fi
   \else
     xxx % notice trailing space
   \fi}
 \gdef\xxx{\@ifnextchar[\xxx@lab\xxx@nolab}
 \long\gdef\xxx@lab[#1]#2{{\bf [\xxxmark #2 ---{\sc #1}]}}
 \long\gdef\xxx@nolab#1{{\bf [\xxxmark #1]}}
}

\newcommand{\abs}[1]{|#1|}

\newcommand{\norm}[2]{\lVert#2\rVert_{#1}}

\newcommand{\wh}{\widehat}

\newcommand{\twid}[1]{\tilde{#1}}

\newcommand{\M}{{\mathcal M}}

\DeclareMathOperator{\supp}{supp}

\def\R{\mathbb{R}}

\def\C{\mathbb{C}}
\def\F{\mathcal{F}}
\def\B{\mathbb{B}}
\def\eps{\epsilon}

\newcommand{\expect}{{\bf \mbox{\bf E}}}
\newcommand{\prob}{{\bf \mbox{\bf Pr}}}

\definecolor{gray}{rgb}{0.5,0.5,0.5}
\newcommand{\e}{{\epsilon}}
\newcommand{\A}{{\mathcal{A}}}
\newcommand{\E}{{\mathcal{E}}}

\newtheorem{theorem}{Theorem}[section]
\newtheorem{lemma}[theorem]{Lemma}
\newtheorem{claim}[theorem]{Claim}

\newtheorem{definition}[theorem]{Definition}
\newtheorem{remark}[theorem]{Remark}

\newenvironment{proofof}[1]{\noindent{\bf Proof of #1:}}{$\qed$\par}
\newenvironment{proofsketch}{{\sc{Proof Outline:}}}{$\qed$\par}
\usepackage{hyperref}
\hypersetup{
bookmarksnumbered
}

\DeclareMathOperator{\err}{Err}
\newcommand{\poly}{\text{poly}}

\newcommand{\fc}{F}
\newcommand{\gl}{\mathcal{M}_{odd}}

\renewcommand{\H}{\mathcal{W}}
\newcommand{\h}{\mathbf{w}}

\newcommand{\quant}{\text{quant}}
\begin{document}

\begin{titlepage}
\title{Sample Efficient Estimation and Recovery in Sparse FFT via Isolation on Average}
\author{Michael Kapralov\thanks{School of Computer and Communication Sciences, EPFL, Lausanne, Switzerland. Email: {\tt michael.kapralov@epfl.ch} }}

\maketitle

\begin{abstract}
The problem of computing the Fourier Transform of a signal whose spectrum is dominated by a small number $k$ of frequencies  quickly and using a small number of samples of the signal in time domain (the Sparse FFT problem) has received significant attention recently. It is known how to approximately compute the $k$-sparse Fourier transform in $\approx k\log^2 n$  time [Hassanieh et al'STOC'12], or using the optimal number $O(k\log n)$ of samples [Indyk et al'FOCS'14] in time domain, or come within $(\log\log n)^{O(1)}$ factors of both these bounds simultaneously, but no algorithm achieving the optimal $O(k\log n)$ bound in sublinear time is known.

At a high level, sublinear time Sparse FFT algorithms operate by `hashing' the spectrum of the input signal into $\approx k$ `buckets', identifying frequencies that are `isolated' in their buckets, subtracting them from the signal and repeating until the entire signal is recovered. The notion of `isolation' in a `bucket', inspired by  applications of hashing in sparse recovery with arbitrary linear measurements, has been the main tool in the analysis of Fourier hashing schemes in the literature. However, Fourier hashing schemes,  which are implemented via filtering, tend to be `noisy' in the sense that a frequency that hashes into a bucket contributes a non-negligible amount to neighboring buckets. This leakage to  neighboring buckets makes identification and estimation challenging, and the standard analysis based on isolation becomes difficult to use without losing $\omega(1)$ factors in sample complexity.

In this paper we propose a new technique for analysing noisy hashing schemes that arise in Sparse FFT, which we refer to as {\em isolation on average}. We apply this technique to two problems in Sparse FFT: estimating the values of a list of frequencies using few samples and computing Sparse FFT itself, achieving sample-optimal results in $k\log^{O(1)} n$ time for both.  We feel that our approach will likely be of interest in designing Fourier sampling schemes for more general settings (e.g. model based Sparse FFT).
\end{abstract}

\thispagestyle{empty}
\end{titlepage}

\newcommand{\nsq}{{[n]}}

\newcommand{\ifshort}[1]{#1}
\newcommand{\iflong}[1]{}
\newcommand{\duallabel}[1]{\label{#1}}
\newcommand{\dualref}[1]{\ref{#1}}
\newcommand{\dualeqref}[1]{\eqref{#1}}

\renewcommand{\ifshort}[1]{}
\renewcommand{\iflong}[1]{#1}
\renewcommand{\duallabel}[1]{\label{#1-extended}}
\renewcommand{\dualref}[1]{\ref{#1-extended}}
\renewcommand{\dualeqref}[1]{\eqref{#1-extended}}
\input{ft-hd.tex}

%\pdfbookmark[1]{\refname}{My\refname} 
%\bibliographystyle{alpha}
%\bibliography{paper}

\newcommand{\etalchar}[1]{$^{#1}$}

\end{document}

%% file: ft-hd.tex
%!TEX root = ./ft-hd-wrap.tex 
\input{intro.tex}

\input{preliminaries.tex}

\input{l1-part1.tex}
\input{l1.tex}
\input{est.tex}
\input{algo.tex}
\input{combo.tex}

\begin{appendix}
\input{app-nice-partitions.tex}

\input{estvals.tex}
\input{loc.tex}

%\input{app-least-squares.tex}
\input{app-tail.tex}
\input{app.tex}
\end{appendix}

%% file: intro.tex
%!TEX root = ./ft-hd-wrap.tex 
%Fourier sampling is important

\section{Introduction}

The Discrete Fourier Transform (DFT) is a fundamental computational primitive with numerous applications in areas such as digital signal processing, medical imaging and data analysis as a whole. The fastest known algorithm for computing the Discrete Fourier Transform of a signal of length $n$ is the FFT algorithm, designed by Cooley and Tukey in 1965. The efficiency of FFT, which runs in time $O(n\log n)$ on any signal of length $n$, has contributed significantly to its popularity as a computational primitive, making FFT one of the top 10 most important algorithms of the 20th century~\cite{citeulike:6838680}. However, computational efficiency of FFT is not the only reason why the Fourier transform emerges in many applications: in signal processing the Fourier basis is often a convenient way of representing signals since it concentrates their energy on a few components, allowing compression (which is the rationale behind  image and video compression schemes such as JPEG and MPEG), and in medical imaging applications such as MRI the Fourier transform captures the physics of the measurement process (the problem of reconstructing an image from MRI data is exactly the problem of reconstructing a signal $x$ from Fourier measurements of $x$).  While FFT works for worst case signals,  signals arising in practice often exhibit structure that can be exploited to speed up the computation of the Fourier transform. For example, it is often the case that most of the energy of these signals is concentrated on  a small number of components in Fourier domain. In other words, the signals that arise in applications are often {\em sparse} (have a small number of nonzeros) or {\em approximately sparse} (can be well approximated by a small number of dominant coefficients) in the Fourier domain.  This motivates the question of (approximately) computing the Fourier transform of a signal that is (approximately) sparse in Fourier domain using {\em few samples of the signal in time domain (i.e. with small sample complexity)} and {\em small runtime}. We note that while runtime is a natural parameter to optimize, sample complexity is at least as important in applications such as medical imaging, where sample complexity governs the {\em measurement} complexity of the imaging process.

In this paper we consider the problem of computing a sparse approximation to a signal $x\in \mathbb{C}^n$ given access to its Fourier transform $\wh{x}\in \mathbb{C}^n$, which is equivalent to the problem above (since the inverse Fourier transform only differs from the Fourier transform by a conjugation), but  leads to somewhat more compact notation. This problem  has been studied extensively. The seminal work of~\cite{CTao,RV} in {\em compressed sensing}  first showed that length $n$ signals with at most $k$ Fourier coefficients can be recovered using only $k \log^{O(1)} n$ samples in time domain. The recovery algorithms are based on linear programming and run in time polynomial in $n$.  A different line of research on the {\em Sparse Fourier Transform} (Sparse FFT), originating from computational complexity and learning theory, has  resulted in algorithms that use $k\log^{O(1)} n$ samples  and $k\log^{O(1)} n$ runtime (i.e. the runtime is {\em sublinear} in the length of the input signal). Many 
such algorithms have been proposed in the literature, including \cite{GL,KM,Man,GGIMS,AGS,GMS,Iw,Ak,HIKP,HIKP2,LWC,BCGLS,HAKI,pawar2013computing,heidersparse, IKP, IK14a,K16,PZ15,ChenKPS16}.  Nevertheless, despite significant progress that has recently been achieved, important gaps in our understanding of sample and time efficient recovery from Fourier measurements remain. We address some of these gaps in this work.

The main contribution of this work is a new technique for designing and analyzing sample efficient sublinear time Sparse FFT algorithms. We refer to this technique as {\em isolation on average}. We apply our technique to two problems in the area of Sparse Fourier Transform computation, namely estimation and recovery with Fourier measurements.

\paragraph{Our results: estimation}   In the first problem we are given a subset $S\subseteq [n]$ of locations in the time domain and are asked to estimate $x_S$ from a few values of $\wh{x}$. Formally, we would like the algorithm to output a signal $x'$ with $\supp (x')\subseteq S$ such that 
\begin{equation}
\label{e:l2l2-est}
\| x-x'\|_2^2 \le (1+\e) \|x_{[n]\setminus S}\|_2^2
\end{equation}
In other words, we would like to output an estimate $x'$ of $x$ that is correct up to the 'noise', i.e. elements outside of $S$ (to achieve~\eqref{e:l2l2-est}, it suffices to ensure that $\| (x-x')_S\|_2^2 \le \e \|x_{[n]\setminus S}\|_2^2$). Note that in some of the applications described above one often has a good prior on which coefficients of $x$ are the dominant ones, and a natural question is whether one can recover the values of $x_S$ quickly, using few samples, and in a noise robust manner, i.e. solve~\eqref{e:l2l2-est}.  Our main result on estimation is Algorithm~\ref{alg:estimate-efficient} (presented in Section~\ref{sec:est}) together with
\begin{theorem}\label{thm:main-estimate}
For every $\e\in (1/n, 1), \delta\in (0, 1/2)$, $x\in \C^n$ and every integer $k\geq 1$, any $S\subseteq \nsq$, $|S|=k$, if $||x||_\infty\leq R^*\cdot ||x_{\nsq\setminus S}||_2/\sqrt{k}, R^*=n^{O(1)}$, an invocation of \textsc{Estimate}$(\hat x, S, k, \e, R^*)$ (Algorithm~\ref{alg:estimate-efficient}) returns $\chi^*\in \C^n$ such that 
$$
||(x-\chi^*)_S||_2^2\leq \e\cdot ||x_{\nsq\setminus S}||_2^2
$$
using $O_\delta(\frac1{\e}k)$ samples and $O_\delta(\frac1{\e}k \log^{3+\delta} n)$ time with at least $4/5$ success probability. 
\end{theorem}

A linear sketch with $O(k)$ measurements and $O(k)$ recovery time that provides the guarantee in~\eqref{e:l2l2-est} was presented in~\cite{Price11}, but this solution uses general linear measurements as opposed to the more restrictive Fourier measurements. To the best of our knowledge, the estimation problem with guarantees ~\eqref{e:l2l2-est} has not been studied explicitly in the setting of Fourier measurements.   We now describe `folklore' results and give a comparison with Theorem~\ref{thm:main-estimate}.

\paragraph{Estimation from Fourier measurements: least squares}  Recall that the problem is as follows: given a set $S\subseteq [n]$, estimate $x_S$ from a small number of Fourier measurements of $x$, i.e. from a small number of accesses to $\wh{x}$. 
A popular approach is to select a subset $T\subseteq [n]$ of frequencies and solve the least squares problem 
\begin{equation}\label{eq:least-squares}
\begin{split}
\min_{y\in \C^n, \supp y\subseteq S} ||\wh{y}_T-\wh{x}_T||_2^2.\\
\end{split}
\end{equation}
A natural choice is to let $T$ be a (multi)set of frequencies selected uniformly at random with replacement from $[n]$.  The solution to~\eqref{eq:least-squares} is then provided by the normal equations 
$y_{OPT}=(F_{T, S}^* F_{T, S})^{-1} F_{T, S}^*x$, where $Fx\in \C^n$ is the Fourier transform of $x$, and $F_{T, S}$ is the $T\times S$ submatrix of $F$ scaled by $\sqrt{n/|T|}$. Writing $x=x_S+x_{[n]\setminus S}$, so that $\wh{x}_T=F_{T, S}x_S+F_{T, [n]\setminus S} x_{[n]\setminus S}$, we get 
$y_{OPT}=(F_{T, S}^* F_{T, S})^{-1} F_{T, S}^*\wh{x}_T=x_S+(F_{T, S}^* F_{T, S})^{-1} F_{T, S}^* F_{T, [n]\setminus S} x_{[n]\setminus S}$,
where the second term corresponds to the estimation error due to tail noise. Thus, if $T$ is such that $\frac{1}{2}I_{S}\preceq F_{T, S}^* F_{T, S} \preceq 2I_S$, 
then $||y_{OPT}-x||_2=O(1)\cdot ||x_{\nsq\setminus S}||_2$ with constant probability.  A simple application of matrix Chernoff bounds shows that the spectral bound $\frac{1}{2}I_{S}\preceq F_{T, S}^* F_{T, S} \preceq 2I_S$ is satisfied when $|T|\geq C |S|\log |S|$ for an absolute constant $C$. Note that the analysis above is tight, as  for certain choices of $S\subseteq[n]$ at least $\Omega(|S|\log |S|)$ samples are needed even to ensure that $F_{S, T}^* F_{S, T}$ is invertible. For example, suppose that $S=(n/k)\cdot [k]$, where $k$ divides $n$, so that  the signal $\wh{x}$ is $k$-periodic. In this case $x$ only becomes recoverable from $\wh{x}_T$ as long as $T$ contains at least one element of every conjugacy class of $\mathbb{Z}_n$ modulo $k$, and by a Coupon Collection argument $\Omega(k\log k)$ samples are needed to ensure that this is the case. 
To summarize, the sample complexity of least squares with a random $T$ is at least $\Omega(k\log k)$. Another significant disadvantage of this approach is that solving the least squares problem requires at least $\Omega(k^2)$ runtime using current techniques. Of course, given the knowledge of $S$ one may be able to design a better than random set $T$, but no such construction is known for general supports $S$. As Theorem~\ref{thm:main-estimate} shows, there exists a distribution over sampling patterns $T$ that is {\em oblivious to $S$} and allows decoding from $O(k)$ samples in $k\log^{O(1)} n$ time.

\paragraph{Estimation from Fourier measurements: Fourier hashing}  Estimation of a subset $S$ of coefficients of $x$ using Fourier measurements can be performed using the idea of Fourier hashing (via filtering) commonly used in the Sparse FFT literature. In this approach one round of hashing allows one to compute estimates $w_i$ for $x_i, i\in S$ such that 
$$
|w_i-x_i|\leq \alpha ||x||_2^2/k
$$
using $O(k/\alpha)$ samples in Fourier domain and $O((k/\alpha) \log (k/\alpha))$ runtime. Here $\alpha\in (0, 1)$ is the oversampling parameter, which is normally set to a small constant, as it directly affects runtime and sample complexity. The approach is similar to standard hashing techniques such as \text{CountSketch}~\cite{CCF}, but the crucial difference is that the error bound depends on the {\em energy of the entire signal} as opposed to energy of the tail\footnote{One way to improve the error bound is to use strong filters~\cite{HIKP}, but that requires a $\Omega(k\log n)$ samples -- see Lemma~\ref{lem:semi_equi_std}}. Indeed, in general one can have $||x||_2^2\gg n^{\Omega(1)}\cdot||x_{[n]\setminus S}||_2^2$, meaning that one round of hashing gives results that are very far from estimating $x_S$ up to the energy of the `noise', i.e. elements outside $S$. This can be fixed by iterating the estimation process on the residual signal. A naive implementation and analysis results in $\Theta(k\log n)$  measurements (due to $\log n$ iterations of refinement) and $k\log^{O(1)} n$ time. Recent works on saving samples by reusing measurements~\cite{IK14a, K16} can lead to improvements over the factor $\log n$ blow up in sample complexity, but all prior approaches inherently lead to $\omega(1)$ factor loss in the number of samples, as we argue below.

\paragraph{Our results: recovery} The second version of the problem is the Sparse FFT (recovery) problem with $\ell_2/\ell_2$ guarantees: we are given access to $\wh{x}$, a precision parameter $\e>0$ and a sparsity parameter $k$, and would like to output $x'$ such that 
\begin{equation}
\label{e:l2l2}
\| x-x'\|_2^2 \le (1+\e) \min_{k \text{-sparse } y }  \|x-y\|_2^2,
\end{equation}
Note that here we are not provided with any information about the `heavy' coefficients of $x$, and the hardest and most sample intensive part of the problem is to recover the identities of the `heavy' elements.

  It is known that any (randomized, non-adaptive) algorithm whose output satisfies~\eqref{e:l2l2} with at least constant probability must use $m=\Omega(k\log (n/k))$ samples~\cite{DIPW}.  An algorithm that matches this bound for every $k\leq n^{1-\delta}$ was recently proposed by~\cite{IK14a}. The algorithm of~\cite{IK14a} required $\Omega(n)$ runtime, however, leaving open the problem of achieving sample-optimality  in $k\log^{O(1)} n$, or even just {\em sublinear time}. Sublinear time algorithms that come close to the optimal sample complexity (within an $O(\log\log n)$ factor) have been proposed~\cite{IKP, K16}, but no algorithm was able to match the lower bound to within constant factors using sublinear runtime\footnote{In this paper we are only interested in algorithms that work for worst case signals. If probabilistic assumptions on the signal are made, better results are possible in some settings (see, e.g.~\cite{GHIKPS}).}.  As we argue below, achieving the $O(k\log n)$ bound in sublinear time appears to require a fundamentally different approach to Fourier hashing, which we provide in this work. Our new technique results in an algorithm that matches the lower bound of~\cite{DIPW} up to constant factors for every $k$ polynomially bounded away from $n$ (i.e. $k\leq n^{1-c}$ for a constant $c>0$) in sublinear time:
\begin{theorem}\label{thm:main}
For any $\e\in (1/n, 1), \delta\in (0, 1/2)$, $x\in \C^n$ and any integer $k\geq 1$, if $R^*\geq ||x||_\infty/\mu, R^*=n^{O(1)}$, $\mu^2\geq ||x_{\nsq\setminus [k]}||_2^2/k$, $\mu^2=O(||x_{\nsq\setminus [k]}||_2^2/k)$, \textsc{SparseFFT}$(\hat x, k, \e, R^*, \mu)$  (Algorithm~\ref{alg:main-sublinear}) solves the $\ell_2/\ell_2$ sparse recovery problem using $O_\delta(k\log n)+O(\frac1{\e}k\log n)$ samples and 
$O_\delta(\frac1{\e}k \log^{4+\delta} n)$ time with at least $4/5$ success probability.
\end{theorem}

We now discuss the technical difficulties that our approach overcomes. In this discussion we concentrate mainly on the estimation problem, as it is easier than sparse recovery, but at the same time exhibits all the relevant technical challenges. We first describe known sample optimal and efficient solutions that use arbitrary linear measurements, and then outline the difficulties that one faces when working with Fourier measurements.

\paragraph{Estimation and sparse recovery with arbitrary linear measurements.}  %The overall outline of our algorithm follows the framework of~\cite{GMS, HIKP2,IKP,IK14a, K16}, which adapt techniques developed in the literature on sparse recovery with arbitrary linear measurements (e.g. ~\cite{CCF,GLPS}) to the Fourier setting. The main idea of the approach is to ``hash'' elements of the spectrum into $B=O(k)$ ``buckets''. 
If arbitrary linear measurements are allowed, one takes, multiple times,  a set of $B=O(k)$
linear measurements of the form $\twid{u}_j = \sum_{i : h(i) = j} s_i x_i$ for a random hash function $h: [n] \to [B]$ and random signs
$s_i\in \{-1, +1\}$.  Since we are hashing in a number of buckets a constant factor (say, $100$) larger than the sparsity of the signal, a large fraction (say, $\approx 90\%$) of the top $k$ components are likely to be isolated in a bucket, and not have too much noise (i.e. elements other than the top $k$) hash into the same bucket. For such isolated elements we can approximate their value {\em up to the noise that hashes into the same bucket} (in the case of sparse recovery, we perform $O(\log (n/k))$ specially crafted linear measurements using the same hash function $h$ to recover the identity of the isolated element). This lets us estimate (resp. recover) $\approx 90\%$ of the top $k$ elements of the signal, we subtract them off and recurse on the remaining $\approx 10\%$ of the top $k$ elements, hashing  into $k/2$ buckets this time. In general, for $t=1,2,\ldots, O(\log k)$ we choose a random hash function $h_t: [n]\to [B_t]$, where $B_t=100k/2^{t-1}$, say (in the case of sparse recovery  we take $O(\log (n/k))$ measurements using each of these hash functions). One can show ~\cite{GLPS} that after $O(\log k)$ iterations of the hashing, recovery and subtraction process we recover an approximation to $x$ that satisfies \eqref{e:l2l2-est} (resp. ~\eqref{e:l2l2} in case of recovery). The sample complexity of this process is dominated by the sample complexity of the first iteration, where we use $B_1=100k$ buckets, resulting in a $O(k)$ (resp. $O(k\log (n/k))$) bound on the sample complexity overall. Note that the recovery process only uses every hash function $h_t$ once, at step $t$: those elements that are isolated under this hashing are perfectly recovered and essentially `disappear' from the system, so $h_t$ can be discarded!

\paragraph{A natural approach to estimation and recovery with Fourier measurements and why it fails.} %As before, we concentrate mainly on the estimation problem, as it is simpler notationally than sparse recovery, but presents the same set of difficulties to overcome. 
In order to achieve $O(k)$ (resp. $O(k\log n)$) sample complexity using Fourier measurements (i.e. in Sparse FFT) it seems natural to revisit the original idea used in recovery from arbitrary linear measurements that we outlined above.  More precisely, we could follow the strategy of choosing, for $t=1,2,\ldots, O(\log k)$, a random hash function $h_t: [n]\to [B_t]$, where $B_t=100k/2^{t-1}$. The problem is that in order to ensure that we hash into $B$ buckets at the cost of $O(B)$ samples, we need to commit to working with rather low quality buckets implemented using crude filters (see section~\ref{sec:prelim}) and this causes `leakage' between hash buckets.  Given this complication, it is not clear at all if estimation (resp. recovery) can be made to work: while with `ideal' hashing each element isolated in a hashing was identified and estimated {\em up to amount of noise in its bucket}, here due to the leakage of our simple filters identification of nominally isolated elements can be precluded by interference from other head elements! This means that the elements that were isolated in the first hashing do not `disappear' from the system (as they essentially do with `ideal' hashing described above in the context of arbitrary linear measurements), but are reduced in value by only about a constant factor, and will influence the recovery process using the second hashing etc. To put this in perspective, note that when for each $t>10$, say, we hash into $B_t=100k/2^{t-1}$ buckets, we generally get $\Omega(k)$ original elements hashing to $\ll k$ buckets! These elements have of course been reduced in value somewhat, but not to the extent that their contribution to $B_t\approx k/2^{t-1}$ buckets is negligible. 

The discussion above implies that two difficulties must be overcome to achieve $O(k)$ (resp. $O(k\log n)$) sample complexity. First, since one round of hashing can at most reduce the `isolated' elements in the residual by a constant factor, $\Omega(\log n)$ iterations are necessary. Furthermore, the process must be set up in such a way that the $\Omega(\log n)$ iterations operate on the same hash functions, and at the same time no adversarial correlations arise to hinder the estimation process. The second difficulty is more subtle, but the harder one to deal with -- this is exactly where our main contribution comes in. Note that if several levels of hashing are used, as above there could be elements whose total contribution to estimation error {\em over all levels $t>1$ is $\omega(1)$}. Indeed, it is easy to see that some of the top $k$ elements will participate in repeated collisions for many values of $t>1$. Such elements could pose significant difficulties, as they introduce large errors to the identification and estimation process.  This issue arises because we reuse hashings that hash $\Omega(k)$ elements into $\ll k$ buckets.  Thus,  we cannot hope to rely on isolation properties that all prior work is based on, since there are more elements to be estimated than buckets.

\paragraph{Our techniques: a new hashing scheme and isolation on average.} To overcome the difficulties outlined above, we use the following approach. As above, we choose a sequence of hash functions $h_t$ that hash the signal into a geometrically decreasing number of buckets. However, a crucial modification is that for each $t$ we repeat the hashing process independently $R_t$ times for an increasing sequence $R_t$ (we use a geometrically increasing sequence; our hashings are denoted by $h_{t, s}$, $s=1,\ldots, R_t$ for each $t$). As we show below in Section~\ref{sec:l1}, the independent repetitions ensure, at a high level, that despite the fact that most elements collide in multiple hashings, the fraction of such collisions is small, ensuring that estimation errors do not propagate -- see Lemma~\ref{lm:l1b-head} and Remark~\ref{rm:overview} after the lemma.

We give a formal analysis of our scheme in the rest of the paper, and provide intuition as to why our scheme fixes the problem outlined above now. Specifically, we would like to see that the head elements do not contribute a large fraction of their weight as estimation error in hashings $h_{t, s}$ for $t>1$. The reason is that, as we show below, given the hash functions $\{h_{t, s}\}$ the set $S$ of head elements can be partitioned into sets $S=S_1\cup S_2\cup\ldots\cup S_T, |S_1|\gg |S_2| \gg \ldots \gg |S_T|$ so that for every $t>1$ every element of $S$ collides with at least one element of $S_t$ in {\em no more than $R_t^{1-\delta}$ out of the $R_t$ hashings $h_{t, s}$ at iteration $t$}, for some constant $\delta>0$ (choosing $\delta$ small improves runtime, at the expense of sample complexity; any small constant $\delta>0$ leads to asymptotically sample optimal results). Thus, even though there are many collisions, {\bf on average over $s\in [1:R_t]$} every element in $S$ collides with at most $\approx R_t^{-\delta}$ elements of $S_t$ -- we refer to this property as `isolation on average'. Since we choose the number of hashings $R_t$ to increase geometrically, the error contributed by an element of $S$ {\em over all hashings} is no more than $\sum_{t\geq 1} R_t^{-\delta}=O_\delta(R_t^{-\delta})\ll 1$. This fact allows us to argue that iterative decoding converges (see section~\ref{sec:l1}). Achieving small runtime with such a scheme requires a delicate balance of parameters, which we exhibit in Section~\ref{sec:est}.

\paragraph{Our techniques: majorizing sequences for controlling residual signals.} Lastly, one should note that the discussion above rests heavily on our ability to control the sequence of residual signals that arise throughout the update process (both in estimation and recovery). We achieve this by showing that  residual signals arising during the update process are  {\em majorized} by short (polylogarithmic length) sequence of signals (referred to as a {\em majorizing sequence}). See Section~\ref{sec:lm41} for the application in estimation and Section~\ref{sec:maj} for the application in recovery. %A rather basic version of this idea was used in~\cite{IK14a} to achieve $O(k\log n)$ sample complexity for Sparse FFT in $n\log^{O(1)} n$ time, but the sublinear

\paragraph{Significance for future work.}
We feel that the idea of `isolation on average' may prove useful in further developments in the area. For example, it would be interesting to see if measurement reuse using our techniques can improve sample complexity of to sublinear algorithms for {\em model based sparse recovery} from Fourier measurements, i.e. to Sparse FFT algorithms that {\em exploit structure of input signals beyond the sparsity assumption} (the a sublinear time algorithm for model based Sparse FFT for the block-sparse model was recently presented in~\cite{CKSZ17}). A strong step in this direction would consist of removing the reliance of our techniques on the $\ell_1$ norm of the residual signal as the measure of progress, and introducing an approach to measurement reuse while provably reducing the $\ell_2$ norm of the residual during the iterative process.

\paragraph{Organization.}  The proofs of Theorem~\ref{thm:main-estimate} and Theorem~\ref{thm:main} rely on a shared set of lemmas that enable analysis via `isolation on average', with the main technical lemma being Lemma~\ref{lm:l1b-head} (see also Remark~\ref{rm:overview} after the lemma). We present these lemmas first (Sections~\ref{sec:prelim} and~\ref{sec:l1}), then prove Theorem~\ref{thm:main-estimate} (Section~\ref{sec:est}) as it is less notationally heavy but still uses all the main technical ideas, and then prove Theorem~\ref{thm:main} (Section~\ref{sec:sublinear}). 
%We concentrate on the estimation result (Theorem~\ref{thm:main-estimate}) due to the fact that it less technical than the sparse recovery result (Theorem~\ref{thm:main}) and at the same time uses a shared set of main lemmas. Preliminaries and basic notation needed for the estimation primitive are presented in section~\ref{sec:prelim}. The first technical innovation of the paper, i.e. the notion of isolation on average and the main technical lemma (Lemma~\ref{lm:l1b-head}) are presented in Section~\ref{sec:l1}. Then Lemma~\ref{lm:l1b-head} is combined with our second technical innovation, i.e. the notion of majorizing sequences,  to obtain a sample efficient estimation algorithm (Algorithm~\ref{alg:estimate-efficient}) in Section~\ref{sec:est}.  Finally, in Section~\ref{sec:sublinear} we present our sample-optimal sublinear Sparse FFT algorithm, Algorithm~\ref{alg:main-sublinear}, and give the analysis  (which at a high level follows the analogous argument for Algorithm~\ref{alg:estimate-efficient} from Section~\ref{sec:est}). 
Proofs omitted from the main body of the paper are given in the Appendices.

%% file: preliminaries.tex
%!TEX root = ./ft-hd-wrap.tex 
\section{Preliminaries and basic notation}\label{sec:prelim}

For a positive even integer $a$ we will use the notation $[a]=\{-\frac{a}{2}, -\frac{a}{2}+1, \ldots, -1, 0, 1,\ldots, \frac{a}{2}-1\}$. We will consider signals of length $n$, where $n$ is a power of $2$.  We use the notation $\omega=e^{2\pi i/n}$ for the root of unity of order $n$. The forward and inverse Fourier transforms are given by
\begin{equation}\label{eq:dft-forward}
\hat x_f=\frac1{\sqrt{n}}\sum_{i\in \nsq}  \omega^{-if}x_i \text{~~and~~}x_{j}=\frac1{\sqrt{n}}\sum_{f\in \nsq}  \omega^{jf}\hat x_f
\end{equation}
respectively, where $f, j\in \nsq$. We will denote the forward Fourier transform by $\F$. 
Note that we use the orthonormal version of the Fourier transform. Thus, we have $||\hat x||_2=||x||_2$ for all $x\in \C^n$ (Parseval's identity). We assume that entries of $x$ are integers bounded by a polynomial in $n$.

\subsection{Filters, hashing and pseudorandom permutations}\label{sec:prelims-permutations}
We will use pseudorandom spectrum permutations, which we now define.  We write $\gl$ for the set of odd numbers between $1$ and $n$.
For $\sigma\in \gl, q\in \nsq$ and $i\in \nsq$ let 
$\pi_{\sigma, q}(i)=\sigma(i-q) \mod n$.
Since $\sigma\in \gl$, this is a permutation. Our algorithm will use $\pi$ to hash heavy hitters into $B$ buckets, where we will choose $B\approx k$. We will often omit the subscript $\sigma, q$ and simply write $\pi(i)$ when $\sigma, q$ is fixed or clear from context. For $i\in [n]$ we let $h(i):=\text{round}((B/n)\pi(i))$ be a hash function that maps $[n]$ to $[B]$, and for $i, j\in \nsq$ we let $o_i(j)= \pi(j) - (n/B) h(i)$  be the ``offset'' of $j\in \nsq$ relative to $i\in \nsq$.  We always have $B$ a power of two.

\iflong{

\begin{definition}
  Suppose that $\sigma^{-1}$ exists $\bmod~n$. For $a, q\in \nsq$ we define the
  permutation $P_{\sigma, a, q}$ by $(P_{\sigma, a, q}\hat x)_i=\hat x_{\sigma(i-a)} \omega^{i\sigma q}$.  
\end{definition}
\begin{lemma}\label{lm:perm}
$\F^{-1}({P_{\sigma, a, q} \hat x})_{\pi_{\sigma, q}(i)}=x_i \omega^{a \sigma i}$
\end{lemma}
The proof is given in~\cite{IK14a} and we do not repeat it here.  Define 
\begin{equation}\label{eq:mu-def}
\begin{split}
\err_k(x)=\min_{k-\text{sparse}~y} ||x-y||_2\text{~~and~~}\mu^2=\err_k^2(x)/k.
\end{split}
\end{equation}}
\ifshort{Given access to samples of $\wh{x}$, we recover a signal $z$ such that 
$||x-z||_2\leq (1+\e)\min_{k-\text{~sparse~} y} ||x-y||_2$. Define 
$\err_k(x)=\min_{k-\text{sparse}~y} ||x-y||_2$ and $\mu^2=\err_k^2(x)/k$. } 
In this paper, we assume knowledge of $\mu$ (a constant factor upper bound on $\mu$ suffices). We also assume that the signal to noise ratio is bounded by a polynomial in the length $n$ of the signal, namely 
that $R^*:=||x||_\infty/\mu\leq n^C$ for a constant $C>0$. It will be convenient to use the notation $\B_\infty(x, r)$ to denote the interval of radius $r$ around $x$:
$\B_\infty(x, r)=\{y\in \nsq: |x-y|_\circ\leq r\}$, where $|x-y|_{\circ}$ is the circular distance on $\mathbb{Z}_n$.  \iflong{For a real number $a$ we write $|a|_+$ to denote the positive part of $a$, i.e. $|a|_+=a$ if $a\geq 0$ and $|a|_+=0$ otherwise.}

We will use the following 
\begin{definition}[Flat filter with $B$ buckets and sharpness $F$] \label{def:filterG}
    A sequence $G \in \R^n$ symmetric about zero with Fourier transform $\wh{G} \in \R^n$ is called a \emph{flat filter with $B$ buckets and sharpness $F$} if {\bf (1)} $G_j \in [0,1]$ for all $j \in [n]$; {\bf (2)} $G_j \ge 1 - \big(\frac{1}{4}\big)^{F-1}$ for all  $j \in [n]$ such that $|j| \le \frac{n}{2B}$; and {\bf (3)} $G_f \le \big(\frac{1}{4}\big)^{F-1} \big( \frac{n}{B|j|} \big)^{F-1}$ for all $j \in [n]$ such that $|j| \ge \frac{n}{B}$.
\end{definition}

We use a construction of such filters from~\cite{CKSZ17}:

\begin{lemma}[~\cite{CKSZ17}, Lemma~2.1] \label{lem:filter_properties}
    \emph{(Compactly supported flat filter with $B$ buckets and sharpness $F$)}
    Fix the integers $(n,B,F)$ with $n$ a power of two, $B < n$, and $F \ge 2$ an even number.  There exists an $(n,B,F)$-flat filter $G \in \R^n$, whose Fourier transform $\wh{G}$ is supported on a length-$O(FB)$ window centered at zero in time domain.
\end{lemma}

Note that most of the mass of the filter is concentrated in an interval of side $O(n/B)$, approximating the ``ideal'' filter (whose value would be equal to $1$ for entries within the square and equal to $0$ outside of it). 
Note that for each $i\in \nsq$ one has $G_{o_i(i)}^{-1}\leq 2$. We refer to the parameter $F$ as the {\em sharpness} of the filter. 
\iflong{Our hash functions are not pairwise independent, but possess a property that still makes hashing using our filters efficient:
\begin{lemma}[Lemma~3.2 in~\cite{IK14a}]\label{lemma:limitedindependence}
Let $i, j\in \nsq$. Let $\sigma$ be uniformly random odd number between $1$ and $n$. Then for all $t\geq 0$ one has $
\Pr[|\sigma(i-j)|_\circ \leq t] \leq 2(2t/n)$.
\end{lemma}}
\ifshort{Pseudorandom spectrum permutations (defined above) give us the ability to `hash' the elements of the input signal into a number of buckets (denoted by $B$). We formalize this using the notion of a {\em hashing}. A hashing is a tuple consisting of a pseudorandom spectrum permutation $\pi$, target number of buckets $B$ and a sharpness parameter $F$ of our filter, denoted by $H=(\pi, B, F)$. Formally, $H$ is a function that maps a signal $x$ to $B$ signals, each corresponding to a hash bucket, allowing us to solve the $k$-sparse recovery problem on input $x$ by reducing it to $1$-sparse recovery problems on the bucketed signals. We give the formal definition below.}

\subsection{Measurements of the signal,  notation for estimation error and basic bounds}
\iflong{Pseudorandom spectrum permutations combined with a filter $G$ give us the ability to `hash' the elements of the input signal into a number of buckets (denoted by $B$). We formalize this using the notion of a {\em hashing}. A hashing is a tuple consisting of a pseudorandom spectrum permutation $\pi$, target number of buckets $B$ and a sharpness parameter $F$ of our filter, denoted by $H=(\pi, B, F)$. Formally, $H$ is a function that maps a signal $x$ to $B$ signals, each corresponding to a hash bucket, allowing us to solve the $k$-sparse recovery problem on input $x$ by reducing it to $1$-sparse recovery problems on the bucketed signals. We give the formal definition below.}

\begin{definition}[Hashing $H=(\pi, B, F)$]
For a permutation $\pi=(\sigma, q)$, parameters $B>1$ and $F$,  a {\em hashing} $H:=(\pi, B, F)$ is a function mapping a signal $x\in \C^n$ to $B$ signals $H(x)=(u_s)_{s\in [B]}$, where $u_s\in \C^n$ for each $s\in [B]$, such that for each $i\in \nsq$
$ u_{s, i}  = \sum_{j\in \nsq} G_{\pi(j)-(n/B)\cdot s}x_j \omega^{i \sigma j}\in \C$,
where $G$ is a filter with $B$ buckets and sharpness $F$ constructed in Lemma~\ref{lem:filter_properties}.
\end{definition}

For a hashing $H=(\pi, B, F), \pi=(\sigma, q)$ we sometimes write $P_{H, a}, a\in \nsq$ to denote $P_{\sigma, a, q}$. %We will consider hashings of the input signal $x$, as well as the residual signal $x-\chi$, where 

\begin{definition}[Measurement $m=m(x, H, a)$]\label{def:measurements}
For a signal $x\in \C^n$, a hashing $H=(\pi, B, F)$ and a parameter $a\in \nsq$, a {\em measurement} $m=m(x, H, a)\in \C^B$ is the $B$-dimensional complex valued vector of evaluations of a hashing $H(x)$ at a point $a\in \nsq$, i.e. for $s\in [B]$
$m_s  = \sum_{j\in \nsq} G_{\pi(j)-(n/B)\cdot s}x_j \omega^{a \sigma j}$,
where $G$ is a filter with $B$ buckets and sharpness $F$ constructed in Lemma~\ref{lem:filter_properties}.
\end{definition}

We access the signal $x$ in Fourier domain via the function $\Call{HashToBins}{\hat x, \chi, (H, a)}$, which evaluates the hashing $H$ of residual signal $x-\chi$ at point $a\in \nsq$, i.e. computes the measurement $m(x, H, a)$ (the computation is done with polynomial precision). We will use the following lemma, which is rather standard (the proof is given in Appendix~\ref{sec:hash2bins} for completeness):
\begin{lemma}\label{l:hashtobins}
  \Call{HashToBins}{$\wh{x}, \chi, (H, a)$}, where $H=(\pi, B, F)$, computes
  $u\in \C^B$ such that for any $i \in [n]$, $u_{h(i)} = \Delta_{h(i)} + \sum_j G_{o_i(j)}(x - \chi)_j \omega^{a\sigma j}$,  where $G$ is the filter defined in section~\ref{sec:prelim}, and for all $i\in [n]$ we have that $\Delta_{h(i)}^2 \leq \norm{2}{\chi}^2
\cdot n^{-c}$ is a negligible error term (and $c>0$ is an absolute constant that governs the precision that semi-equispaced FFT, i.e. Lemma~\ref{lem:semi_equi_std}, is invoked with).  It takes $O(B \fc)$
  samples, and $O(F\cdot B\log B+\norm{0}{\chi} \log n)$ time.
\end{lemma}

\if 0
The following lemma is a straightforward modification of Lemma~9.2 in~\cite{K16}:

\begin{lemma}[Lemma~9.2 in~\cite{K16}]\label{l:hashtobins}
  \Call{HashToBins}{$\wh{x}, \chi, (H, a)$} computes
  $u$ such that for any $i \in [n]$, $u_{h(i)} = \Delta_{h(i)} + \sum_j G_{o_i(j)}(x - \chi)_j \omega^{a\sigma j}$
  where $G$ is the filter defined in Lemma~\ref{lm:filter-prop}, and $|\Delta_{h(i)}| \leq n^{-\Omega(c)}$ is a negligible error term, where $c>0$ is a sufficiently large constant.  It takes $O(B \fc)$
  samples, and $O(F B \log B+\norm{0}{\chi} \log n)$ time.\xxx{Check or merge with with next lemma}
\end{lemma}
\fi

%% file: l1-part1.tex
%!TEX root = ./ft-hd-wrap.tex 

\iflong{

We now introduce relevant notation for bounding the error induced by our measurements in locating or estimating an element $i\in [n]$. For a hashing $H=(\pi, B, F)$ and an evaluation point $z\in [n]$, we have by Definition~\ref{def:measurements}
  \[
  m_{h(i)}(x, H, z)  = \sum_{j\in \nsq} G_{o_i(j)}x_j \omega^{z \sigma j},
  \]
where the filter $G_{o_i(j)}$ is the filter corresponding to hashing $H$ (note that $o_i(j)$ implicitly depends on $\pi$). In particular, one has:
\begin{equation*}
\begin{split}
G_{o_i(i)}^{-1}m_{h(i)} \omega^{-z \sigma i} = x_i + \underbrace{G_{o_i(i)}^{-1}\sum_{j\in \nsq\setminus \{i\}} G_{o_i(j)}x_j \omega^{z \sigma (j-i)}}_{\text{noise term}}
\end{split}
\end{equation*}
%In~\cite{K16} the residual signal $x$ was represented as a sum of three terms: $x=(x-\chi)_S-\chi_{\nsq\setminus S}+x_{\nsq\setminus S}$, but as our algorithm will ensure that $\chi_{\nsq\setminus S}\equiv 0$, we will write $x=x_S+x_{\nsq\setminus S}$.
% The first term is the residual signal coming from the `heavy' elements in $S$, the second corresponds to the tail of the signal. 

A common idea underlying our analysis of estimation and recovery is to split the estimation/recovery error induced on an element $i$ into the contribution from the carefully defined `head' of the signal and the contribution from the `tail'. The `head' of the signal is denoted by a set $S\subseteq [n]$ throughout the paper.  For each $i\in [n]$ we write
\begin{equation}\label{eq:uexp1}
\begin{split}
G_{o_i(i)}^{-1}m_{h(i)}\omega^{-z \sigma i}&=x_i+\underbrace{G_{o_i(i)}^{-1}\cdot \sum_{j\in S\setminus \{i\}} G_{o_i(j)}x_j \omega^{z \sigma (j-i)}}_{\text{noise from `heavy' elements}}+\underbrace{G_{o_i(i)}^{-1}\cdot \sum_{j\in \nsq\setminus (S\cup \{i\})} G_{o_i(j)}x_j \omega^{z \sigma (j-i)}}_{\text{`tail' noise}}\\
\end{split}
\end{equation}

We now define special notation for the two noise terms in~\eqref{eq:uexp1}. These two noise terms will be handled very differently in our analysis.

\paragraph{Noise from heavy hitters.} The first term in \eqref{eq:uexp1} corresponds to noise from $x_{S\setminus \{i\}}$, i.e. noise from `head' of the signal. For every $i\in S$, hashing $H$  we let 

\begin{equation}\label{eq:eh-pi}
e^{head}_{i}(H, x):=G_{o_i(i)}^{-1}\cdot \sum_{j\in S\setminus \{i\}} G_{o_i(j)} |x_j|.\\
\end{equation}

\begin{remark}\label{rm:dependenceonS}
Note that $e^{head}$ depends implicitly on the set $S$. We do not make this dependence explicit to avoid complicated notation, but state which set $S$ the quantity $e^{head}$ is defined with respect to whenever this quantity is used. We also note that $e^{head}$ provides a bound on the error induced by the tail of the signal that (for a random hashing) depends on the $\ell_1$ norm of the head of the signal.
\end{remark}

We thus get that $e^{head}_i(H, x)$ upper bounds the absolute value of the first error term in~\eqref{eq:uexp1} {\em for every value of evaluation point $z$} (note that $e^{head}(H, x)$ only depends on the hashing $H$ and $x$). Note that $G\geq 0$ by Lemma~\ref{lem:filter_properties} and Definition~\ref{def:filterG} as long as $F$ is even, which is the setting that we are in. We will often use several hashings to estimate or locate an element $i$. It is thus convenient to define, for a sequence of hashings $H_1,\ldots, H_r$
\begin{equation}\label{eq:eh}
e^{head}_{i}(\{H_r\}, x):=\quant^{1/5}_r e^{head}_i(H_r, x),
\end{equation}
where for a list of reals $u_1,\ldots, u_s$ and a number $f\in (0, 1)$ we let $\quant^{f}(u_1,\ldots, u_s)$ denote the $\lceil f \cdot s\rceil$-th largest element of $u_1,\ldots, u_s$.   

\paragraph{Tail noise.} 
To capture the second term in \eqref{eq:uexp1} (corresponding to tail noise), we define, for any $i\in [n], z\in \nsq$, permutation $\pi=(\sigma, q)$ and hashing $H=(\pi, B, F)$
\begin{equation}\label{eq:et-pi}
e^{tail}_i(H, z, x):=\left|G_{o_i(i)}^{-1}\cdot \sum_{j\in \nsq\setminus (S\cup \{i\})} G_{o_i(j)}x_j \omega^{z \sigma (j-i)}\right|.
\end{equation}

\begin{remark}
Note that $e^{tail}$ depends implicitly on the set $S$. We do not make this dependence explicit to avoid complicated notation, but state which set $S$ the quantity $e^{tail}$ is defined with respect to whenever this quantity is used. We also note that $e^{tail}$ provides a bound on the error induced by the tail of the signal that (for a random hashing and a random evaluation point $z$) depends on the $\ell_2$ norm of the tail (as opposed to the $\ell_1$ norm used by $e^{head}$).
\end{remark}

With this definition in place $e^{tail}_i(H, z, x)$ upper bounds the absolute value of second term in~\eqref{eq:uexp1}. We will sometimes use several hashings and values $z$ to obtain better estimates. For a sequence $\{(H_r, z_r)\}_{r=1}^{r_{max}}$ for some $r_{max}\geq 1$ we let 
\begin{equation}\label{eq:et-pi-quant}
e^{tail}_i(\{H_r, z_r\}, x):=\quant^{1/5}_r \left|G_{o_i(i)}^{-1}\cdot \sum_{j\in \nsq\setminus (S\cup \{i\})} G_{o_i(j)}x_j \omega^{z \sigma (j-i)}\right|,
\end{equation}
where $o_i(j)$ on the rhs implicitly depends on the hashing $H$. 

The definitions above are sufficient for our analysis of the signal estimation procedure in Section~\ref{sec:est}. The analysis of the sparse recovery procedure requires several further specialized definitions, which are presented in Appendix~\ref{sec:location}. With the definitions above we can state the following simple guarantees on the performance of a basic estimation procedure that is the main building block of our analysis of the more powerful \textsc{Estimate} primitive in~Section~\ref{sec:est}. 

\begin{lemma}[Bounds on estimation quality for Algorithm~\ref{alg:est}]\label{lm:estimate-l1l2}
For every $x, \chi\in \C^n$,  every $L\subseteq \nsq$, every set $S\subseteq \nsq$ the following conditions hold for functions $e^{head}$ and $e^{tail}$ defined with respect to $S$ (see ~\eqref{eq:eh-pi} and~\eqref{eq:et-pi}). If $r_{max}$ is larger than an absolute constant, then for every sequence $H_r=(\pi_r, B, F), r=1,\ldots, r_{max}$ of hashings and every sequence $a_1,\ldots, r_{max}$ of evaluation points the output $w$ of 
$$
\Call{EstimateValues}{\chi, L, \{(H_r, a_r, m(x, H_r, a_r))\}_{r=1}^{r_{max}}}
$$
 satisfies, for each $i\in L$  
$$
|w_i- (x-\chi)_i|\leq 2\cdot\quant^{1/5}_r e^{head}_i(H_r, x-\chi)+ 2\cdot\quant^{1/5}_r e^{tail}_i(H_r, a_r, x-\chi)+n^{-\Omega(c)},
$$
where $c\geq 2$ is an absolute constant that governs the precision of our approximate semi-equispaced FFT computations (see \textsc{HashToBins}, Lemma~\ref{l:hashtobins}). The sample complexity is bounded by $O(F B r_{max})$, and the runtime by $O((F\cdot B\cdot \log n+||\chi||_0\log n+|L|)\cdot r_{max})$.
\end{lemma}
 The proof of the lemma is given in Appendix~\ref{sec:app}. The proof is rather standard modulo our definitions of $e^{head}$ and $e^{tail}$, as well as the fact that the statement of the lemma is entirely deterministic. We will later apply this lemma to random hashings and evaluation points, but the deterministic nature of the claim will be crucial in analyzing measurement reuse.

As both our \textsc{Estimate} and \textsc{SparseFFT} algorithms (Section~\ref{sec:est} and Section~\ref{sec:sublinear}) respectively iteratively update the signal, we will need to analyse the performance of \textsc{EstimateValues} on various residual signals derived from the original input signal $x$.  The notion of a {\em majorant} is central to this part of our analysis:
\begin{definition}[Majorant]\label{def:majorant}
For any $S\subseteq \nsq$ and any $x, y\in \C^n$ we say that $y$ is an majorant for $x$ with respect to $S$ if $|x_i|\leq |y_i|$ for all $i\in S$.
\end{definition}

With this definition and definition of $e^{head}$ above the following crucial lemma follows immediately:
\begin{lemma}\label{lm:mon}
For every hashing $H$, every set $S\subseteq \nsq$ one has for every pair $x, y\in \C^n$ that if $x\prec_S y$, then for every $i\in \nsq$
$e^{head}_i(H, x)\leq  e^{head}_i(H, y)$.
\end{lemma}
\begin{proof}

Recall that by~\eqref{eq:eh-pi} one has 
$e^{head}_{i}(H, x)=G_{o_i(i)}^{-1}\cdot \sum_{j\in S\setminus \{i\}} G_{o_i(j)} |x_j|$. Since $G\geq 0$ by Definition~\ref{def:filterG} and Lemma~\ref{lem:filter_properties}, we have by using $x\prec_S y$ that 
$e^{head}_{i}(H, x)=G_{o_i(i)}^{-1}\cdot \sum_{j\in S\setminus \{i\}} G_{o_i(j)} |x_j|\leq G_{o_i(i)}^{-1}\cdot \sum_{j\in S\setminus \{i\}} G_{o_i(j)} |y_j|=e^{head}_{i}(H, y)$ as required.

\end{proof}

}

%% file: l1.tex
%!TEX root = ./ft-hd-wrap.tex 
\section{Isolating partitions}\label{sec:l1}

In this section we prove the main lemmas that allow us to reason about performance of the SNR reduction process in our algorithms (Algorithm~\ref{alg:estimate-efficient} and Algorithm~\ref{alg:main-sublinear}). Both algorithms perform SNR reduction (see lines~16 to~22 in Algorithm~\ref{alg:estimate-efficient} and lines~25-36 in Algorithm~\ref{alg:main-sublinear}) using two loops. The first loop (over $r$), controls the $\ell_1$ norm of the signal, with the (upper bound on the) norm being reduced by a factor of $4$ in each iteration. This reduction is achieved via a sequence of calls to \textsc{EstimateValues} (in \textsc{Estimate}, Algorithm~\ref{alg:estimate-efficient}) or \textsc{LocateSignal} (in Section~\ref{sec:location}) using a separate collection of  hashings for each $t=1,\ldots, T$. Our correctness analysis for this process proceeds by showing that, with high constant probability over the choice of the hashings $\{\{H_{t, s}\}_{s=1}^{R_t}\}_{t=1}^T$ any set $S$ of size $\approx k$ the hashings induce a partition of $S$ into at most $T$ sets $S_1\cup\ldots\cup S_T$ such that hashings used in the $t$-th round allow the algorithm to make progress on elements in $S_t$. The main result of this section is a formalization of this claim, achieved by  Lemma~\ref{lm:l1b-head}.

\begin{lemma}\label{lm:l1b-head}
For every integer $k\geq 1$, every $S\subseteq \nsq, |S|\leq k$,  every $\delta\in (0, 1/2)$, if the parameters $B_t, R_t$ are selected to satisfy {\bf (p1)} $R_t=C_1\cdot 2^t$ and {\bf (p2)} $B_t\geq C_2\cdot k/R_t^2$ for every $t\in [0:T]$, where $C_1$ is a sufficiently large constant and $C_2$ is sufficiently large as a function of $C_1$ and $\delta$, then the following conditions hold.

For every collection of hashings $\{\{H_{t, s}\}_{s=1}^{R_t}\}_{t=1}^T$, if the filters used in hashings $H_{t, s}$ are at $F$-sharp for even $F\geq 6$ for every $t\in [1:T]$, and $S$ admits a $\delta$-isolating partition (as per Definition~\ref{def:isolating-partition}) $S=S_1\cup\ldots\cup S_T$ with respect to $\{H_{t, s}\}$, then for every $x, \chi\in \C^n, x'=x-\chi$, 
for every $t\in [1:T]$  one has 
$||e^{head}_{S_t}(\{H_{t, s}\}_{s\in [1:R_t]}, x')||_1\leq 20 R_t^{-\delta}||x'_S||_1$.
\end{lemma}

\begin{remark}\label{rm:overview}
Note that the result of Lemma~\ref{lm:l1b-head} implies that the cumulative error induced by the entire set $S$ of `heavy' coefficients on $S_t$ is only a $\approx R_t^{-\delta}$ fraction of the $\ell_1$ norm of $x'_S$, despite the fact that when estimating $S_t$ we hash into $B_t\ll k$ buckets, and in general each bucket will contain many elements of $S$. Furthermore, if we choose $R_t$ to increase fast enough so that $\sum_{t\geq 1} R_t^{-\delta}\ll 1$, we get that the cumulative contribution of elements in $S$ to estimation/location error over all $t\geq 1$ is less than $1$, meaning that errors do not accumulate much. This is exactly what we achieve by setting $R_t=C_1 2^t$ for a large constant $C_1>1$ -- see proof of Lemma~\ref{lm:main-estimate-snr-reduction} in Section~\ref{sec:est}.
\end{remark}

In the rest of the section we first introduce relevant notation and in particular define the central notion of an {\em isolating partition} of the set $S$ of head elements (in section~\ref{sec:notation-isolating}), then prove that any fixed set $S$ of size about $k$ admits an isolating partition with at least high constant probability over the choice of hashings $\{\{H_{t, s}\}_{s=1}^{R_t}\}_{t=1}^T$ (section~\ref{sec:partitions-exist}), and show how to construct the partition efficiently (Lemma~\ref{lm:construct-partition-runtime}) if the set $S$ is given explicitly (used in Algorithm~\ref{alg:estimate-efficient}, line~13). Using this result we then give a proof of Lemma~\ref{lm:l1b-head}.

\subsection{Main definitions}\label{sec:notation-isolating}
Let $S$ be any subset of $\nsq$ (we will later instantiate $S$ to the set of `large' elements of $x$).  We now define a decomposition of the set $S$ into $T=\frac1{1-\delta}\log_2 \log (k+1)+O(1)$ disjoint sets $S_1,S_2,\ldots, S_T$ with respect to a sequence $1\leq R_0\leq R_1 \leq R_2\leq \ldots\leq R_T$ and hashings $\{\{H_{t, s}\}_{s=1}^{R_t}\}_{t=1}^T$.  We start with several auxiliary definitions.

\begin{definition}[$t$-Collision]
We say that an element $a\in [n]$ participates in a $t$-collision with another element $b\in [n]$ under hashing $H=(\pi, B, F)$ if $a$ hashes within at most $t$ buckets of $b$ under  $H$, i.e. if $|\pi(a)-\pi(b)|\leq \frac{n}{B}(t-1)$.
\end{definition}

\begin{definition}[$\delta$-bad element]\label{def:bad}
For $\delta\in (0, 1)$, for each $t\in [1:T]$, sequence $1\leq R_0\leq R_1\leq\ldots\leq R_T$, where $T\geq 1$ is an integer, 
we say that an element $a$ of $S$ is {\em $\delta$-bad for $S_t$} with respect to a partition $S=S_1\cup S_2\cup\ldots\cup S_T$ and hashings $\{\{H_{t, s}\}_{s=1}^{R_t}\}_{t=1}^T$ if $a$ participates in an $R_t$-collision with at least one element of $S_t$ under more than a $R_t^{-\delta}$ fraction of hashings $H_{t, 1},\ldots, H_{t, R_t}$.
 \end{definition}

 \begin{definition}[$\lambda$-crowded element]\label{def:lambda-crowded}
For a hashing $H=(\pi, B, F)$ and a real number $\lambda\in (0, 1)$, an element $a\in \nsq$ is $\lambda$-crowded at scale $q\geq 0$ by a set $Q\subseteq \nsq$ if 
$\left|\mathbb{B}\left(\pi(a), \frac{n}{B}\cdot 2^q\right) \cap \pi(Q\setminus \{a\})\right|\geq \lambda 2^{2q}$. We say that an element $a$ is simply $\lambda$-crowded if it is $\lambda$-crowded at least at one scale $q\geq 0$.
\end{definition}
\begin{remark}
The intuition for the definition above is that if the permutation $\pi$ was pairwise independent,  then for every $a\in \nsq$ the expectation of $\left|\mathbb{B}\left(\pi(a), \frac{n}{B}\cdot 2^q\right) \cap \pi(Q\setminus \{a\})\right|$ would be about $2^{q}$ if $|Q|\leq B$, i.e. about the number of buckets that fall into the interval. We say that an element is $\lambda$-crowded when the number of elements in its vicinity exceeds $\lambda 2^{2q}$ for at least one scale $q$, i.e. exceeds expectation by a $\lambda 2^{q}$ factor. Our choice of $\lambda=R_t^{-3}$ serves the purpose of enforcing that no element of $S_t$ is hashed too close to another element of $S_t$, and no (large) neighborhood of any element of $S_t$ it too crowded. These two parameter regimes have somewhat distinct applications in the proof of Lemma~\ref{lm:l1b-head} -- see footnotes 3 and 4 in the proof of the lemma on pages~13 and~14 respectively.
\end{remark}
 
 \begin{definition}[$\delta$-isolating partition]\label{def:isolating-partition}
For $\delta\in (0, 1)$, for any $k\geq 1$ and any $S\subseteq \nsq, |S|\leq k$, a partition $S=S_1\cup S_2\cup \ldots \cup S_T$ of $S\subseteq [n]$ into disjoint subsets is {\em $\delta$-isolating} with respect to a sequence of integers $1\leq R_0\leq R_1\leq \ldots\leq R_T$ and hashings $\{H_{t, s}\}_{s=1}^{R_t}$ for $t=[1:T]$ if the following conditions are satisfied for each $t\in [1:T]$:
\begin{description}
\item[(1)] $|S_t|\leq k\cdot \frac{R_0}{R_{t-1}}2^{-2^{(1-\delta)\cdot (t-1)}+1}$ (the sizes of $S_t$'s decay doubly exponentially);
\item[(2)] no element of $S_t$ is $R_t^{-3}$-crowded by $S_t$ under any of $\{H_{t, s}\}_{s=1}^{R_t}$ (elements of $S_t$ are rather uniformly spread under hashings $H_{t, s}$);
\item[(3)] no element of $S_t$ $R_t$-collides with an element of $S$ that is $\delta$-bad for $S_t$ under any of $\{H_{t, s}\}_{s=1}^{R_t}$ (collisions between $S$ and $S_t$ are rare).
\end{description}
 \end{definition}
 
\iflong{ Note that the bound on the size of $S_1$ is at most $k$, which will be trivial for our instantiation of $S$. Nontrivial decay of the $S_t$ starts at $t=2$.}  Also note that property {\bf (3)} is exactly why we call our approach `isolation on average': as the proof of Lemma~\ref{lm:l1b-head} below shows, the fact that no element of $S$ that is $\delta$-bad for $S_t$ collides with $S_t$ under any of the hashings implies that every element of $S$ has a limited contribution to estimation error, and errors do not propagate.

\subsection{Existence of isolating partitions of $S$}\label{sec:partitions-exist}

In this section we prove that any set $S$ of size at most $k$ admits an isolating partition with respect to a random set of hashings as in Algorithm~\ref{alg:main-sublinear} with at least high constant probability. We prove this claim by giving an algorithm that we show constructs such a partition successfully with high constant probability. The algorithm is Algorithm~\ref{alg:partition}, presented below.  We prove that Algorithm~\ref{alg:partition} terminates correctly in Lemma~\ref{lm:isolating-partition-construction} below, assuming that the number of hashings at each step $t$ and their parameters are chosen appropriately. 

If the set $S$ is known explicitly, the partition can be constructed efficiently by running Algorithm~\ref{alg:partition} -- the details are given in Lemma~\ref{lm:construct-partition-runtime}. An efficient construction is needed in our sample efficient primitive (see Algorithm~\ref{alg:estimate-efficient}). Our sample-optimal Sparse FFT algorithm is oblivious to the actual partition, as its task is to identify the set $S$. However, as we show in Section~\ref{sec:sublinear}, the existence of an isolating partition of $S$ is sufficient for the algorithm to work.

\begin{algorithm}
\caption{Construction of an isolating partition $\{S_j\}$}\label{alg:partition}
\begin{algorithmic}[1]
\Procedure{ConstructPartition}{\{$\{H_{t, s}\}_{s\in [1:R_t]}\}_{t=1}^T$}\Comment{Hashings sampled uniformly}
\State $S^1_1\gets S$, $t\gets 1$
\While{$S^t_t\neq \emptyset$} 
\State $\text{Bad}_t\gets \{\text{elements of~}S\text{~that are $\delta$-bad wrt~$S_t^t$ under $\{H_{t, s}\}_{s\in [1:R_t]}$}\}$
\State $U_t\gets \{\text{elements $a\in S^t_t$ that $R_t$-collide with $\text{Bad}_t$ under at least one of $\{H_{t, s}\}_{s\in [1:R_t]}$}\}$
\State $V_t\gets \{\text{elements $a\in S^t_t$ that are $R_t^{-3}$-crowded by $S^t_t$ under at least one of $\{H_{t, s}\}_{s\in [1:R_t]}$}\}$
\State Set $S^{t+1}_{t+1}\gets \text{Bad}_t\cup U_t\cup V_t$ and $S^{t+1}_j\gets S^t_j\setminus S^{t+1}_{t+1}$ for $j=1,\ldots, t$
\State $t\gets t+1$ 
\EndWhile
\State \textbf{return} the partition $\{S^t_j\}_{j=1}^t$
\EndProcedure 
\end{algorithmic}
\end{algorithm}

\iflong{

We now argue that Algorithm~\ref{alg:partition} constructs an isolating partition of any set $S\subseteq \nsq$ that satisifies $|S|\leq k$ with at least high constant probability:
\begin{lemma}\label{lm:isolating-partition-construction}
For every integer $k\geq 1$, every $S\subseteq \nsq, |S|\leq k$,  every $\delta\in (0, 1/2)$, if the parameters $B_t, R_t$ are selected to satisfy {\bf (p1)} $R_t=C_1\cdot 2^t$ and {\bf (p2)} $B_t\geq C_2\cdot k/R_t^2$ for every $t\in [0:T]$, where $C_1$ is a sufficiently large constant and $C_2$ is sufficiently large as a function of $C_1$ and $\delta$, then the following conditions hold.  

With probability at least $1-1/25$ over the choice of hashings $\{\{H_{t, s}\}_{s\in [1:R_t]}\}_{t=1}^T$  Algorithm~\ref{alg:partition} terminates in $T=\frac1{1-\delta}\log_2 \log (k+1)+O(1)$ steps. When the algorithm terminates, the output partition $\{S_j\}_{j=1}^T$  is isolating as per Definition~\ref{def:isolating-partition}.
\end{lemma}
\begin{proofsketch}
The proof consists of two parts: showing that once the algorithm terminates, the resulting partition is $\delta$-isolating, and then showing that the algorithm terminates with large constant probability as long as parameters are set appropriately (to satisfy {\bf p1} and {\bf p2}). The first part is rather direct from definitions, and is presented in Appendix~\ref{app:isolating-partition} together with the full proof of Lemma~\ref{lm:isolating-partition-construction}. We now outline the proof of the second part, i.e. that the algorithm terminates.

The crux of the proof consists of bounding the sizes of the sets $V_t, U_t, \text{Bad}_t$ obtained at time step $t=1,\ldots, T$.  It is easiest to start with the intuition for $V_t$, i.e. elements in $S_t^t$ that are crowded by other elements of $S_t^t$. An element $a$ is crowded by $S_t^t$ under permutation $\pi$ if there are too many elements of $S_t^t$ in a neighborhood of a certain size of $a$. The definition of being crowded (Definition~\ref{def:lambda-crowded} above) involves the notion of being crowded at a given scale, and one can see that an element $a$ is more likely to fail because of small scales as opposed to large scales. This means that for every $a\in S_t^t$ one has 
$$
\prob[a\in V^t]\approx \poly(R_t)\cdot k_t/B_t=\poly(R_t)\cdot k_t/B, \text{~~~~~(since $B_t=\Theta(B/R_t^2)$)}
$$
i.e., up to terms polynomial in $R_t$, the probability of being crowded is about the probability of $O(1)$-colliding with one other element of $S_t^t$ (see below for a formalization of this claim).  Since decay of the size of $S_t^t$ that we will exhibit is doubly exponential in $t$, factors polynomial in $R_t$ are negligible, as they are only singly exponential in $t$. Given the expression for $\prob[a\in V^t]$ above, we get that 
$$
\expect[|V_t|]=\sum_{a\in S_t^t} \poly(R_t)\cdot k_t/B=\poly(R_t)\cdot (k_t/B)^2\cdot B.
$$
This means that if $V_t$ were the only contribution to $S_{t+1}^{t+1}$, then, discounting the $\poly(R_t)$ terms, we would get the recurrence 
$k_{t+1}/B=(k_t/B)^2$, which implies that $k_t\approx B\cdot 2^{-2^t}$. The $\poly(R_t)$ terms do not affect the decay substantially, and in the actual proof below we get that $k_t\leq O(k\cdot 2^{-2^{(1-\delta) t}})$ for a small constant $\delta>0$.

The more interesting term is the contribution from $U_t$ and $\text{Bad}_t$. We first describe the intuition behind the asymptotic growth of $\text{Bad}_t$. For an element $a\in S$ we have 
$$
\prob[a\text{~~$R_t$-collides with~an element of $S_t^t$}]\approx k_t/B_t=\poly(R_t) \cdot k_t/B\ll 1/(10R_t),
$$
since $k_t\approx B\cdot 2^{-2^t}$ and $R_t$ are only singly exponential in $t$. Note that this means that the expected number of collisions {\bf over all $R_t$ hashings $H_{t, s}$, $s=1,\ldots, R_t$, is less than $1/10$!} At the same time recall that $a$ is $\delta$-bad (as per Definition~\ref{def:bad} above) if $a$ collides with at least one element of $S_t^t$ under at least $R_t^{1-\delta}$ hashings $\{H_{t, s}\}_{s=1}^{R_t}$. Since the hashings are independent, we have by Chernoff bounds that the probability of the latter event is bounded by about 
$e^{-R_t^{1-\delta}}=e^{-C_1^{1-\delta} 2^{(1-\delta)t}}$,
so we again get doubly exponential decay, but for a different reason this time: by our choice of parameters $R_t$ to grow singly exponentially, and Chernoff bounds show that the probability of getting at least $R_t^{1-\delta}$ collisions while the expected number of collisions is less than $1/10$ decays exponentially in $R_t^{1-\delta}$. This is exactly the point at which we say that most elements of $S$ are `isolated on average'. A formal version of this argument lets us argue that the size of $\text{Bad}_t$ is about $k e^{-C_1^{1-\delta} 2^{(1-\delta) t}}$. It remains to bound the size of $U_t$, i.e. elements that collide with a bad element in at least one of the hashings. Since the number of bad elements is doubly exponentially small and the number of hashings $R_t$ is only single exponential, a union  bound essentially shows that this quantity is small as well. Some care is needed in arguing this due to a dependency issue, but the intuition is the same as the one described for $\text{Bad}_t$ above. The formal proof is given in Appendix~\ref{app:isolating-partition}.
}
\end{proofsketch}

If the set $S$ is known explicitly, Algorithm~\ref{alg:partition} admits a simple efficient implementation (the proof of the lemma is given in Appendix~\ref{app:isolating-partition}):
\begin{lemma}\label{lm:construct-partition-runtime}
For any integer $k\geq 1$, any $S\subseteq \nsq, |S|\leq k$,  if the hashings $\{\{H_{t, s}\}_{s\in [1:R_t]}\}_{t=1}^T$ are such that the partition  $\{S_j\}_{j=1}^T$ defined by Algorithm~\ref{alg:partition} is isolating as per Definition~\ref{def:isolating-partition}, this partition can be constructed explicitly in time $O\left((\sum_{t=1}^T R_t)\cdot |S|\log |S|\right)$.
\end{lemma}

\iflong{
\subsection{Proof of main technical lemma (Lemma~\ref{lm:l1b-head})}

We now prove the main result of this section, Lemma~\ref{lm:l1b-head}. This lemma then forms the basis of our sample-efficient estimation primitive, as well as the location primitive. The proof crucially relies on the existence of an isolating partition of the set $S$ of `head elements'  (which is guaranteed by Lemma~\ref{lm:isolating-partition-construction} from the previous section with at least high constant probability). 

{\noindent {\bf Lemma~\ref{lm:l1b-head}} (restated) \em
For every integer $k\geq 1$, every $S\subseteq \nsq, |S|\leq k$,  every $\delta\in (0, 1/2)$, if the parameters $B_t, R_t$ are selected to satisfy {\bf (p1)} $R_t=C_1\cdot 2^t$ and {\bf (p2)} $B_t\geq C_2\cdot k/R_t^2$ for every $t\in [0:T]$, where $C_1$ is a sufficiently large constant and $C_2$ is sufficiently large as a function of $C_1$ and $\delta$, then the following conditions hold.

For every collection of hashings $\{\{H_{t, s}\}_{s=1}^{R_t}\}_{t=1}^T$, if the filters used in hashings $H_{t, s}$ are at $F$-sharp for even $F\geq 6$ for every $t\in [1:T]$, and $S$ admits a $\delta$-isolating partition (as per Definition~\ref{def:isolating-partition}) $S=S_1\cup\ldots\cup S_T$ with respect to $\{H_{t, s}\}$, then for every $x, \chi\in \C^n, x'=x-\chi$, 
for every $t\in [1:T]$  one has 
$||e^{head}_{S_t}(\{H_{t, s}\}_{s\in [1:R_t]}, x')||_1\leq 20 R_t^{-\delta}||x'_S||_1$.
}
\begin{proof}
Fix $t\in [1:T]$. For every $s\in [1:R_t]$ by \eqref{eq:eh-pi} for all $i\in S$ we have
\begin{equation}\label{eq:2398hfgxc}
\begin{split}
e^{head}_{i}(H_{t, s}, x')&= G_{o_i(i)}^{-1}\cdot \sum_{j\in S\setminus  \{i\}} G_{o_i(j)}\cdot |x'_j|\\
\end{split}
\end{equation}
where $o=o_{H_{t, s}}$ implicitly depends on the hashing $H$. We omit the dependence on $H$ when $H$ is fixed to simplify notation.
By summing~\eqref{eq:2398hfgxc} over $i\in S_t$ we get
\begin{equation}\label{eq:asdf111}
\begin{split}
e^{head}_{S_t}(H_{t, s})= \sum_{i\in S_t} G_{o_i(i)}^{-1}\cdot \sum_{j\in  S\setminus  \{i\}} G_{o_i(j)} \cdot |x'_j|&\leq 2\sum_{j\in S} |x'_j|\cdot \sum_{i\in S_t\setminus  \{j\}} G_{o_i(j)}\\
&=2\sum_{j\in S} |x'_j|\cdot D_j^s,
\end{split}
\end{equation}
where for all $j\in S$ and $s\in [1:R_t]$ we let $D^{s}_j:=\sum_{i\in S_t\setminus  \{j\}} G_{o_i(j)}$, and used the fact that $G_{o_i(i)}\geq 1/2$ by Definition~\ref{def:filterG} and assumption that $F\geq 6$.  Note that $D^s_j$ depends on $t$, but we omit $t$ to simplify notation. This will not cause confusion since $t$ is fixed for the entire proof. We now bound $D^s_j$ for $j\in S$. We have for $j\in S$
\begin{equation}\label{eq:sum}
\begin{split}
D^s_j&=\sum_{i\in S_t\setminus  \{j\}} G_{o_i(j)}\leq \left|\left\{i\in S_t\setminus  \{j\}: |\pi_{s, t}(i)-\pi_{s, t}(j)|_\circ<\frac{n}{B_t}\cdot R_t\right\}\right|+\sum_{\substack{i\in S_t\setminus  \{j\}: \\|\pi_{s, t}(i)-\pi_{s, t}(j)|_\circ\geq \frac{n}{B_t}\cdot R_t}} G_{o_i(j)}\\
&=:Z^s_j+X^s_j.
\end{split}
\end{equation}
We used the fact that $||G||_\infty\leq 1$ by Definition~\ref{def:filterG}, {\bf (1)} and assumption that $F$ is even, to go from the first line to the second.

 \paragraph{Bounding $Z^s_j$.} We start by showing that $Z^s_j\in \{0, 1\}$ for all $s, j$. We have by  property {\bf (2)} of an isolating partition (see Definition~\ref{def:isolating-partition}) that  no element of $S_t$ is $R_t^{-3}$-crowded with respect to $S_t$\footnote{Note that the assumption that no element of $S_t$ is $\lambda$-crowded with respect to $S_t$ for $\lambda=R_t^{-3}$ is used twice in the proof: first to show that $Z^s_j\in \{0, 1\}$, since no two elements of $S_t$ can hash too close to each other by the choice of $\lambda$, and then later to upper bound the number of neighbors in $S_t$ an element can have. The specific choice of $\lambda=R_t^{-3}$ is only used here, however: for the other application $\lambda=O(1)$ would have been sufficient.}. This in particular means that no $i\in S_t$ is $R_t^{-3}$-crowded at scale $q=1+\log_2 R_t$ by $S_t$, i.e.
$$
\left|\mathbb{B}_{\infty}\left(\pi_{s, t}(i), \frac{n}{B_t}\cdot (2R_t)\right) \cap \pi(S_t\setminus \{i\})\right|\leq R_t^{-3} 4\cdot 2^{2q}\leq 4R_t^{-3} R_t^2\leq 4/R_t<1
$$
for all $i\in S_t$, as long as $C_1>4$ (recall that $R_t=C_1 2^t$ and $C_1$ is larger than an absolute constant). For every $j\in S$ (as opposed to $S_t$) and every $a, b\in S_t$ we have  by triangle inequality
$|\pi_{s, t}(a)-\pi_{s, t}(b)|_\circ\leq |\pi_{s, t}(a)-\pi_{s, t}(j)|_\circ+|\pi_{s, t}(b)-\pi_{s, t}(j)|_\circ$.
This means that if for some $j\in S$ 
$$
Z^s_j=\left|\left\{i\in S_t\setminus  \{j\}: |\pi_{s, t}(i)-\pi_{s, t}(j)|_\circ<\frac{n}{B_t}\cdot R_t\right\}\right|> 1,
$$
we have  $|\pi_{s, t}(a)-\pi_{s, t}(b)|_\circ\leq |\pi_{s, t}(a)-\pi_{s, t}(j)|_\circ+|\pi_{s, t}(b)-\pi_{s, t}(j)|_\circ<\frac{n}{B_t}\cdot 2R_t$, a contradiction. Thus,
$Z^s_j\in \{0, 1\}$ for all $j\in S$, and  by property {\bf (3)} of an isolating partition we now conclude that
\begin{equation}\label{eq:z-sum}
\sum_{s=1}^{R_t} Z^s_j\leq R_t^{1-\delta}
\end{equation}
for all $j\in S$.

\paragraph{Bounding $X^s_j$.}We now turn to bounding $X^s_j$ (i.e. second term on the rhs of ~\eqref{eq:sum}). Let 
\begin{equation}\label{eq:njq-def}
N(j, q):=\{i\in S_t\setminus \{j\}\text{~s.t.~} |\pi_{s, t}(j)-\pi_{s, t}(i)|_\circ\leq \frac{n}{B_t}(2^{q+1}-1)\text{~and~} |\pi_{s, t}(j)-\pi_{s, t}(i)|_\circ\geq \frac{n}{B_t}\cdot R_t\}
\end{equation}
for convenience. We have
\begin{equation}\label{eq:2pogh}
\begin{split}
X^s_j=&\sum_{i\in S_t\setminus  \{j\}: |\pi_{s, t}(i)-\pi_{s, t}(j)|_\circ\geq \frac{n}{B_t}R_t} G_{o_i(j)}\leq \sum_{q\geq \log_2 R_t} \sum_{\substack{i\in S\setminus \{j\}\text{~s.t.~}\\ |\pi_{s, t}(j)-\pi_{s, t}(i)|_\circ\in \frac{n}{B_t}[2^q, 2^{q+1}-1)}} G_{o_i(j)}\\
&\leq  \sum_{q\geq \log_2 R_t} \left|N(j, q)\right|\cdot \max_{|\pi_{s, t}(j)-\pi_{s, t}(i)|_\circ\geq \frac{n}{B_t}2^q} G_{o_i(j)}.\\
\end{split}
\end{equation}
We now upper bound both terms in the last line of the equation above. 

To bound the second term it suffices to note that for every $q\geq 0$
 every $i, j$ with $|\pi_{t, s}(j)-\pi_{t, s}(i)|_\circ\geq \frac{n}{B_t}2^q$ satisfy
$|o_i(j)|_\circ= |\pi_{t, s}(j)-\frac{n}{B_t}\cdot h_{t, s}(i)|_\circ\geq |\pi_{t, s}(j)-\pi_{t, s}(i)|_\circ-|\pi_{t, s}(i)-\frac{n}{B_t}h_{t, s}(i)|_\circ\geq \frac{n}{B_t}\cdot (2^q-1)$.
Using this bound together with Definition~\ref{def:filterG}, {\bf (3)}, and assumption that the filter $G$ is at least $6$-sharp we have for $q\geq \log_2 R_t\geq 1$
\begin{equation}\label{eq:cell-ub}
\max_{|\pi(j)-\pi(i)|_\circ\geq \frac{n}{B_t}\cdot 2^q} G_{o_i(j)}\leq \left(\frac{1}{4(2^q-1)}\right)^{5}\leq (2^q)^{-5}\leq 2^{-5q}.
\end{equation}
Note that we also used the assumption that $q\geq \log_2 R_t\geq 1$ to lower bound $4(2^q-1)$ by $2^q$.

We now upper bound the first term on the last line of~\eqref{eq:2pogh}, i.e. the size of $N(j, q)$ for every $j\in S$. Let  $i^*:=\text{argmin}_{i\in S_t} |\pi_{s, t}(j)-\pi_{s, t}(i)|_\circ$ denote a point in $S_t$ that is mapped closest to $j$. Let $L^*:=|\pi_{s, t}(j)-\pi_{s, t}(i^*)|_\circ$ denote the distance to this point, and let $q^*$ be the smallest such that $(n/B_t)2^{q^*}\geq L^*$.  By triangle inequality we have for all $i\in \nsq$
\begin{equation}\label{eq:pi-triangle}
|\pi_{t, s}(i)-\pi_{t, s}(i^*)|_\circ\leq |\pi_{t, s}(j)-\pi_{t, s}(i)|_\circ+L^*,
\end{equation}
allowing us to upper bound the size of $N(j, q)$ (points not too far from $j$) using the fact that $S_t$ is not crowded (property {\bf (2)} of an isolating partition; see Definition~\ref{def:isolating-partition}).
Specifically, for every $q\geq 0$ we have, combining~\eqref{eq:pi-triangle} and~\eqref{eq:njq-def},
\begin{equation*}
\begin{split}
N(j, q)& \subseteq \{i\in S_t\setminus \{j\}\text{~s.t.~} |\pi_{s, t}(i)-\pi_{s, t}(j)|_\circ\leq \frac{n}{B_t}(2^{q+1}-1)\}\text{~~~~~~~~~~~~~~~~~~(by~\eqref{eq:njq-def})}\\
&\subseteq \{i\in S_t\setminus \{j\}\text{~s.t.~} |\pi_{s, t}(i)-\pi_{s, t}(i^*)|_\circ\leq \frac{n}{B_t}(2^{q+1}+2^{q^*}-1)\} \text{~~~~~~~(by~\eqref{eq:pi-triangle})}
\end{split}
\end{equation*}
By property {\bf (2)} of an isolating partition (see Definition~\ref{def:isolating-partition}) we have for any $q\geq 0$ the number of $i\in S_t$ such that  $|\pi_{s, t}(i)-\pi_{s, t}(i^*)|_\circ\leq \frac{n}{B_t}(2^{q+1}+2^{q^*}-1)$ is bounded by $R_t^{-3} (2^{q+1}+2^{q^*}-1)^2+1$, where the $+1$ accounts for $i^*$ itself. Since we will only use the bound for $q\geq \log_2 R_t$, the $R_t^{-3}$ factor in first term will not be consequential\footnote{Note that this is the second time we are using the assumption that no element of $S_t$ is $\lambda$-crowded with respect to $S_t$, but in this case the choice of $\lambda=R_t^{-3}$ is not important, any constant, even $\lambda>1$, would have sufficed to this particular application.}, and we hence use the simpler form 
$R_t^{-3} (2^{q+1}+2^{q^*}-1)^2+1\leq  (2^{q+1}+2^{q^*}-1)^2+1\leq 2(2^{q+1}+2^{q^*}-1)^2$,
where we used the fact that $q\geq \log_2 R_t\geq 0$.
  We thus have for all $q\geq 0$
\begin{equation}\label{eq:n-bound}
|N(j, q)|\leq \left\lbrace
\begin{array}{cc}
0&\text{~if~}q<q^*\\
2(2^{q+1}+2^{q^*})^2&\text{~o.w.}
\end{array}
\right.
\end{equation}

Substituting ~\eqref{eq:cell-ub} and~\eqref{eq:n-bound} into~\eqref{eq:2pogh}, we get 
\begin{equation*}
\begin{split}
X^s_j&=\sum_{i\in S_t\setminus  \{j\}: |\pi_{s, t}(i)-\pi_{s, t}(j)|_\circ\geq \frac{n}{B_t}R_t} G_{o_i(j)}\leq  \sum_{q\geq \log_2 R_t} \max_{|\pi_{s, t}(j)-\pi_{s, t}(i)|_\circ\geq \frac{n}{B_t}2^q} G_{o_i(j)}\cdot \left|N(j, q)\right|\\
&\leq  \sum_{\substack{q\geq \log_2 R_t\\q\geq q^*}} 2^{-5q}\cdot (2\cdot (2^{q+1}+2^{q^*})^2)\leq  \sum_{\substack{q\geq \log_2 R_t\\q\geq q^*}} 2^{-5q}\cdot (32 \cdot 2^{2q})\\
&\leq  32 \cdot \sum_{\substack{q\geq \log_2 R_t\\q\geq q^*}} 2^{-3q}\text{~~~~(summing the geometric sum)}\\
&\leq  64\cdot R_t^{-3}\leq R_t^{-2}
\end{split}
\end{equation*}
as long as $C_1$ is larger than $64$ (since $R_t=C_1\cdot 2^t\geq C_1$ by assumption {\bf p1} of the lemma).  Substituting this bound into~\eqref{eq:sum}, we get $D^s_j\leq Z^s_j+X^s_j\leq Z^s_j+R_t^{-2}$,
which means that 
\begin{equation}\label{eq:d-final-bound}
\sum_{s=1}^{R_t} D^s_j\leq R_t^{1-\delta}+R_t^{-1}.
\end{equation}

To complete the argument, recall that by \eqref{eq:eh} one has $e^{head}_i(\{H_{t, s}\}, x')=\quant^{1/5}_{s\in [1:R_t]} e^{head}_i(H_{t, s}, x')$. This means that for each $i\in S_t$ there exist at least $(1/5)r_{max}$ values of $s\in [1:R_t]$ such that 
$e^{head}_i(H_{t, s}, x')>e^{head}_i$, and hence
\begin{equation}\label{eq:pingrg43g}
||e^{head}_{S_t}(\{H_{t, s}\}, x')||_1\leq \frac1{(1/5)R_t}\sum_{s=1}^{R_t} ||e^{head}_{S_t}(H_{t, s}, x')||_1.
\end{equation}

Substituting the bound from~\eqref{eq:pingrg43g} into ~\eqref{eq:asdf111}, we get 
\begin{equation*}
\begin{split}
||e^{head}_{S_t}(\{H_{t, s}\}, x')||_1\leq \frac1{(1/5)R_t}\sum_{s=1}^{R_t} ||e^{head}_{S_t}(H_{t, s}, x')||_1&\leq \frac2{(1/5)R_t} \sum_{s=1}^{R_t}\sum_{j\in S} |x'_j|\cdot D_j^s\\
&\leq \sum_{j\in S} |x'_j|\cdot \left(\frac2{(1/5)R_t} \sum_{s=1}^{R_t} D_j^s\right).\\
\end{split}
\end{equation*}

%Substituting this bound on $D_j$ given by~\eqref{eq:d-final-bound} into ~\eqref{eq:asdf111}, we get 

Substituting the above into~\eqref{eq:pingrg43g}, we get
\begin{equation}\label{eq:923hthgg}
\begin{split}
||e^{head}_{S_t}(\{H_{t, s}\}, x')||_1&\leq \sum_{j\in S} |x'_j|\cdot  \left(\frac2{(1/5)R_t}\sum_{s=1}^{R_t} D_j^s\right)\leq \sum_{j\in S} |x'_j|\cdot  \left(\frac2{(1/5)R_t}(R_t^{1-\delta}+R_t^{-1})\right)\\
&\leq \sum_{j\in S} |x'_j|\cdot  \left(\frac{20}{R_t} R_t^{1-\delta}\right)\text{~~~~~(since $R_t^{-1}\leq R_t^{1-\delta}$)}\\
&\leq 20 R_t^{-\delta} ||x'_S||_1.\\
\end{split}
\end{equation}

Substituting the bound from~\eqref{eq:923hthgg} into~\eqref{eq:pingrg43g}  we get $||e^{head}_{S_t}(\{H_{t, s}\}, x')||_1\leq 20 \cdot R_t^{-\delta} ||x'_S||_1$, as required.
\end{proof}

}

%% file: est.tex
%!TEX root = ./ft-hd-wrap.tex 
\section{Sample efficient estimation}\label{sec:est}
In this section we state our sample optimal estimation algorithm (Algorithm~\ref{alg:estimate-efficient}) and provide its analysis.

\subsection{Algorithm and overview of analysis}

 Our algorithm (Algorithm~\ref{alg:estimate-efficient}) contains three major components: it starts by taking measurements $m$ of the signal $x$ (accessing the signal in Fourier domain, i.e. accessing $\wh{x}$), then uses these measurements to perform a sequence of $\ell_1$ norm reduction steps, reducing $\ell_1$ norm of the residual signal on the target set $S$ of coefficients to about the noise level. Finally, a simple cleanup procedure is run to convert the $\ell_1$ norm bounds on the residual to $\ell_2/\ell_2$ guarantees of~\eqref{e:l2l2-est}.  In this section we will use the functions $e^{head}, e^{tail}$ (see Section~\ref{sec:prelim}) defined with respect to the set $S$. 

\noindent{\bf Measuring $\wh{x}$.} All measurements that the algorithm takes are taken in lines~6-12, and then line~31. The measurements in lines~6-12 are taken over $T=\frac1{1-\delta}\log_2\log (k+1)+O(1)$ rounds for small constant $\delta\in (0, 1/2)$, where in round $t=1,\ldots, T$ we are hashing the signal into $B_t\approx k/R_t^2$ buckets, where $R_t$ grows exponentially with $t$. For each $t$ we perform $R_t$ independent hashing experiments of this type. This matches the setup of Lemma~\ref{lm:l1b-head}, which is our main analysis tool (see proof of Lemma~\ref{lm:main-estimate-snr-reduction}).

\noindent{\bf $\ell_1$ norm reduction loop.} Once the samples have been taken, Algorithm~\ref{alg:estimate-efficient} proceeds to the $\ell_1$ norm reduction loop (lines~16-22). The objective of this loop is to reduce the $\ell_1$ norm of the residual signal on the target set $S$ of coefficients that we would like to estimate to about the noise level, namely to $O(||x_{[n]\setminus S}||_2\sqrt{k})$. The formal guarantees are provided by 

\begin{lemma}\label{lm:main-estimate-snr-reduction}
For every $\delta\in (0, 1/2)$, if $C_1, C_2$ (parameters in Algorithm~\ref{alg:estimate-efficient}) are sufficiently large constants, then the following conditions hold.

For every $x\in \C^n$, every integer $k\geq 1$, every $S\subseteq \nsq$, $|S|=k$, if $||x||_\infty\leq R^*\cdot ||x_{\nsq\setminus S}||_2/\sqrt{k}, R^*=n^{O(1)}$, the vector $\tilde \chi$ computed in line~23 of an invocation of \textsc{Estimate}$(\hat x, S, k, \e, R^*)$ (Algorithm~\ref{alg:estimate-efficient}) satisfies 
$$
||(x-\tilde \chi)_S||_1\leq O(||x_{\nsq\setminus S}||_2\cdot \sqrt{k})
$$ conditioned on an event $\E_{maj}$ that occurs with probability at least $1-2/25$.
\end{lemma}

\noindent{\bf Cleanup  phase and final result.} Once the $\ell_1$ norm of the residual on $S$ has been reduced to $O(||x_{[n]\setminus S}||_2\sqrt{k})$, we run the \textsc{EstimateValues} procedure once to convert $\ell_1$ norm bounds on the residual into $\ell_2/\ell_2$  guarantees~\eqref{e:l2l2-est}. This results in a proof of Theorem~\ref{thm:main-estimate}, restated below for convenience of the reader. The theorem establishes correctness of Algorithm~\ref{alg:estimate-efficient}, as well as its runtime and sample complexity bounds:

{\noindent {\bf Theorem~\ref{thm:main-estimate}} (Restated) \em 
For every $\e\in (1/n, 1), \delta\in (0, 1/2)$, $x\in \C^n$ and every integer $k\geq 1$, any $S\subseteq \nsq$, $|S|=k$, if $||x||_\infty\leq R^*\cdot ||x_{\nsq\setminus S}||_2/\sqrt{k}, R^*=n^{O(1)}$, an invocation of \textsc{Estimate}$(\hat x, S, k, \e, R^*)$ (Algorithm~\ref{alg:estimate-efficient}) returns $\chi^*\in \C^n$ such that 
$$
||(x-\chi^*)_S||_2^2\leq \e\cdot ||x_{\nsq\setminus S}||_2^2
$$
using $O_\delta(\frac1{\e}k)$ samples and $O_\delta(\frac1{\e}k \log^{3+\delta} n)$ time with at least $4/5$ success probability. 
}

\begin{algorithm}[h!]
\caption{\textsc{Estimate}($\hat x, S, k, \e, R^*$)}\label{alg:estimate-efficient} 
\begin{algorithmic}[1] 
\Procedure{Estimate}{$\hat x, S, k, \e, R^*$}\Comment{List $S$ of size $k$}
\State $T\gets \frac1{1-\delta}\log_2 \log (k+1)+O(1)$ for a small constant $\delta\in (0, 1/2)$
\State $R_t\gets C_1\cdot 2^t$ for $t\in [1:T]$ \Comment{$C_1>0$ sufficiently large absolute constant}
\State $B_t\gets C_2 \cdot k/R_t^2$ for $t\in [1:T]$ \Comment{$C_2>0$ sufficiently large absolute constant}
\State $G_t\gets$ filter with $B_t$ buckets and sharpness $\fc=8$.
\For {$t=1$ to $T$}\Comment{Take samples}
\For {$s=1$ to $R_t$}
\State Choose $\sigma\in \gl$, $q\in \nsq$ u.a.r., let $\pi_{t, s}\gets (\sigma, q)$, $H_{t, s}:=(\pi_{t, s}, B_t, \fc)$
\State Let $a_{t, s}\gets $ an element of $\nsq$ u.a.r.
\State $m(x, H_{t, s}, a_{t, s}) \gets \Call{HashToBins}{\hat x, 0, (H_{t, s}, a_{t, s})}$
\EndFor
\EndFor
\State Explicitly construct a $\delta$-isolating partition $S=S_1\cup S_2 \ldots \cup S_T$ \Comment{As per Lemma~\ref{lm:construct-partition-runtime}}
\State $\chi^{(0, 0)}\gets 0$
\State $r'\gets 0, t'\gets 0$
\For{$r = 0, 1, \dotsc, C\log_4 R^*$} \Comment{For any constant $C\geq 1$}
\For{$t=1$ to $T$}
\State $\chi' \gets \Call{EstimateValues}{\chi^{(r', t')}, S_t, \{(H_{t, s}, a_{t, s}, m(x, H_{t, s}, a_{t, s}))\}_{s=1}^{R_t}}$
\State $\chi^{(r, t)}\gets \chi^{(r', t')}+\chi'$ \Comment{$(r', t')$ are the previous indices}
\State $r'\gets r$, $t'\gets t$
\EndFor
\EndFor
\State $\tilde \chi \gets \chi^{(C\log_4 R^*, T)}$ \Comment{$\tilde \chi$ is the final residual computed by the loop}
\State $B\gets C_2 \cdot k/\e$
\State $G\gets$ filter with $B$ buckets and sharpness $\fc=8$.
\State $r_{max}\gets O(1)$ \Comment{A sufficiently large absolute constant}
\For {$r=1$ to $O(1)$}
\State Choose $\sigma_r\in \gl$, $q_r, a_r\in \nsq$ u.a.r., let $\pi_r\gets (\sigma_r, q_r)$, $H_r:=(\pi_r, B, \fc)$
\State $m(x, H_r, a_r) \gets \Call{HashToBins}{\hat x, \tilde \chi, H_r, a_r}$
\EndFor
\State $\chi'' \gets \Call{EstimateValues}{\tilde \chi, S, \{(H_r, a_r, m(x, H_r, a_r))\}_{r=1}^{r_{max}}}$
\State $\chi^*\gets \tilde \chi+\chi''$
\State \textbf{return} $\chi^*$
\EndProcedure 
\end{algorithmic}
\end{algorithm}

\subsection{Proof of Lemma~\ref{lm:main-estimate-snr-reduction}}\label{sec:lm41}
We now give 

\begin{proofof}{Lemma~\ref{lm:main-estimate-snr-reduction}}
Recall that in this section we use the quantities $e^{head}$ and $e^{tail}$ defined with respect to the set $S$.
We will also use an isolating partition of $S$, denoted by $S=S_1\cup S_2\cup \ldots\cup S_T$. We argue the existence of such a partition with high probability later.

We prove by induction on the pair $(r, t)$ that conditional on a high probability success event $\E_{maj}$ (defined below) the residual signals $x-\chi^{(r, t)}$ are {\em majorized} on $S$ (in the sense of Definition~\ref{def:majorant}) by a fixed sequence $y^{(r, t)}$ whose $\ell_1$ norm converges to $O(||x_{\nsq\setminus S}||_2\cdot \sqrt{k})$ after $O(\log R^*)$ iterations. Since we only update elements in $S$, this gives the result. We now give the details of the argument. In what follows we let $\mu^2:=||x_{\nsq\setminus S}||_2^2/k$ for convenience. Note, however, that Algorithm~\ref{alg:estimate-efficient} is oblivious to the value of $\mu$: we only need an upper bound on $\log R^*$.

We start by defining the majorizing sequence $y^{(r, t)}$. We first let $y^{(0, 0)}_i=R^* \mu$ for all $i\in S$ and $y^{(0, 0)}_i=x_i$ otherwise. Note that $y^{(0, 0)}$ trivially majorizes $x$ as $||x||_\infty\leq R^*\cdot \mu$ by assumption of the lemma. 
The construction of $y^{(r, t)}$ proceeds by induction on $(t, r)$. Given $y^{(r', t')}$, as per Algorithm~\ref{alg:estimate-efficient} the next signal to be defined is $y^{(r, t)}$ with $(r, t)=(r', t'+1)$ if $t<T$ and $(r, t)=(r'+1, 1)$ otherwise (as per lines 15-22 of Algorithm~\ref{alg:estimate-efficient}).  We now define the signal $y^{(r, t)}$ by letting for each $i\in [n]$ (recall that $S_t$ is the $t$-th set in an isolating partition $S=S_1\cup S_2\cup \ldots \cup S_T$)
\begin{equation}\label{eq:maj-def-estimation}
y^{(r, t)}_i:=\left\{\begin{array}{ll}
20 e^{head}_i(\{H_{t, s}\}_{s\in [1:R_t]}, y^{(r',t')})+20 e^{tail}_i(\{H_{t, s}, a_{t, s}\}_{s\in [1:R_t]}, x)+n^{-\Omega(c)}&\text{~~if~}i\in S_t\\
y^{(r', t')}_i&\text{~~o.w.}
\end{array}
\right.
\end{equation}
Here $n^{-\Omega(c)}$ corresponds to the (negligible) error term due to polynomial precision of our computations. Note that there are two contributions to $y^{(r, t)}$: one coming from the previous signal in the majorizing sequence, namely $y^{(r', t')}$, and the other coming from the tail of the signal $x$. 

We now prove by induction on $(t, r)$ that the loop in our estimation primitive reduces the $\ell_1$ norm of the residual to $O(\mu \cdot k)$ (recall that $\mu^2=||x_{\nsq\setminus S}||_2^2/k$).
Specifically, we prove that  there exists an event $\E_{maj}$ with $\prob_{\{\{H_{t, s}\}_{s\in [1:R_t]}\}\}_{t=1}^T}[\E_{maj}]\geq 1-2/25$ such that conditioned on $\E_{maj}$ the set $S$ admits an isolating partition $S=S_1\cup S_2\cup\ldots\cup S_T$ with respect to $\{\{H_{t, s}\}\}$, and for every $(r, t)\in ([0:+\infty)\times [1:T])\cup \{(0, 0)\}$ 
\begin{description}
\item[(A)] for all $q\in [1: t]$ one has $||y^{(r, t)}_{S_q}||_1\leq (R^*\cdot (1/4)^{r+1}\mu+2\mu) \cdot k\cdot (R_0/R_{q-1})^\delta$;
\item[(B)] for all $q\in [t+1: T]$ one has $||y^{(r, t)}_{S_q}||_1\leq (R^*\cdot (1/4)^r\mu +2\mu)\cdot k\cdot (R_0/R_{q-1})^\delta$;
\item[(C)] $||y^{(r, t)}_S||_1\leq (2/\delta)\cdot (R^*(1/4)^r\mu+2\mu)\cdot k$;
\item[(D)] $(x-\chi^{(r, t)})\prec_S y^{(r, t)}$ and $\supp \chi^{(r, t)}\subseteq S$.
\end{description}

First, note that the set $S$ admits an isolating partition with respect to the hash functions $\{\{H_{t, s}\}\}$ with probability at least $1-1/25$ by Lemma~\ref{lm:isolating-partition-construction}. Denote the success event by $\E_{partition}$. We condition on this event in what follows. We give the inductive argument, and finally define the event $\E_{maj}$ as the intersection of $\E_{partition}$ with several other high probability success events.

The {\bf base} is provided by $r=0$ and $t=0$. Indeed, by property {\bf (1)} of an isolating partition (see Definition~\ref{def:isolating-partition}) we have for any $q\in [1:T]$
$$
||y||_{S_q}\leq R^*\mu\cdot |S_q|\leq R^*\mu\cdot k\cdot \frac{R_0}{R_{q-1}}2^{-2^{(1-\delta)(q-1)}+1}\leq R^*\mu\cdot k\cdot (R_0/R_{q-1})^\delta
$$
since $2^{-2^{(1-\delta)(q-1)}+1}\leq 1$ for all $q\geq 1$ and $\frac{R_0}{R_{q-1}}\leq (R_0/R_{q-1})^\delta$ (as $\delta<1$ by assumption of the lemma).

We now prove the {\bf inductive step}. Let $(r', t'):=(r, t-1)$ if $t>1$, else let $(r', t'):=(r-1, T)$ if $r>1$ and $t=1$, and $(r', t')=(0, 0)$ otherwise. Note that $(t', r')$ is the element preceding $y^{(r, t)}$ in the majorizing sequence (as per lines 15-22 of Algorithm~\ref{alg:estimate-efficient}).

\noindent {\bf Proving (C)}. We start with an upper bound on the $\ell_1$ norm of $y^{(r', t')}$, i.e. prove {\bf (C)}. Using the inductive hypothesis {\bf (A)} and {\bf  (B)} for $(t', r')$, we get
\begin{equation}\label{eq:3ibrbrg-est}
\begin{split}
||y^{(r', t')}_S||_1&\leq \sum_{q=1}^{t'} (R^*(1/4)^{r'+1}\mu+2\mu)\cdot k\cdot (R_0/R_{q-1})^\delta+\sum_{q=t'+1}^{\infty} (R^*(1/4)^r\mu+2\mu)\cdot k\cdot (R_0/R_{q-1})^\delta\\
&\leq (R^*(1/4)^{r'}\mu+2\mu)\cdot k\cdot \sum_{q=1}^{\infty} (R_0/R_{q-1})^\delta\\
&= (R^*(1/4)^{r'}\mu+2\mu)\cdot k\cdot \sum_{q=1}^{\infty} 2^{-(q-1) \delta}\text{~~~~~~(since $R_t=C_1 2^t$ by {\bf p1})}\\
&\leq \frac1{2^\delta-1}\cdot (R^*k(1/4)^{r'}\mu+2\mu)\\
&\leq \frac1{e^{\delta \ln 2}-1}\cdot (R^*k(1/4)^{r'}\mu+2\mu)\\
&\leq (2/\delta)\cdot (R^*k(1/4)^{r'}\mu+2\mu)\text{~~~~~~~(since $e^x-1\geq x$, and $\ln 2> 1/2$)}\\
\end{split}
\end{equation}

This establishes {\bf (C)}, and we now turn to {\bf (A)} and {\bf (B)}.  By definition the signal $y^{(r, t)}$ is obtained from $y^{(r', t')}$ by modifying the latter on $S_t$. We need to bound the error introduced by head and tail elements of $y^{(r', t')}$ to $y^{(r, t)}_{S_t}$ (see~\eqref{eq:maj-def-estimation}).  We now bound both terms.

\noindent {\bf Proving (A) and (B):  analyzing contribution from the tail $e^{tail}$.}  By Lemma~\ref{lm:hashing}, {\bf (2)}, one has  for every $i\in [n]$, $t=1,\ldots, T$ and $s=1,\ldots, R_t$
$$
\expect_{H_{t, s}, a_{t, s}}\left[(e^{tail}_{i}(H_{t, s}, a_{t, s}, x))^2\right]\leq \nu^2
$$
for some $\nu>0$ such that $\nu^2=O(||x_{[n]\setminus S}||_2^2/B_t)$. 
By Jensen's inequality we thus have
$$
\expect_{H_{t, s}, a_{t, s}}\left[e^{tail}_{i}(H_{t, s}, a_{t, s}, x)\right]\leq \nu.
$$

To upper bound $\expect_{\{H_{t, s}, a_{t, s}\}}\left[||e^{tail}_{S_t}(\{H_{t, s}, a_{t, s}\}_{s\in [1:R_{t}]}, x)||_1\right]$, we note that by conditioning on $\E_{partition}$ we have 
$|S_t|\leq k\frac{R_0}{R_{t-1}}2^{-2^{(1-\delta)(t-1)}+1}$.
Letting $U:=k\frac{R_0}{R_{t-1}}2^{-2^{(1-\delta)(t-1)}+1}$ to simplify notation, we get that
\begin{equation}\label{eq:9hewGGFBNXNxx}
\begin{split}
\expect_{\{H_{t, s}, a_{t, s}\}}\left[||e^{tail}_{S_t}(\{H_{t, s}, a_{t, s}\}_{s\in [1:R_{t}]}, x)||_1\right]&\leq \expect_{\{H_{t, s}, a_{t, s}\}}\left[\max_{Q\subseteq S, |Q|\leq U} ||e^{tail}_Q(\{H_{t, s}, a_{t, s}\}_{s\in [1:R_{t}]}, x)||_1\right]\\
\end{split}
\end{equation}

We now recall that by ~\eqref{eq:et-pi-quant} one has
\begin{equation*}
e^{tail}_i(\{H_{t, s}, a_{t, s}\}, x):=\quant^{1/5}_{s=1,\ldots, R_t} e^{tail}_{i}(H_{t, s}, a_{t, s}, x),
\end{equation*}
 and apply Lemma~\ref{lm:noisiest-buckets} with $\gamma=1/5$, $m=|S|, n=R_t$ and 
$$
X_i^s=e^{tail}_{i}(H_{t, s}, a_{t, s}, x)\text{~~~for~}i\in S\text{~and~}s=1,\ldots, R_t,
$$
so that $\expect_{H_{t, s}, a_{t, s}}[X_i^s]\leq \nu$ for each $i\in S$, $s=1,\ldots, R_t$. Note that $Y_i:=\quant^{1/5}_{s=1,\ldots, R_t} X_i^s=e^{tail}_i(\{H_{t, s}, a_{t, s}\}, x)$ is exactly the quantity that we are interested in. 
We thus have by Lemma~\ref{lm:noisiest-buckets}
\begin{equation}\label{eq:pwhego23g}
\begin{split}
\expect_{\{H_{t, s}, a_{t, s}\}}\left[\max_{Q\subseteq S, |Q|\leq U} ||e^{tail}_Q(\{H_{t, s}, a_{t, s}\}_{s\in [1:R_{t}]}, x)||_1\right]&=\expect_{\{H_{t, s}, a_{t, s}\}}\left[\max_{Q\subseteq S, |Q|\leq U} \sum_{i\in Q}Y_i\right]\\
&\leq U\cdot (20e\nu)\cdot \left(|S|/U\right)^{10/R_t}
\end{split}
\end{equation}

Since $R_{t'}=C_1 2^{t'}$ for every $t'$,  $|S|=|S_0|\leq k$ and $U=k\frac{R_0}{R_{t-1}}2^{-2^{(1-\delta)(t-1)}+1}=k2^{-2^{(1-\delta)(t-1)}+1-(t-1)}$, we have
$$
\left(|S|/U\right)^{10/R_t}=2^{10(2^{(1-\delta)(t-1)}-1+(t-1))/(C_1 2^t)}\leq 2^{10(1+(t-1)/2^t)/C_1}\leq 2^{20/C_1}\leq 2
$$
for all $t\geq 1$ as long as $C_1>20$. Substituting the above into~\eqref{eq:pwhego23g}, we get 
\begin{equation*}
\expect_{\{H_{t, s}, a_{t, s}\}}\left[\max_{Q\subseteq S, |Q|\leq U} ||e^{tail}_Q(\{H_{t, s}, a_{t, s}\}_{s\in [1:R_{t}]}, x)||_1\right]\leq (40e) \cdot U\cdot \nu, 
\end{equation*}

and thus by~\eqref{eq:9hewGGFBNXNxx}
\begin{equation*}
\begin{split}
\expect_{\{H_{t, s}, a_{t, s}\}}\left[||e^{tail}_{S_t}(\{H_{t, s}, a_{t, s}\}_{s\in [1:R_{t}]}, x)||_1\right]&=O(U\cdot \nu)\\
&=O(2k\frac{R_0}{R_{t-1}} R_t 2^{-2^{(1-\delta)(t-1)}+1}\cdot ||x_{[n]\setminus S}||_2/\sqrt{C_2 k})\\
&=\mu k\frac{1}{R_{t-1}^2}\cdot O\left(R_t^2 R_0 2^{-2^{(1-\delta)(t-1)}+1}/\sqrt{C_2}\right)\\
&=\mu k\frac{1}{R_{t-1}^2}\cdot \xi_t, \text{~~~~}\xi_t=O\left(R_t^2 R_0 2^{-2^{(1-\delta)(t-1)}+1}/\sqrt{C_2}\right).\\
\end{split}
\end{equation*}
Since $R_t=C_1 2^t$ increases only exponentially, whereas the second multiplier decreases at a doubly exponential rate, as long as $C_2$ is larger than a constant, we get that $\xi_t\leq 1/10000$ for all $t\geq 1$ (formally, this follows by Claim~\ref{cl:max-expr}). By Markov's inequality, for each $t\geq 1$ we have
\begin{equation*}
||e^{tail}_{S_t}(\{H_{t, s}, a_{t, s}\}_{s\in [1:R_{t}]}, x)||_1\leq \frac1{200}\mu k\frac{1}{R_{t-1}}
\end{equation*}
with probability at least $1-1/(50 R_{t-1})\geq 1-1/(50 \cdot 2^{t-1})$. Thus, by a union bound over all $t\geq 1$ we have with probability at least $1-1/25$
\begin{equation}\label{eq:tail-bound-12345-est}
||e^{tail}_{S_t}(\{H_{t, s}, a_{t, s}\}_{s\in [1:R_{t}]}, x)||_1\leq \frac1{200}\mu k\frac{1}{R_{t-1}}.
\end{equation}
Denote the success event above by $\mathcal{E}_{small-noise}$. 

\noindent {\bf Proving (A) and (B):  analyzing contribution from the head $e^{head}$.} By Lemma~\ref{lm:l1b-head} we have 
\begin{equation}\label{eq:head-bound-12345-est}
||e^{head}_{S_t}(\{H_{t, s}\}_{s\in [1:R_t]}, y^{(r', t')})||_1\leq 20  R_t^{-\delta}||y^{(r', t')}_S||_1.
\end{equation}
  We now define the event $\E_{maj}$ by letting $\E_{maj}:=\E_{small-noise}\cap \E_{partition}$. Note that $\prob[\E_{maj}]\geq 1-2/25$ by a union bound, as required. We condition on $\E_{maj}$ for the rest of the proof.

\noindent {\bf Proving (A) and (B):  putting it together.} We now use the bounds above to prove the result. By definition of $y$ above we have 
\begin{equation*}
\begin{split}
y^{(r, t)}_i:=20 e^{head}_i(\{H_{t, s}\}_{s\in [1:R_t]}, y^{(r',t')})+20 e^{tail}_i(\{H_{t, s}, a_{t, s}\}_{s\in [1:R_t]}, x)+n^{-\Omega(c)},
\end{split}
\end{equation*}
so 
\begin{equation*}
\begin{split}
||y^{(r, t)}_{S_t}||_1&\leq \sum_{i\in S_t}\left(20 e^{head}_i(\{H_{t, s}\}_{s\in [1:R_t]}, y^{(r',t')})+20 e^{tail}_i(\{H_{t, s}, a_{t, s}\}_{s\in [1:R_t]}, x)+n^{-\Omega(c)}\right)\\
&\leq 400\cdot  R_t^{-\delta} \cdot ||y^{(r',t')}_S||_1+20||e^{tail}_{S_t}(\{H_{t, s}, a_{t, s}\}, x)||_1+n^{-\Omega(c)}\\
\end{split}
\end{equation*}
We now substitute ~\eqref{eq:3ibrbrg-est} together with~\eqref{eq:tail-bound-12345-est} and~\eqref{eq:head-bound-12345-est} into the last line above, and obtain
\begin{equation*}
\begin{split}
||y^{(r, t)}_{S_t}||_1&\leq 400\cdot  R_t^{-\delta} (2/\delta)\cdot (R^*(1/4)^{r'}\mu+2\mu)k+  \frac1{10}\mu k/R_{t-1}+n^{-\Omega(c)}\\
&\leq 400\cdot (2/\delta)\cdot  R_0^{-\delta}\cdot (R^*(1/4)^{r'}\mu+2\mu)k\cdot (R_0/R_{t-1})^\delta+  \frac1{10}\mu k\cdot (R_0/R_{t-1})^\delta+n^{-\Omega(c)}\text{~~~~~~(since $\delta\in (0, 1)$)}\\
&\leq \left[400\cdot (2/\delta)\cdot  C_1^{-\delta}+\frac1{10}\right]\cdot (R^*(1/4)^{r'}\mu+2\mu)k\cdot (R_0/R_{t-1})^\delta+  n^{-\Omega(c)}\\
\end{split}
\end{equation*}
where we upper bounded $\frac1{10}\mu k/R_{t-1}$ by $\frac1{10}\mu k\cdot (R_0/R_t)^\delta$ (which is justified since $R_0\geq 1$ and $\delta\in (0, 1)$)  and used the bound 
$R_t^{-\delta}=R_0^{-\delta}\cdot (R_0/R_t)^\delta$.

We now conclude that as long as $C_1\geq (40000/\delta)^{1/\delta}$, we have 
\begin{equation*}
\begin{split}
||y^{(r, t)}_{S_t}||_1&\leq \left[400\cdot (2/\delta)\cdot  C_1^{-\delta}+\frac1{10}\right]\cdot (R^*(1/4)^{r'}\mu+2\mu)k\cdot (R_0/R_{t-1})^\delta+  n^{-\Omega(c)}\\
&\leq  (R^*(1/4)^{r'+1}\mu+2\mu)k\cdot (R_0/R_{t-1})^\delta+  n^{-\Omega(c)}.
\end{split}
\end{equation*}
This completes the proof of the inductive step for {\bf (A)} and {\bf (B)}. It remains to prove {\bf (D)}.

\noindent {\bf Proving (D).} Our main tool in arguing {\bf (D)} is Lemma~\ref{lm:estimate-l1l2}, which we invoke with the set $S$. By that lemma we have for every $i\in S$
$$
|x_i-\chi^{(r', t')}_i-\chi'_i|\leq 2\quant^{1/5}_s e^{head}_i(H_{t, s}, x-\chi^{(r', t')})+2\quant^{1/5} e^{tail}_i(H_{t, s}, a_{t, s}, x)+n^{-\Omega(c)},
$$
since $\text{supp~} \chi^{(r', t')} \subseteq S$ by the inductive hypothesis.
This implies by definition of $y^{(r, t)}$ that for every $i\in S_t$ 
\begin{equation}\label{eq:2o3nognewgwge}
\begin{split}
|x_i-\chi^{(r', t')}_i-\chi'_i|&\leq 2\quant^{1/5}_r e^{head}_i(H_{t, s}, x-\chi^{(r', t')})+2\quant^{1/5}_s e^{tail}_i(H_{t, s}, a_{t, s}, x)+n^{-\Omega(c)}\\
&\leq 20\quant^{1/5}_s e^{head}_i(H_{t, s}, x-\chi^{(r', t')})+20\quant^{1/5}_s e^{tail}_i(H_{t, s}, a_{t, s}, x)+n^{-\Omega(c)}.\\
%&= y^{(r, t)}_i.
\end{split}
\end{equation}

By part {\bf (D)} of the inductive hypothesis we have 
$x-\chi^{(r', t')}\prec_S y^{(r', t')}$, and thus by Lemma~\ref{lm:mon} together with~\eqref{eq:2o3nognewgwge}  for every $i\in S_t$
\begin{equation*}
\begin{split}
|x_i-\chi^{(r', t')}_i-\chi'_i|&\leq 20\quant^{1/5}_s e^{head}_i(H_{t, s}, x-\chi^{(r', t')})+20\quant^{1/5}_s e^{tail}_i(H_{t, s}, a_{t, s}, x)+n^{-\Omega(c)}\\
&\leq 20\quant^{1/5}_s e^{head}_i(H_{t, s}, y^{(r', t')})+20\quant^{1/5}_s e^{tail}_i(H_{t, s}, a_{t, s}, x)+n^{-\Omega(c)}\\
&= y^{(r, t)}_i.
\end{split}
\end{equation*}

We thus have for every $i\in S_t$
$$
|x_i-\chi^{(r', t')}_i-\chi'_i|\leq y^{(r, t)}_i.
$$
Since $y^{(r, t)}_i=y^{(r', t')}_i$ for $i\not \in S_t$, $\chi^{(r, t)}_i=\chi^{(r', t')}_i$ for $i\not \in S_t$ and $x-\chi^{(r', t')}\prec_S y^{(r', t')}$ by the inductive hypothesis, we get 
$$
x-\chi^{(r', t')}-\chi'\prec_S y^{(r, t)}
$$
as required. Since we only update elements of $S$, we have $\text{supp~} \chi^{(r, t)}\subseteq \text{supp~} \chi^{(r', t')}\cup \text{supp~} \chi'\subseteq S$. This completes the proof of {\bf (D)}, and the proof of the induction.

To obtain the final result of the lemma, we note that by part {\bf (C)} of the inductive claim for every $r\geq C\log_4 R^*$ (for any $C\geq 1$) one has 
$$
||y^{(r, T)}_S||_1\leq (2/\delta)\cdot (R^*(1/4)^r\mu+2\mu)\cdot k\leq (6/\delta)\cdot \mu k.
$$
Now recall that by line~23 of Algorithm~\ref{alg:estimate-efficient} we have $\tilde \chi=\chi^{(C\log_4 R^*, T)}$, which implies by part {\bf (D)} of the inductive claim, since $(x-\chi^{(C\log_4 R^*, T)})\prec_S y^{(C\log_4 R^*, T)}$, that 
$$
||(x-\tilde \chi)_S||_1=||(x-\chi^{(C\log_4 R^*, T)})_S||_1\leq ||y^{(C\log_4 R^*, T)}_S||_1\leq (6/\delta)\cdot \mu k=O(\mu k),
$$
as required.

\end{proofof}

\subsection{Proof of Theorem~\ref{thm:main-estimate}}
We now give 

\begin{proofof}{Theorem~\ref{thm:main-estimate}}
Recall that in this section we use the quantities $e^{head}$ and $e^{tail}$ defined with respect to the set $S$.
By Lemma~\ref{lm:main-estimate-snr-reduction} we have that conditioned on a high probability event $\E_{maj}$ (which occurs with probability at least $1-2/25$) the vector $\tilde \chi$ computed in line~23 satisfies
\begin{equation}\label{eq:l1-bound-final-est}
||(x-\tilde \chi)_S||_1=O(||x_{[n]\setminus S}||_2 \sqrt{k})
\end{equation}

To complete the proof, we show that the output $\chi''$ of the invocation of \textsc{EstimateValues} in line~31, when added to $\tilde \chi$, yields guarantee claimed by the lemma. First, by Lemma~\ref{lm:estimate-l1l2}  with $S$ one has for each $i\in S$  
\begin{equation}\label{eq:diff-bound}
|\chi''_i- (x-\tilde \chi)_i|\leq 2\cdot\quant^{1/5}_r e^{head}_i(H_r, x-\tilde \chi)+ 2\cdot\quant^{1/5}_r e^{tail}_i(H_r, a_r, x)+n^{-\Omega(c)},
\end{equation}
since $\supp \tilde \chi\subseteq S$. 

Squaring both sides of ~\eqref{eq:diff-bound}, using the bound $(a+b)^2\leq 2a^2+2b^2$ and taking expectations over the randomness in measurements taken in lines~27-30,  we get
\begin{equation}\label{eq:est-est}
\expect[|\chi''_i- (x-\tilde \chi)_i|^2]\leq 8\cdot\expect\left[(\quant^{1/5}_r e^{head}_i(H_r, x-\tilde \chi)^2\right]+ 8\cdot\expect\left[(\quant^{1/5}_r e^{tail}_i(H_r, a_r, x))^2\right]+n^{-\Omega(c)}.
\end{equation}

We now upper bound the expectation of~\eqref{eq:est-est}. By Lemma~\ref{lm:median-etail-ehead}, {\bf (1)} one has, letting $Z^{head}:=\quant^{1/5}_r e^{head}_i(H_r, x-\tilde \chi)$ to simplify notation,
$$
\expect\left[(Z^{head})^2\right]=O\left(\left(\frac1{B}||(x-\tilde \chi)_S||_1\right)^2\right)=O\left(\left(\frac1{C_2 k/\e}||(x-\tilde \chi)_S||_1\right)^2\right)=O(\e^2 ||x_{[n]\setminus S}||_2^2/(C_2k)),
$$
where we used that by conditioning on $\E_{maj}$ one has $||(x-\tilde \chi)_S||_1=O(||x_{[n]\setminus S}||_2 \sqrt{k})$ (by~\eqref{eq:l1-bound-final-est}).

By Lemma~\ref{lm:median-etail-ehead}, {\bf (2)} with $S$ one has, letting $Z^{tail}:=\quant^{1/5}_r e^{tail}_i(H_r, a_r, x)$ to simplify notation, 
$$
\expect\left[(Z^{tail})^2\right]=O(||(x-\tilde \chi)_{[n]\setminus S}||_2^2/B)=O(\e ||x_{[n]\setminus S}||_2^2/(C_2 k)),
$$ 
where we used the fact that $\supp \tilde\chi \subseteq S$.

Substituting these bounds into~\eqref{eq:est-est} and summing over all $i\in S$, we get  

\begin{equation*}
\expect[||(x-\tilde \chi-\chi'')_S||^2]\leq O(\e^2 ||x_{[n]\setminus S}||_2^2/C_2)+O(\e ||x_{[n]\setminus S}||_2^2/C_2)\leq (\e/1000) ||x_{[n]\setminus S}||_2^2
\end{equation*}
as long as $C_2$ is sufficiently large.

An application of Markov's inequality then gives $||(x-\tilde \chi-\chi'')_S||^2\leq \e ||x_{[n]\setminus S}||_2^2$ with probability at least $1-1/1000$. By a union bound over this failure event and $\bar \E_{maj}$, we conclude that the algorithm outputs the correct answer with probability at least $1-3/25\geq 4/5$.

We now upper bound the sample complexity and runtime.

{\bf Sample complexity.} The sample complexity of lines 6-11 is bounded by $\sum_{t=1}^T \sum_{s=1}^{R_t} O(F\cdot B_t )=\sum_{t=1}^T R_t\cdot O(F\cdot k/R_t^2 )=O(k)\cdot \sum_{t=1}^T 1/R_t=O(k)$ by the choice of $R_t$ as geometrically increasing. The sample complexity of lines~27-30 is upper bounded by $O(F\cdot B)=O(k/\e)$ by Lemma~\ref{lm:estimate-l1l2} and choice of $F=O(1)$.

{\bf Runtime.} The runtime of \textsc{HashToBins} in line~10 of Algorithm~\ref{alg:estimate-efficient} is $O(F\cdot B_t \log B_t)=O(B_t \log B_t)$ by Lemma~\ref{l:hashtobins}, the setting of $F=O(1)$ and the fact that the residual signal passed to the call is zero. Since this line is executed for $t=1,\ldots, T$ and $s=1,\ldots, R_t$, the total runtime of the loop is 
$$
\sum_{t=1}^T \sum_{s=1}^{R_t}  O(B_t \log B_t)= O\left(\sum_{t=1}^T R_t \cdot (C_2 k/R_t^2) \log (C_2k)\right)=O(k\log k)\cdot \sum_{t=1}^T 1/R_t=O(k\log k).
$$

The runtime for construction of the partition $S_1\cup S_2\cup \ldots \cup S_T$ in line~13 is $O((\sum_{t=1}^T R_t) |S|\log |S|)=O(R_T k\log  k)$ by Lemma~\ref{lm:construct-partition-runtime} and the fact that $\sum_{t=1}^T R_t=O(R_T)$.
We now note that since $T=\frac1{1-\delta}\log_2 \log (k+1)+O(1)$, then 
\begin{equation}\label{eq:rt}
R_T=C_1 2^T= C_1 2^{\frac1{1-\delta}\log_2 \log k+O(1)}=O(\log_2^{1/(1-\delta)} (k+1))=O(\log_2^{1+2\delta} (k+1)),
\end{equation}
where we used the fact that $1/(1-\delta)\leq 1+2\delta$ for $\delta\in (0, 1/2)$.
Thus, the runtime for construction of the partition $S_1\cup S_2\cup \ldots \cup S_T$ in line~13 is $O(k\log^{2+2\delta}  k)$.

By Lemma~\ref{lm:estimate-l1l2} each invocation of \textsc{EstimateValues} takes time $O((||\chi^{(r, t)}||_0\log n+F B_t \log n)\cdot R_t)=O(R_t k\log n+R_t B_t \log n)$, as $F=O(1)$ by choice of parameters in line~25 of Algorithm~\ref{alg:estimate-efficient}. 
The total runtime per iteration in lines~16-22 is thus 
\begin{equation*}
\begin{split}
&\sum_{t=1}^T O(R_t k\log n+R_t B_t \log n)\\
&=O\left(k  R_T\log n+\sum_{t=1}^T B_t R_t \log n\right)\text{~~~~~(since $\sum_{t=1}^T R_t=O(R_T)$, as $R_t$ grow geometrically)}\\
&=O(k R_T \log n)+O\left(\sum_{t=1}^T k/R_t\right) \log n\text{~~~~~(since $B_t=C_2 k/R_t^2$)}\\
\end{split}
\end{equation*}
We now note that $\sum_{t=1}^T k/R_t=O(k)$ since $R_t$ grow geometrically, and thus the expression on the last line above is $O(kR_T\log n)=k\log^{2+2\delta} n$ by~\eqref{eq:rt}. Since the loop in lines~16-22 proceeds over $O(\log n)$ iterations, the final runtime bound is  $k\log^{3+2\delta} n$, as required (after rescaling $\delta$).

Finally, lines 27-31 take $O(\frac1{\e}k\log n)$ time for the invocation of \textsc{HashToBins}  by Lemma~\ref{l:hashtobins} and $O(\frac1{\e}k\log n)$ time for \textsc{EstimateValues}  by Lemma~\ref{lm:estimate-l1l2}. Putting the bounds above together, we obtain the runtime of $k\log^{3+2\delta} n+O(\frac1{\e} k\log n)$, as required.
\end{proofof}

%% file: algo.tex
%!TEX root = ./ft-hd-wrap.tex 
\section{Sample efficient recovery}\label{sec:sublinear}

In this section we state our algorithm for sparse recovery from Fourier measurements that achieves $O(k\log n)$ sample complexity in $k\log^{O(1)} n$ runtime, give an outline of the analysis, and then present the formal proof. 
The proof reuses the core primitives developed in Section~\ref{sec:l1} together with the idea of majorizing sequences used in Section~\ref{sec:est} to argue about correctness of our estimation primitive to analyze the performance of a natural iterative recovery scheme.

\subsection{Algorithm and outline of the analysis}
 Our algorithm (Algorithm~\ref{alg:main-sublinear}) contains three major components: it starts by taking measurements $m$ of the signal $x$ (accessing the signal in Fourier domain, i.e. accessing $\wh{x}$), then uses these measurements to perform a sequence of recovery steps that reduce the $\ell_1$ norm of the `heavy' elements of $x$ down to (essentially) noise level $\mu$. Finally, a simple cleanup procedure (\textsc{RecoverAtConstantSNR}) is run to achieve the $\ell_2/\ell_2$ sparse recovery guarantees (see ~\eqref{e:l2l2}). We reuse the location primitive from~\cite{K16} (\textsc{LocateSignal}, Algorithm~\ref{alg:location}).

\noindent{\bf Measuring $\wh{x}$.} All measurements that the algorithm takes are taken in lines~6-22. Two sets of measurements are taken: one for location (\textsc{LocateSignal}), another for estimation purposes (calls to \textsc{EstimateValues} in line~32 of Algorithm~\ref{alg:main-sublinear}). Location relies on a very structured set of measurements: the measurements are taken over $T=\frac1{1-\delta}\log_2\log (k+1)+O(1)$ rounds for small constant $\delta\in (0, 1/2)$, where in round $t$ we are hashing the signal into $B_t\approx k/R_t^2$ buckets, where $R_t$ grows exponentially with $t$. For each $t$ we perform $R_t$ independent hashing experiments of this type. For each hashing $H_{t, s}, t=1,\ldots, T, s=1,\ldots, R_t$ we select a random set $\A_{t, s}\subseteq\nsq\times \nsq$ that encodes the locations that our measurements access.
Besides measurements used for location we take a separate set of measurements to use in the call to \textsc{EstimateValues}. These measurements are quite unstructured: we simply make measurements using $C\log n$ random hashings and evaluation points for sufficiently large constant $C>0$. It is crucial that these measurements are independent of the measurements used for location. Intuitively, the first set of measurements allows us to decode dominant coefficients of the residual signal in sublinear time, whereas the second (unstructured) set of measurements allows us to prune false positives, ensuring that no erroneous coefficients are introduced throughout the update process. The latter idea is similar to the approach used in~\cite{IK14a}, but is harder to implement in our setting as the number of possible trajectories along which the decoding process can evolve is larger. We handle this issue by using the notion of majorizing sequences introduced in Section~\ref{sec:prelim} (see Definition~\ref{def:majorant} and Lemma~\ref{lm:mon}) and used to analyze Algorithm~\ref{alg:estimate-efficient} in Section~\ref{sec:est}.

\noindent{\bf Signal to noise ratio (SNR) reduction loop.} Once the samples have been taken, Algorithm~\ref{alg:main-sublinear} proceeds to the signal to noise (SNR) reduction loop (lines~25-36). The objective of this loop is to reduce the mass of the top (about $k$) elements in the residual signal to roughly the noise level $\mu\cdot k$, where $\mu\geq ||x_{\nsq\setminus [k]}||_2/\sqrt{k}$. Specifically, we define the set $S$ of `head elements' in the original signal $x$ as 
\iflong{\begin{equation}\label{eq:s-def-l1}
S=\{i\in \nsq: |x_i|>\mu\}.
\end{equation}
 Note that we have $|S|\leq 2k$. Indeed, if $|S|>2k$, more than $k$ elements of $S$ belong to the tail, amounting to more than $\mu^2\cdot k=\err_k^2(x)$ tail mass. The quantities $e^{head}$ and $e^{tail}$ (see~\eqref{eq:eh} and~\eqref{eq:et-pi} in Section~\ref{sec:prelim}) used in this section are defined with respect to this set $S$.

 The SNR reduction loop of Algorithm~\ref{alg:main-sublinear} constructs a vector $\tilde \chi$ supported only on $S$ such that 
\begin{equation}\label{eq:l1snr-o1}
 ||(x-\tilde \chi)_S||_1=O(\mu k)\text{~~~and~~~}\supp \tilde \chi\subseteq S,
\end{equation}
 i.e. the $\ell_1$-SNR of the residual signal on the set $S$ of heavy elements is reduced to a constant. }

The main technical contribution lies in our SNR reduction loop, and our main technical result in this section is
\begin{theorem}\label{thm:l1snr}
For any  $\delta\in (0, 1/2)$, for any $x\in \C^n$, any integer $k\geq 1$, if $\mu^2\geq \err_k^2(x)/k$ and $R^*\geq ||x||_\infty/\mu, R^*=n^{O(1)}$, the following conditions hold for the set $S:=\{i\in \nsq: |x_i|>\mu\}\subseteq \nsq$.

Then the SNR reduction loop of Algorithm~\ref{alg:main-sublinear} (lines~25-36) returns $\tilde \chi$ such that 
\begin{equation*}
\begin{split}
&||(x-\tilde \chi)_S||_1=O_\delta(\mu k)\\
&\supp \tilde \chi\subseteq S
\end{split}
\end{equation*}
with probability at least $1-3/25$ over the internal randomness used by Algorithm~\ref{alg:main-sublinear}. The sample complexity is $O_\delta(k\log n)$. The runtime is bounded by $O_\delta(k\log^{4+2\delta} n)$.
\end{theorem}

\noindent{\bf Recovery at constant $\ell_1$-SNR and final result.} Once ~\eqref{eq:l1snr-o1} has been achieved, we run the \textsc{RecoverAtConstantSNR} primitive from~\cite{K16} on the residual signal. Adding the correction that it outputs to the output  of the SNR reduction loop gives the final output of the algorithm.  Given Theorem~\ref{thm:l1snr}, the proof of the main result is simple using 

\begin{lemma}[Lemma~3.4 of~\cite{K16}]\label{lm:const-snr}
For any $\e>0$, $\hat x, \chi\in \C^n$, $x'=x-\chi$ and any integer $k\geq 1$ if $||x'_{[2k]}||_1\leq O(||x_{\nsq\setminus [k]}||_2\sqrt{k})$ and $||x'_{\nsq\setminus [2k]}||_2^2\leq ||x_{\nsq\setminus [k]}||_2^2$, the following conditions hold. If $||x||_\infty/\mu=n^{O(1)}$, then the output $\chi'$ of 
\Call{RecoverAtConstantSNR}{$\hat x, \chi, 2k, \e$} satisfies $||x'-\chi'||^2_2\leq (1+O(\e))||x_{\nsq\setminus [k]}||_2^2$
with at least $99/100$ probability over its internal randomness. The sample complexity is $O(\frac1{\e}  k\log n)$, and the runtime complexity is at most $O(\frac1{\e}  k \log^2 n).$
\end{lemma}

{\noindent {\bf Theorem~\ref{thm:main}}  (Restated) \em 
For any $\e\in (1/n, 1), \delta\in (0, 1/2)$, $x\in \C^n$ and any integer $k\geq 1$, if $R^*\geq ||x||_\infty/\mu, R^*=n^{O(1)}$, $\mu^2\geq ||x_{\nsq\setminus [k]}||_2^2/k$, $\mu^2=O(||x_{\nsq\setminus [k]}||_2^2/k)$, \textsc{SparseFFT}$(\hat x, k, \e, R^*, \mu)$  (Algorithm~\ref{alg:main-sublinear}) solves the $\ell_2/\ell_2$ sparse recovery problem using $O_\delta(k\log n)+O(\frac1{\e}k\log n)$ samples and 
$O_\delta(\frac1{\e}k \log^{4+\delta} n)$ time with at least $4/5$ success probability.
}
\begin{proof}
Let the set $S\subseteq [n]$ be defined as in Theorem~\ref{thm:l1snr}.
By Theorem~\ref{thm:l1snr} one has that 
$||(x-\tilde \chi)_S||_1=O_\delta( \mu)$ and $\supp \tilde \chi\subseteq S$ with probability at least $1-3/25$. 
Thus, the signal $x-\tilde \chi$ satisfies preconditions of Lemma~\ref{lm:const-snr}, and we get
$||x-\tilde \chi-\chi'||_2\leq  (1+O(\e))\err_{k}(x)$  with probability at least $99/100$, resulting in success probability at least $1-3/25-1/100\geq 4/5$ overall.

The sample complexity of the SNR reduction loop is $O(k\log n)$ by Theorem~\ref{thm:l1snr}. The sample complexity of \textsc{RecoverAtConstantSNR} is $O(\frac1{\e}k\log n)$. The runtime of the SNR reduction loop is bounded by $k \log^{4+2\delta} n$ by Theorem~\ref{thm:l1snr}, and the runtime of \textsc{RecoverAtConstantSNR} is at most $O(\frac1{\e}k \log^2 n)$ by Lemma~\ref{lm:const-snr}, so the final runtime bound follows (after rescaling $\delta$).
\end{proof}

\iflong{
\begin{algorithm}[h!]
\caption{\textsc{SparseFFT}($\hat x, k, \e, R^*, \mu$)}\label{alg:main-sublinear} 
\begin{algorithmic}[1] 
\Procedure{SparseFFT}{$\hat x, k, \e, R^*, \mu$}
\State $\H\gets \{\mathbf{0}\}$, $\Delta\gets 2^{\lfloor \frac1{2}\log_2 \log_2 n\rfloor}$, $N\gets \Delta^{\lceil \log_\Delta n\rceil}$
\For {$g=1$ to $\log_\Delta N$}
\State $\H\gets \H\cup \{N \Delta^{-g}\}$
\EndFor
\State $T\gets \frac1{1-\delta}\log_2 \log (k+1)+O(1)$ 
\State $R_t\gets C_1\cdot 2^t$ for $t\in [1:T]$ \Comment{$C_1>0$ an absolute constant, $\delta\in (0, 1/2)$ small constant}
\State $B_t\gets C_2 \cdot k/R_t^2$ for $t\in [1:T]$ \Comment{$C_2$ sufficiently large}
\State $G_t\gets$ filter with $B_t$ buckets and sharpness $\fc=8$.
\For {$t=1$ to $T$}\Comment{Take samples to be used for location}
\For {$s=1$ to $R_t$}
\State Choose $\sigma\in \gl$ u.a.r., let $\pi_{t, s}\gets (\sigma, 0)$, $H_{t, s}:=(\pi_{t, s}, B_t, \fc)$
\State Let $\A_{t, s}\gets $ $C\log\log n$ elements of $\nsq\times \nsq$ u.a.r.
\State $m(x, H_{t, s}, \alpha+\h\cdot \beta) \gets \Call{HashToBins}{\hat x, 0, (H_{t, s}, \alpha+\h\cdot \beta)}$ for $(\alpha, \beta)\in \A_{t, s}, \h\in \H$
\EndFor
\EndFor
\State $B\gets k/\alpha^2$, $\alpha\in (0, 1)$ smaller than a constant 
\For {$t=1$ to $C\log n$}
\State Choose $\sigma\in \gl$, $q, z_t\in \nsq$ u.a.r., let $\pi^{est}_t\gets (\sigma, q)$, $ H^{est}_t:=(\pi^{est}_t, B, \fc)$
\State $m(x, H^{est}_t, z_t) \gets \Call{HashToBins}{\hat x, 0, (H^{est}_t, z_t)}$
\EndFor
\State $\M^{est}\gets \{(H^{est}_t, z_t, m(x, H^{est}_t, z_t))\}_{t=1}^{C\log n}$
\State $\chi^{(0, 0)}\gets 0$, $\chi'\gets 0$
\State $r'\gets 0, t'\gets 0$
\For{$r = 0, 1, \dotsc,\left\lfloor \log_4 R^*\right\rfloor-3$}
\For{$t=1$ to $T$}
\State $\chi^{(r, t)}\gets \chi^{(r', t')}+\chi'$
\For{$s=1$ to $R_t$} \Comment{Invocation of \textsc{LocateSignal} below does not take any fresh samples}
\State $L_s\gets \Call{LocateSignal}{\chi^{(r', t')},  H_{t, s},  \{m(x, H_{t, s}, \alpha+\h\cdot \beta)\}_{(\alpha, \beta)\in \A_{t, s}, \h\in \H}}$
\EndFor
\State $L\gets \bigcup_{s=1}^{R_t} L_s$\Comment{Invocation of \textsc{EstimateValues} below does not take any fresh samples}
\State $\chi \gets \Call{EstimateValues}{\chi^{(r', t')}, L, \M^{est}}$
\State For all $j\in \supp  \chi$ let $\chi'_j\gets  \chi_j$ if $| \chi_j|\geq \frac1{16}R^*\mu(1/4)^r$ and $\chi'_j\gets 0$ otherwise
\State $r'\gets r$, $t'\gets t$
\EndFor
\EndFor
\State $\tilde \chi \gets \chi^{(r', t')}+\chi'$
\State $\chi'' \gets \textsc{RecoverAtConstantSNR}(\hat x, \tilde \chi, 2k, \eps)$
\State $\chi^*\gets \tilde \chi+\chi''$
\State \textbf{return} $\chi^*$
\EndProcedure 
\end{algorithmic}
\end{algorithm}

}

%% file: combo.tex
%!TEX root = ./ft-hd-wrap.tex 
%\section{Construction of the majorizing sequence and analysis of SNR reduction loop}\label{sec:combo}

\iflong{In the rest of this section we prove performance guarantees for the SNR reduction loop in Algorithm~\ref{alg:main-sublinear} (lines~23-35). These guarantees are formally stated in Theorem~\ref{thm:l1snr}, our main result in the rest of the section.
The main tool in our analysis is the notion of a majorizing sequence for the intermediate residual signals that arise in the SNR reduction loop: we show that with high probability over the measurements taken, the intermediate residual signals that arise during the execution of the algorithm are (assuming perfect estimation) majorized by a fixed sequence of signals $y^{(r, t)}$, constructed in section~\ref{sec:maj}. \ifshort{We then argue that all signals in the majorizing sequence are approximated by the measurements $\M^{est}$ well with high probability, which lets us prove by induction that all residual signals arising in our algorithm are majorized appropriately.}

To prove that the residual signal is indeed with high probability majorized by this sequence $y^{(r, t)}$, we use the fact that our estimation primitive uses $C\log n$ random measurements and hence yields precise bounds for all signals $y^{(r, t)}$ in the majorizing sequence. This means that estimates provided by \textsc{EstimateValues} essentially provide perfect estimation for our algorithm, and a simple inductive argument shows that $y^{(r, t)}$ majorizes $x-\chi^{(r, t)}$ at each iteration indeed, and no false positives are created. This argument crucially relies on the definition of a majorant (see Definition~\ref{def:majorant}) and a monotonicity property of $e^{head}$ (Lemma~\ref{lm:mon}). We first state notation relevant to bounding the effect of tail noise on location in section~\ref{sec:loc-tail-noise}. Then the construction of the majorizing sequence is given in section~\ref{sec:maj}, and then section~\ref{sec:snr-loop} proves Theorem~\ref{thm:l1snr}.

\subsection{Notation for bounding tail noise in location}\label{sec:loc-tail-noise}

 Our location algorithm (presented in Appendix~\ref{sec:location}) uses several values of $(\alpha, \beta)\in \A_r\subseteq \nsq\times \nsq$ to perform location, a more robust version of $e^{tail}_i(H, z)$ will be useful. To that effect we let for any $\mathcal{Z}\subseteq \nsq$ 
\begin{equation}\label{eq:et-pi-a-h}
e^{tail}_i(H, \mathcal{Z}, x):=\quant^{1/5}_{z\in \mathcal{Z}} \left|G_{o_i(i)}^{-1}\cdot \sum_{j\in \nsq\setminus S} G_{o_i(j)}x_j \omega^{z \sigma (j-i)}\right|.
\end{equation}
Note that our Sparse FFT algorithm (Algorithm~\ref{alg:main-sublinear})  at various iterations $r$, first selects sets $\A_r\subseteq \nsq\times \nsq$, and then accesses the signal at locations $\mathcal{Z}=\{\alpha+\h\cdot \beta\}_{(\alpha, \beta)\in \A_r}$ for various $\h\in \H$. It should also be noted here that in the definition above the quantile is taken over all values of $z\in \mathcal{Z}$ for a fixed hashing $H$.

The definition of $e^{tail}_i(H, \{\alpha+\h \cdot \beta\}, x)$ for a fixed $\h\in \H$ above allows us to capture the amount of noise that our measurements that use $H$ suffer from for locating a specific set of bits of $\sigma i$. Since the algorithm requires all $\h\in \H$ to be not too noisy in order to succeed, the following quantity will be useful in analysis. We define
\begin{equation}\label{eq:et-pi-a}
e^{tail, \H}_i(H, \A, x):=40\mu_{H, i}(x)+\sum_{\h\in \H} \left|e^{tail}_i(H, \{\alpha+\h \cdot \beta\}_{(\alpha, \beta)\in \A}, x)-40\mu_{H, i}(x)\right|_+
\end{equation}
where for any $\eta\in \mathbb{R}$ one has $|\eta|_+=\eta$ if $\eta>0$ and $|\eta|_+=0$ otherwise. 

The following definition is useful for bounding the norm of elements $i\in S$ that are not discovered by several calls to \textsc{LocateSignal} on a sequence of hashings $\{H_r\}$. For a sequence of measurement patterns $\{H_r, \A_r\}$ we let
\begin{equation}\label{eq:et}
e^{tail, \H}(\{H_r, \A_r\}, x):=\quant^{1/5}_r e^{tail, \H}_i(H_r, \A_r, x).
\end{equation}
%Finally, for any $T\subseteq \nsq$ we let $e^{head}_T(\cdot):=\sum_{i\in T} e^{head}_i(\cdot)$ and $e^{tail}_T(\cdot):=\sum_{i\in T} e^{tail}_i(\cdot)$, where $\cdot$ stands for any set of parameters as above.

We will use the following lemma, whose proof is given in Appendix~\ref{sec:location}:

\begin{lemma}\label{cor:loc}
For any integer $r_{max}\geq 1$,  for any sequence of $r_{max}$ hashings $H_r=(\pi_r, B, R), r\in [1:r_{max}]$ and evaluation points $\A_r\subseteq \nsq\times \nsq$,  for every $S\subseteq \nsq$ and for every $x, \chi\in \C^n, x':=x-\chi$, the following conditions hold.
If for each $r\in [1:r_{max}]$ $L_r\subseteq \nsq$ denotes the output of \Call{LocateSignal}{$\wh{x}, \chi, H_r, \{m(x, H_r, \alpha+\h\cdot \beta)\}_{(\alpha, \beta)\in \A_r, \h\in \H}$}, $L=\bigcup_{r=1}^{r_{max}} L_r$, and the sets $\{\beta\}_{(\alpha, \beta)\in \A_r}$ are balanced  $r\in [1:r_{max}]$, then
\begin{equation}
||x'_{S\setminus L}||_1\leq 20 ||e^{head}_S(\{H_r\}, x')||_1+20 ||e^{tail, \H}_S(\{H_r, \A_r\}, x)||_1+|S|\cdot n^{-\Omega(c)}.\tag{*}
\end{equation}
Furthermore, every element $i\in S$ such that 
\begin{equation}
|x'_i|>20 (e^{head}_i(\{H_r\}, x')+e^{tail, \H}_i(\{H_r, \A_r\}, x))+n^{-\Omega(c)}\tag{**}
\end{equation}
belongs to $L$.
\end{lemma}

We will also use the following lemma, whose proof is given in Appendix~\ref{app:tail-bounds}:
\begin{lemma}\label{eq:tail-bound}
For every $C_1$ larger than an absolute constant, every integer $k\geq 1$ and every $x\in \C^n$, if the parameter $\mu$ satisfies $\mu\geq ||x_{\nsq\setminus [k]}||_2/\sqrt{k}$, the following conditions hold.  If hashings $\{\{H_{t, s}\}_{s=1}^{R_t}\}_{t=1}^T$ and locations $\{\{\A_{t, s}\}_{s=1}^{R_t}\}_{t=1}^T$ are selected as in Algorithm~\ref{alg:main-sublinear}, lines 6-16,  the sequence $R_1,\ldots, R_T$ satisfies 
\begin{description}
\item[q1] $R_t=C_1 2^t$ for all $t\geq 0$;
\item[q2] $B_t=C_2 (2k)/R_t^2$, 
\end{description}
wheret $C_2>0$ is sufficiently large (as a function of $C_1$), then there exists an event $\E_{small-noise}$ (that depends on $H_{t, s}$ and $\A_{t, s}$) with $\prob[\bar \E_{small-noise}\wedge  \E_{partition}]\leq 1/1000$ (where $\E_{partition}$ is the success event for Lemma~\ref{lm:isolating-partition-construction}) such that  the following conditions hold conditioned on $\E_{small-noise}\cap \E_{partition}$. For $e^{tail, \H}$ defined with respect to $S:=\left\{i\in \nsq: |x_i|>\mu\right\}$  one has for all $t\in [1:T]$ simultaneously 
$||e^{tail, \H}_{S_t}(\{H_{t, s}, \A_{t, s}\}_{s\in [1:R_t]}, x)||_1\leq \frac1{200}||x_{\nsq\setminus k}||_2 \sqrt{k}/R_{t-1}$.
\end{lemma}

\subsection{Construction of a majorizing sequence} \label{sec:maj}
We now construct a sequence of vectors $y^{t, r}\in \R_+^\nsq$, where $t=1,\ldots, T$ and $r=0, 1,\ldots, \lfloor \log_2 R^*\rfloor -3$, which, as we show later, will majorize the actual sequence of residual signals that arise in the execution of our algorithm on the set of head elements $S$ assuming expected behaviour of our estimation primitive. These two properties together will later ensure that the update vectors $\chi^{(r', t')}$ that the SNR reduction loop computes are always supported on $S$.

To define the majorizing sequence, we first let $y^{(0, 0)}_i=R^* \mu$ for all $i\in S$ and $y^{(0, 0)}_i=0$ otherwise. Note that $y^{(0, 0)}$ trivially majorizes every $x$ with the property that $||x||_\infty\leq R^*\cdot \mu$ on $S$. 
The construction of $y^{(r, t)}$ proceeds by induction on $(r, t)$. Given $y^{(r', t')}$, the next signal to be defined is $y^{(r, t)}$, where $(r, t)=(r', t'+1)$ if $t'<T$ and $(r, t)=(r'+1, 1)$ otherwise (note that this notation matches the notation in lines~23-35) of Algorithm~\ref{alg:main-sublinear}, i.e. the SNR reduction loop.  We now define the signal $y^{(r, t)}$ by letting for each $i\in S_t$
\begin{equation}\label{eq:maj-def}
y^{(r, t)}_i:=\max\left\{20 e^{head}_i(\{H_{t, s}\}_{s\in [1:R_t]}, y^{(r',t')})+20 e^{tail, \H}_i(\{H_{t, s}, \A_{t, s}\}_{s\in [1:R_t]}, x)+n^{-\Omega(c)},\frac1{8}\cdot (1/4)^r R^*\mu \right\}
\end{equation}
and letting $y^{(r, t)}_i:=y^{(r', t')}_i$ otherwise. Here $n^{-\Omega(c)}$ corresponds to the (negligible) error term due to polynomial precision of our computations. Note that there are two contributions to $y^{(r, t)}$: one coming from the previous signal in the majorizing sequence, namely $y^{(r', t')}$, and the other coming from the tail of the signal $x$. Also, recall that the quantities $e^{head}$ and $e^{tail}$ (see~\eqref{eq:eh} and~\eqref{eq:et-pi} in Section~\ref{sec:prelim}) used in this section are defined with respect to the set $S$ given by~\eqref{eq:s-def-l1}.

The $\ell_1$ norm of the majorizing sequence satisfies useful decay properties:

\begin{lemma}\label{lm:majorizing-sequence}
For every $\delta\in (0, 1/2)$,  every even $\fc\geq 6$, every $x\in \C^n$, every integer $k\geq 1$, if $\mu\geq ||x_{\nsq\setminus [k]}||_2/\sqrt{k}$, $R^*\geq ||x||_\infty/\mu, R^*=n^{O(1)}$, and $S=\{i\in \nsq: |x_i|\geq \mu\}$, then the following conditions hold. 

If $e^{head}, e^{tail, \H}$ are defined with respect to $S$, hashings $\{H_{t, s}\}$, sets $\{\A_{t, s}\}$ are defined as in Algorithm~\ref{alg:main-sublinear}, parameters $R_t, B_t$ satisfy 
\ifshort{ {\bf (q1)} $R_t=C_1 2^t$ for all $t\geq 0$ and {\bf (q2)}
$B_t=C_2 (2k)/R_t^2$, where $C_2$ is a function of $C_1$
}
\iflong{\begin{description}
\item[q1] $R_t=C_1 2^t$ for all $t\geq 0$, $C_1$ larger than a function of $\delta$;
\item[q2] $B_t= C_2 (2k)/R_t^2$, where $C_2$ is larger than a function of $C_1$ and $\delta$,
\end{description}}
and the sequence $y^{(r, t)}$ is defined as in~\eqref{eq:maj-def}, then there exists an event $\E_{maj}$ with $\prob_{\{\{H_{t, s}\}_{s\in [1:R_t]}\}\}_{t=1}^T}[\E_{maj}]\geq 1-2/25$ such that conditioned on $\E_{maj}$ the set $S$ admits an isolating partition (as per Definition~\ref{def:isolating-partition}) $S=S_1\cup S_2\cup \ldots\cup S_T$, and the following hold.

For every $(r, t)\in [1:T]\times [0:\lfloor \log_4 R^*\rfloor]\cup \{(0, 0)\}$ 
\ifshort{{\bf (A)} for all $q\in [1: t]$ one has $||y^{(r, t)}_{S_q}||_1\leq R^*\cdot (1/4)^{r+1} k\cdot (R_0/R_{q-1})^\delta$, 
{\bf (B)} for all $q\in [t+1: T]$ one has $||y^{(r, t)}_{S_q}||_1\leq R^*\cdot (1/4)^r k\cdot (R_0/R_{q-1})^\delta$ and {\bf (C)}
 $||y_S||_1\leq 2\cdot R^*k(1/4)^r$.
}
\iflong{\begin{description}
\item[(A)] for all $q\in [1: t]$ one has $||y^{(r, t)}_{S_q}||_1\leq R^*\mu\cdot (1/4)^{r+1} \cdot (2k)\cdot (R_0/R_{q-1})^\delta$;
\item[(B)] for all $q\in [t+1: T]$ one has $||y^{(r, t)}_{S_q}||_1\leq R^*\mu\cdot (1/4)^r\cdot (2k)\cdot (R_0/R_{q-1})^\delta$;
\item[(C)] $||y^{(r, t)}_S||_1\leq (2/\delta)\cdot R^*\mu(1/4)^r\cdot (2k)$
\end{description}
}
\end{lemma}
\begin{proof}

By Lemma~\ref{lm:isolating-partition-construction} applied to the set $S$ (recall that $|S|\leq 2k$) we get that conditioned on an event $\E_{partition}$ that occurs with probability at least $1-1/25$ there exists an isolating partition $S=S_1\cup \ldots\cup S_T$. We condition on $\E_{partition}$ in what follows, and the event $\E_{maj}$ that we construct later will be a subset of $\E_{partition}$.

We prove the claims by induction on $(r, t)$.
The {\bf base} is provided by $r=0$ and $t=0$. Indeed, by property {\bf (1)} of an isolating partition (see Definition~\ref{def:isolating-partition}) and the fact that $|S|\leq 2k$ we have for any $q\in [1:T]$
$$
||y||_{S_q}\leq R^*\mu\cdot |S_q|\leq R^*\mu\cdot (2k)\cdot \frac{R_0}{R_{q-1}}2^{-2^{(1-\delta)(q-1)}+1}\leq R^*\mu (2k)\cdot (R_0/R_{q-1})^\delta
$$
since $\delta\in (0, 1)$  by assumption of the lemma and $2^{-2^{(1-\delta)(q-1)}+1}\leq 1$ for all $q\geq 1$.

We now prove the {\bf inductive step}.  There are two cases, depending on whether $t\in [1:T-1]$ or $t=T$. Let $t'=t-1$, $r=r'$ if $t>1$ and $t'=T, r'=r-1$ otherwise. If $t=1, r=0$, then let $t'=0, r'=0$. Note that $(r', t')$ is the element preceding $y^{(r, t)}$ in the majorizing sequence.

We start with an upper bound on the $\ell_1$ norm of $y^{(r', t')}$. Using the inductive hypothesis for $(r', t')$, we get
\begin{equation}\label{eq:3ibrbrg}
\begin{split}
||y^{(r', t')}_S||_1&\leq \sum_{q=1}^{t'} R^*\mu\cdot (2k)\cdot(1/4)^{r'+1}\cdot (R_0/R_{q-1})^\delta+\sum_{q=t'+1}^{\infty} R^*\mu\cdot (2k)\cdot(1/4)^r\cdot (R_{q-1}/R_0)^{-\delta}\\
&\leq R^*\mu\cdot (2k)\cdot(1/4)^{r'}\sum_{q=1}^{\infty} (R_{q-1}/R_0)^{-\delta}\\
&= R^*\mu\cdot (2k)\cdot(1/4)^{r'}\sum_{q=1}^{\infty} 2^{-(q-1) \delta}\text{~~~~~~(since $R_t=C_1 2^t$ by {\bf p1})}\\
&\leq \frac1{2^\delta-1}\cdot R^*\mu\cdot (2k)\cdot(1/4)^{r'}\\
&\leq \frac1{e^{\delta \ln 2}-1}\cdot R^*\mu(1/4)^{r'}\cdot (2k)\\
&\leq (2/\delta)\cdot R^*\mu(1/4)^{r'}\cdot (2k)\text{~~~~~~~(since $e^x-1\geq x$ when $x\leq 1$ and $\ln 2>1/2$)}\\
\end{split}
\end{equation}

 By definition of the majorizing sequence~\eqref{eq:maj-def} the signal $y^{(r, t)}$ is obtained from $y^{(r', t')}$ by modifying the latter on $S_t$. We need to bound the error introduced by head and tail elements of $y^{(r', t')}$ to $y^{(r, t)}_{S_t}$ (see~\eqref{eq:maj-def}).  We now bound both terms. By Lemma~\ref{eq:tail-bound} conditioned on $\E_{small-noise}\cap \E_{partition}$ (defined in the lemma) we have
\begin{equation}\label{eq:tail-bound-12345}
||e^{tail, \H}_{S_t}(\{H_{t, s}, \A_{t, s}\}_{s\in [1:R_{t}]}, x)||_1\leq \frac1{200}\mu k/R_{t-1}.
\end{equation}
By Lemma~\ref{lm:l1b-head} we have 
\begin{equation}\label{eq:head-bound-12345}
||e^{head}_{S_t}(\{H_{t, s}\}_{s\in [1:R_t]}, y^{(r', t')})||_1\leq 40  R_t^{-\delta}||y^{(r', t')}_S||_1
\end{equation}
as long as $S$ admits an isolating partition with respect to the hash functions $\{\{H_{t, s}\}\}$, which it does with probability at least $1-1/25$ by Lemma~\ref{lm:isolating-partition-construction}  (the success event is denoted by $\E_{partition}$).  We now define the event $\E_{maj}$ by letting $\E_{maj}:=\E_{small-noise}\cap \E_{partition}$. Note that $\prob[\E_{maj}]\geq 1-2/25$, as required. We condition on $\E_{maj}$ for the rest of the proof.

We now use the bounds above to prove the result. By definition of the majorizing sequence~\eqref{eq:maj-def} we have 
\begin{equation*}
\begin{split}
y^{(r, t)}_i:=\max\left\{20 e^{head}_i(\{H_{t, s}\}_{s\in [1:R_t]}, y^{(r',t')})+20 e^{tail, \H}_i(\{H_{t, s}, \A_{t, s}\}_{s\in [1:R_t]}, x)+n^{-\Omega(c)},\frac1{8}\cdot (1/4)^r R^*\mu \right\}.
\end{split}
\end{equation*}
so using the bound from~\eqref{eq:head-bound-12345} we get
\begin{equation*}
\begin{split}
||y^{(r, t)}_{S_t}||_1&\leq \sum_{i\in S_t}\left(20 e^{head}_i(\{H_{t, s}\}_{s\in [1:R_t]}, y^{(r',t')})+20 e^{tail, \H}_i(\{H_{t, s}, \A_{t, s}\}_{s\in [1:R_t]}, x)+n^{-\Omega(c)}\right.\\
&\left.+\frac1{8}\cdot (1/4)^r R^*\mu\right)\\
&\leq 400\cdot  R_t^{-\delta} \cdot ||y^{(r',t')}_S||_1+\frac1{8}((1/4)^r R^*\mu)\cdot |S_t|+20||e^{tail, \H}_{S_t}(\{H_{t, s}, \A_{t, s}\}, x)||_1+n^{-\Omega(c)}\\
\end{split}
\end{equation*}
We now substitute ~\eqref{eq:3ibrbrg} together with~\eqref{eq:tail-bound-12345} into the last line above, and use the bound $|S_t|\leq 2k\cdot \frac{R_0}{R_{t-1}}2^{-2^{(1-\delta)(t-1)}+1}\leq 2k \frac{R_0}{R_{t-1}}$ (from the definition of an isolating partition, Definition~\ref{def:isolating-partition}) to obtain
\begin{equation}\label{eq:ibtg43yg84jhGHIOG}
\begin{split}
||y^{(r, t)}_{S_t}||_1&\leq 400\cdot  R_t^{-\delta} ((2/\delta)\cdot R^*\mu \cdot k(1/4)^{r'})+\frac1{8} R^*\mu k(1/4)^{r'}\cdot \frac{R_0}{R_{t-1}}+  \frac1{10}\mu k/R_{t-1}+n^{-\Omega(c)}\\
&\leq \left((2000/\delta)\cdot (R_t/R_{t-1})^{-\delta}\cdot R_0^{-\delta}+\frac{1}{8}+\frac1{10}\right)\cdot R^*\mu \cdot k(1/4)^{r'}\cdot (R_0/R_{t-1})^\delta+n^{-\Omega(c)},\\
&\leq \left((2000/\delta)\cdot R_0^{-\delta}+\frac{1}{8}+\frac1{10}\right)\cdot R^*\mu \cdot k(1/4)^{r'}\cdot (R_0/R_{t-1})^\delta+n^{-\Omega(c)},\\
%&\leq R^*k(1/4)^{r'+1}\cdot (R_0/R_{t-1})^\delta.\\
\end{split}
\end{equation}
where we used the assumption that $r\leq \lfloor \log_4 R^*\rfloor$, so that $R^*\mu(1/4)^r\geq \mu$.

We now note that for every $t\geq 1$, since $R_t=C_1 2^t$  by assumption of the lemma, we have
$$
(2000/\delta)\cdot R_0^{-\delta}=(2000/\delta)\cdot C_1^{-\delta}.
$$
Thus, as long as $C_1\geq (2000\cdot 100 \cdot \delta)^{1/\delta}$, the rhs is upper bounded by $1/100$. Substituting this into ~\eqref{eq:ibtg43yg84jhGHIOG}, we get 
$$
||y^{(r, t)}_{S_t}||_1\leq \left(\frac1{100}+\frac1{8}+\frac1{10}\right)R^* \mu \cdot k(1/4)^{r'}\cdot (R_0/R_{t-1})^\delta+n^{-\Omega(c)}\leq R^* \mu k(1/4)^{r'+1}\cdot (R_0/R_{t-1})^\delta+n^{-\Omega(c)}.
$$

This completes the proof of the inductive step.
\end{proof}
}

\iflong{\subsection{Proof of Theorem~\ref{thm:l1snr}}\label{sec:snr-loop}

We now prove the following lemma, which captures the correctness part of Theorem~\ref{thm:l1snr}. We then put it together with runtime and sample complexity estimates to obtain a proof of Theorem~\ref{thm:l1snr}.

\begin{lemma}\label{lm:main}
For every $\delta\in (0, 1/2)$, every even $\fc\geq 6$, every $x\in \C^n$ if the parameter $\mu$ satisfies $\mu\geq ||x_{\nsq\setminus [k]}||_2/\sqrt{k}$, $R^*= ||x||_\infty/\mu, R^*=n^{O(1)}$, the following conditions hold for the SNR reduction loop in Algorithm~\ref{alg:main-sublinear}. If the hashings $\{H_{t, s}\}$ and locations $\{\A_{t, s}\}$ are chosen as in Algorithm~\ref{alg:main-sublinear}, and parameters satisfy
\begin{description}
\item[q1] $R_t=C_1 2^t$ for all $t\geq 0$, $C_1$ larger than a function of $\delta$;
\item[q2] $B_t=C_2 (2k)/R_t^2$, where $C_2$ larger than a function of $C_1$ and $\delta$,
\end{description}
then the following conditions hold. If $S:=\left\{i\in \nsq: |x_i|>\mu\right\}$, 
then the output $\tilde \chi$ of the $\ell_1$-SNR reduction loop in Algorithm~\ref{alg:main-sublinear} satisfies
$$
||(x-\chi)_S||_1=O_\delta(||x_{\nsq\setminus [k]}||_2\cdot \sqrt{k}),
$$  
and all intermediate $\chi^{r', t'}$ satisfy 
$$
\supp \chi^{r', t'}\subseteq S
$$ 
with probability at least $1-3/25$ over the randomness used in the measurements.
\end{lemma}

\begin{proof}
%We first note that $|S|\leq 2k$, as otherwise more than $k$ elements of $S$ belong to the tail, amounting to more than $\left(||x_{\nsq\setminus k}||_2/\sqrt{k}\right)^2\cdot k>||x_{\nsq\setminus k}||_2^2$, a contradiction. 

We start with an outline of the proof. Throughout the proof we rely on the quantities $e^{head}$ and $e^{tail}$ defined with respect to the set $S=\left\{i\in \nsq: |x_i|>\mu\right\}$ defined in the lemma. 
The proof is by induction on the number of iterations of the SNR reduction loop. We will show that with high probability over the initial measurements the residual signals $x-\chi^{(r, t)}$ are majorized on $S$ by the sequence $y^{(r, t)}$ defined in~\eqref{eq:maj-def}. This lets us argue that {\bf (1)} with high probability over the measurements used for \textsc{EstimateValues} estimation error {\em on the signals $y^{(r, t)}$} is small, and then {\bf (2)} conclude that since $x-\chi^{(r, t)}$ are majorized by $y^{(r, t)}$ on $S$, \textsc{EstimateValues} gives precise estimates for all such residuals. This lets us argue that updates of the residual are always confined to the set $S$, and the residual is still majorized appropriately at the next iteration, giving the inductive proof. In what follows we condition on the event $\E_{maj}$ defined in Lemma~\ref{lm:majorizing-sequence}, which occurs with probability at least $1-2/25$.

{\bf Precision bounds for \textsc{EstimateValues}.} We first prove bounds on the precision of the estimates provided by calls to \textsc{EstimateValues} in the SNR reduction loop of Algorithm~\ref{alg:main-sublinear} (line~32). We have by Lemma~\ref{lm:estimate-l1l2}, {\bf (1a)} applied to the signals $y^{(r, t)}+x_{\nsq\setminus S}$ and the set $S$, that with probability $1-n^{-2}$ over the choice of measurements ${\mathcal M}^{est}$ (lines 17-21) of Algorithm~\ref{alg:main-sublinear} {\bf for any $\chi^{(r, t)}\in \C^n$ such that $\supp \chi^{(r, t)}\subseteq S$ and $x-\chi^{(r, t)}$ is majorized by $y^{(r, t)}$ on $S$} (as per Definition~\ref{def:majorant}), one has that the estimates $w_i$ computed in the call $\Call{EstimateValues}{\chi^{(r, t)}, L, \M^{est}}$ in line~32 satisfy
\begin{equation*}
|w_i-(x-\chi^{(r, t)})_i|\leq 2\quant^{1/5}_r e^{head}(\{H_r\}, x-\chi^{(r, t)})+2\quant^{1/5}_r e^{tail}(\{H_r, a_r\}, x)
\end{equation*}

So in particular by Lemma~\ref{lm:mon} if $x-\chi^{(r, t)}\prec_S y^{(r, t)}$ and $\supp \chi^{(r, t)} \subseteq S$, we get 
\begin{equation*}
\begin{split}
|w_i-(x-\chi^{(r, t)})_i|&\leq 2\quant^{1/5}_r e^{head}(\{H_r\}, x-\chi^{(r, t)})+2\quant^{1/5}_r e^{tail}(\{H_r, a_r\}, x)\\
&\leq 2\quant^{1/5}_r e^{head}(\{H_r\}, y^{(r, t)})+2\quant^{1/5}_r e^{tail}(\{H_r, a_r\}, x).
\end{split}
\end{equation*}

By Lemma~\ref{lm:median-etail-ehead}, {\bf (3)} and {\bf (4)} one has $\quant^{1/5}_r e^{head}(H_r, y)=O(||y_S||_1/B)$ and $\quant^{1/5}_r e^{tail}(H_r, a_r, x)=O(||x_{[n]\setminus S}||_2/\sqrt{B})$ with probability $1-2^{-\Omega(C\log n)}\geq 1-n^{-C/2}$ as long as the constant $C$ is sufficiently large.\footnote{Note that this is the place where we crucially use the notion of majorizing sequences: even though the actual residual signals that arise throughout the update process depend on the measurements $\mathcal{M}^{est}$,  it suffices to invoke Lemma~\ref{lm:median-etail-ehead} on the majorizing sequence $y$, which is fixed and independent of $\mathcal{M}^{est}$.}

Since $B=k/\alpha^2$ by our setting of parameters (line~17 of Algorithm~\ref{alg:main-sublinear}), we have 
$$
O(||y_S||_1/B)\leq O(\alpha) \cdot\frac1{5} \alpha ||y_S||_1/k\leq \frac1{5} \alpha ||y_S||_1/k
$$
and 
$$
O(||x_{[n]\setminus S}||_2/\sqrt{B})\leq O(\sqrt{\alpha}) \cdot \frac1{2}\sqrt{\alpha} ||x_{[n]\setminus S}||_2/\sqrt{k}\leq \frac1{2}\sqrt{\alpha} ||x_{[n]\setminus S}||_2/\sqrt{k}
$$
as long as $\alpha$ is smaller than a constant, and in particular smaller than $\delta^2$ (we will need to set $\alpha$ smaller than $\delta^2$ below to offset the $2/\delta$ factor in the upper bound on the $\ell_1$ norm of $y^{(r, t)}$ in Lemma~\ref{lm:majorizing-sequence}, {\bf (C)}). We thus get that the estimates computed in $\Call{EstimateValues}{\chi, L, \M^{est}}$ in line~32 satisify
\begin{equation}\label{eq:est-bound-12inweigf}
|w_i-(x-\chi^{(r, t)})_i|\leq \frac1{5}\alpha ||y^{(r, t)}_S||_1/k+\frac1{2}\sqrt{\alpha} ||x_{\nsq\setminus S}||_2/\sqrt{k}.
\end{equation}

We have  by Lemma~\ref{lm:majorizing-sequence}, {\bf (C)}, that $||y^{(r, t)}_S||_1\leq (2/\delta) R^*\mu\cdot (2k)(1/4)^r$, and we have by definition of $S$ that $||x_{[n]\setminus S}||_2^2\leq k\cdot ||x_{[n]\setminus S}||_\infty^2+||x_{[n]\setminus [k]}||_2^2\leq k\mu^2+||x_{[n]\setminus [k]}||_2^2\leq 2\mu^2 k$. Substituting these bounds into~\eqref{eq:est-bound-12inweigf}, we get by a union bound over all $i\in [n]$ and all sequences $y^{(r, t)}$ that if $x-\chi^{(r, t)}\prec_S y^{(r, t)}$, then  for all $i\in [n]$ and all $(r, t)$ one has 
\begin{equation}\label{eq:est-bound}
\begin{split}
|w_i-(x-\chi^{(r, t)})_i|&\leq \frac1{5}\alpha ||y^{(r, t)}_S||_1/k+\frac1{2}||x_{\nsq\setminus S}||_2/\sqrt{k}\\
&\leq \frac1{5}\alpha (4/\delta) R^*\mu\cdot  (1/4)^r+\frac1{2}\sqrt{\alpha}\sqrt{2}\mu\\
&\leq \sqrt{\alpha}\left(R^*\mu (1/4)^r+\mu\right),
\end{split}
\end{equation}
where we used the assumption that $\alpha\leq \delta^2$ to obtain the last inequality.

Equipped with the bounds on estimation quality in~\eqref{eq:est-bound}, we now give the proof of the theorem. The proof is by induction on $(r, t)$. We prove that for every $(r, t)\in [1:T]\times [0:\lfloor \log_4 R^*\rfloor]\cup \{(0, 0)\}$ 
\begin{description}
\item[(A)] $y^{(r, t)}$ majorizes $x-\chi^{(r, t)}$ on $S$;
\item[(B)] $\chi^{(r, t)}_{\nsq\setminus S}\equiv 0$.
\end{description}

The {\bf base} is provided by $(r, t)=(0, 0)$, where $y^{(0, 0)}_i=R^*\mu$ for $i\in S$ and $y^{(0, 0)}_i=0$ otherwise. Since $||x||_\infty\leq R^*\mu$ by assumption of the lemma and $\chi^{(0, 0)}=0$ in Algorithm~\ref{alg:main-sublinear},  the base of the induction holds. We now prove the {\bf inductive step}. Let $t'=t-1$, $r=r'$ if $t>1$ and $t'=T, r'=r-1$ otherwise. If $t=1, r=0$, then let $t'=0, r'=0$. Note that $(r', t')$ is the element preceding $y^{(r, t)}$ in the majorizing sequence.

Since $x-\chi^{(r', t')}\prec_S y^{(r', t')}$ and $\supp \chi^{(r', t')}\subseteq S$ by the inductive hypothesis, we have by~\eqref{eq:est-bound}, 
\begin{equation}\label{eq:est-concrete}
|w_i-(x-\chi^{(r', t')})_i|\leq \sqrt{\alpha}\left(R^* (1/4)^{r'} \mu+\mu\right),
\end{equation}
where $\alpha$ is smaller than an absolute constant (see line~17 of Algorithm~\ref{alg:main-sublinear}).

We first prove part {\bf (B)} of the inductive step. Since only elements with $|w_i|\geq (1/16)R^*\mu (1/4)^r$ are updated (by the pruning step in line~33 of Algorithm~\ref{alg:main-sublinear}), for all such $i$ we have by triangle inequality using~\eqref{eq:est-concrete} that 
\begin{equation}\label{eq:triangle}
|(x-\chi^{(r', t')})_i|\geq (1/16)R^*\mu (1/4)^r-\sqrt{\alpha}\left(R^* (1/4)^r \mu+\mu\right)\geq (1/32)R^*\mu (1/4)^r,
\end{equation}
where we used the assumption that $\alpha$ is smaller than a sufficiently small absolute constant. Since the upper bound for $r$ in the SNR reduction loop is $\lfloor \log_4 R^* \rfloor-3$ in the SNR reduction loop, we have $(1/32)R^*\mu (1/4)^r\geq 2\mu>\mu$ for all such $r$.
Since $\supp \chi^{(r', t')}\subseteq S$ by the inductive hypothesis, this means that the output $\chi'$ of the call to \textsc{EstimateValues} is such that any $i\in \nsq$ with $\chi'_i\neq 0$ belongs to $S$. We have shown that $\supp \chi^{(r, t)}\subseteq \supp \chi^{(r', t')}\cup \supp \chi'\subseteq S$, proving part {\bf (B)} of the inductive step. 

We now prove part {\bf (A)} of the inductive step, i.e. prove that $x-\chi^{(r,t)}=x-\chi^{(r', t')}-\chi'$ is majorized by $y^{(r, t)}$ (defined by \eqref{eq:maj-def}). 

\paragraph{Bounding elements reported in $L$.} 
We first consider $i\in L\cap \supp \chi'$, i.e. elements that were reported in $L$ and estimated as being above the threshold. For such $i\in L\cap \supp \chi'$ we have by ~\eqref{eq:est-concrete}
\begin{equation*}
\begin{split}
|(x-\chi^{(r, t)})_i|=|(x-\chi^{(r', t')}-\chi')_i|&\leq \sqrt{\alpha}\left(R^* (1/4)^r \mu+\mu\right)< (1/32)R^*\mu (1/4)^r.
\end{split}
\end{equation*}
At the same time for such elements ($i\in L\cap \supp \chi'$) we have by~\eqref{eq:triangle} 
$|(x-\chi^{(r', t')})_i|\geq (1/32)R^*\mu (1/4)^r$.
This means that for all $i\in L\cap \supp \chi'$ one has 

\begin{equation}\label{eq:monontone}
|(x-\chi^{(r, t)})_i|\leq |(x-\chi^{(r', t')})_i|,
\end{equation}
 as well as 
 \begin{equation}\label{eq:absolute}
 |(x-\chi^{(r', t')})_i|\leq (1/32)R^*\mu (1/4)^r.
 \end{equation}

At the same time for $i\in \nsq$ such that $\chi'_i=0$ we have
\begin{equation}\label{eq:ub-absolute}
|(x-\chi^{(r', t')})_i|=|(x-\chi^{(r, t)})_i|\leq (1/16)R^*\mu (1/4)^r+\sqrt{\alpha}(R^* (1/4)^r \mu+\mu)\leq (1/8)R^*\mu (1/4)^r
\end{equation}
as long as $\alpha$ is smaller than an absolute constant.

\paragraph{Bounding elements not reported in $L$.} Let $x':=x-\chi^{(r', t')}$ to simplify notation. If an element $i\in S_t$ is not reported in any of the calls to \textsc{LocateSignal} (i.e. does not belong to $L$), then by Corollary~\ref{cor:loc} it
satisfies 
\begin{equation}\label{eq:part1}
\begin{split}
|x'_i|&\leq 20 e^{head}_i(\{H_{t, s}\}_{s\in [1:R_t]}, x')+20 e^{tail, \H}_i(\{H_{t, s}, \A_{t, s}\}_{s\in [1:R_t]}, x)+n^{-\Omega(c)}\\
&\leq 20 e^{head}_i(\{H_{t, s}\}_{s\in [1:R_t]}, y^{(r, t)})+20 e^{tail, \H}_i(\{H_{t, s}, \A_{t, s}\}_{s\in [1:R_t]}, x)+n^{-\Omega(c)},
\end{split}
\end{equation}
where we used Lemma~\ref{lm:mon} to upper bound error induced by head elements of $x'$ by error induced by head elements of $y^{(r, t)}$, which majorizes $x'$ on $S$ by the inductive hypothesis.

\paragraph{Putting it together.} Recall that by~\eqref{eq:maj-def}  the signal $y^{(r, t)}$ is defined by letting for each $i\in S_t$
\begin{equation}\label{eq:maj-def-p}
y^{(r, t)}_i:=\max\left\{20 e^{head}_i(\{H_{t, s}\}_{s\in [1:R_t]}, y^{(r',t')})+20 e^{tail}_i(\{H_{t, s}, \A_{t, s}\}_{s\in [1:R_t]}, x)+n^{-\Omega(c)},\frac1{8}\cdot (1/4)^r R^*\mu \right\}
\end{equation}
and letting $y^{(r, t)}_i:=y^{(r', t')}_i$ otherwise.

We now have that any element $i\in S_t$ that is not reported in any of the calls to \textsc{LocateSignal} $x'_i$ satisfies~\eqref{eq:part1}, which is upper bounded by the first argument in the maximum above. By~\eqref{eq:ub-absolute} together with \eqref{eq:absolute} we have $|(x'-\chi')_i|\leq (1/8)R^*\mu (1/4)^r$ for all $i\in S_t$, which is upper bounded by the second term in the maximum above. Finally, for any $i\in \nsq$ (not necessarily in $S_t$) by ~\eqref{eq:monontone} we have $|(x'-\chi')_i|\leq |x'_i|$, so $|(x'-\chi')_i|\leq |x'_i|\leq y^{(r', t')}_i=y^{(r, t)}_i$ for such $i$ as well (note that we are also using the fact that $i\in S$ necessarily for all $i$ with $\chi'_i\neq 0$, by part {\bf (B)} of the inductive step, which we proved already).  This completes the proof of part {\bf (A)} of the inductive step, and the proof of the lemma.

Note that we conditioned on the event $\E_{maj}$ defined in Lemma~\ref{lm:majorizing-sequence}, which occurs with probability at least $1-2/25$, as well as a high probability ($1-1/\poly(n)$) success event for \textsc{EstimateValues}. Thus, success probability is at least $1-3/25$ by a union bound, as required.
\end{proof}

We can now give a proof of Theorem~\ref{thm:l1snr}, the main technical result of this section. We restate the theorem here for convenience of the reader:

{\em \noindent {\bf Theorem~\ref{thm:l1snr}} (Restated)
For any  $\delta\in (0, 1/2)$, for any $x\in \C^n$, any integer $k\geq 1$, if $\mu^2=\err_k^2(x)/k$ and $R^*\geq ||x||_\infty/\mu, R^*=n^{O(1)}$, the following conditions hold for the set $S:=\{i\in \nsq: |x_i|>\mu\}\subseteq \nsq$.

Then the SNR reduction loop of Algorithm~\ref{alg:main-sublinear} (lines~25-36) returns $\tilde \chi$ such that 
\begin{equation*}
\begin{split}
&||(x-\tilde \chi)_S||_1=O_\delta(\mu k)\\
&\supp \tilde \chi\subseteq S
\end{split}
\end{equation*}
with probability at least $1-3/25$ over the internal randomness used by Algorithm~\ref{alg:main-sublinear}. The sample complexity is $O_\delta(k\log n)$. The runtime is bounded by $O_\delta(k\log^{4+2\delta} n)$.}

\begin{proof}
Correctness follows by Lemma~\ref{lm:main} and setting of parameters in Algorithm~\ref{alg:main-sublinear}. It remains to bound the sample and runtime complexity. For each $t=1,\ldots, T$ we take $B_t$ measurements using a filter of sharpness $\fc=O(1)$, so the total sample complexity is 
\begin{equation*}
\begin{split}
\sum_{t=1}^T O( B_t) |\H|\cdot |\A_{t, s}| \cdot R_t&=\sum_{t=1}^T O( B_t) (\log n/\log\log n) (\log\log n) \cdot R_t\\
&\leq O(C_2 )((2k \log n)/R_t^2)\cdot R_t\\
&=O(C_2 k) \log n\cdot \sum_{t=1}^{T} R_t=O(k \log n),
\end{split}
\end{equation*} 
where we used the fact that $C_1$ is an absolute constant and $C_2$ as prescribed by Lemma~\ref{lm:isolating-partition-construction}, as well as the setting $|\A_{t, s}|=O(\log\log n)$ and $\H=O(\log_{\Delta} n)=O(\log n/\log\log n)$) Algorithm~\ref{alg:main-sublinear}.

{\bf Runtime.} We start by bounding the runtime for the SNR reduction loop
\begin{itemize}
\item Each call to \textsc{LocateSignal} costs $O(F B_t \log^2  n+||\chi||_0\log^2 n)$ by Lemma~\ref{lm:loc}.

The total cost for calls to \textsc{LocateSignal} in a single iteration (i.e. one value of $r$) is hence bounded by 
\begin{equation*}
\begin{split}
&O\left(\sum_{t=1}^T \sum_{s=1}^{R_t} B_t \log^2  n+(\max_{r', t'} ||\chi^{r', t'}||_0)\log^2 n\right)\\
&=O\left(\sum_{t=1}^T R_t B_t\right) \log^2  n+O(\sum_{t=1}^T  R_t k\log^2 n)\text{~~~~~(since $\supp \chi^{r', t'}\subseteq S$ for all $r', t'$ by Lemma~\ref{lm:main})}\\
&=\sum_{t=1}^T O(k/R_t) \log^2  n+O(R_T) ||\chi||_0\log^2 n \text{~~~~~(since $R_t$ increase geometrically by setting of parameters in line~7)}\\
&=O(k \log^2  n)+||\chi||_0\log^{3+2\delta} n\\
&=k\log^{3+2\delta} n,\\
\end{split}
\end{equation*}
where we used the fact that 
$$
R_T=C_1 2^T= C_1 2^{\frac1{1-\delta}\log_2 \log k+O(1)}=O(\log_2^{1/(1-\delta)} k)=O(\log_2^{1+2\delta} k)
$$
by setting of parameters in Algorithm~\ref{alg:main-sublinear} and the fact that $1/(1-\delta)\leq 1+2\delta$ for $\delta \in (0, 1/2)$. Finally, accounting for $O(\log n)$ iterations of the SNR reduction loop over $r$, we obtain a bound of $k\log^{4+2\delta} n$, as claimed.

\item Each call to \textsc{EstimateValues} costs 
$O(F B \cdot \log n\cdot C\log n +(\max_{r', t'} ||\chi^{r', t'}||_0)\cdot  \log n\cdot C\log n)$
by Lemma~\ref{lm:estimate-l1l2}.  The total runtime over $O(\log n)$ iterations of the SNR reduction loop is hence  $O(k\log^3 n)$.

\end{itemize}

Summing the contributions, we get runtime $k\log^{4+2\delta} n$, as required. Success probability follows from the success probability of Lemma~\ref{lm:main}.

\end{proof}

}

%% file: app-nice-partitions.tex
%!TEX root = ./ft-hd-wrap.tex 
\section{Proof of Lemma~\ref{lm:isolating-partition-construction}}\label{app:isolating-partition}
We restate the lemma for convenience of the reader:

{\noindent {\bf Lemma~\ref{lm:isolating-partition-construction}} (Restated) \em
For every integer $k\geq 1$, every $S\subseteq \nsq, |S|\leq k$,  every $\delta\in (0, 1/2)$, if the parameters $B_t, R_t$ are selected to satisfy {\bf (p1)} $R_t=C_1\cdot 2^t$ and {\bf (p2)} $B_t\geq C_2\cdot k/R_t^2$ for every $t\in [0:T]$, where $C_1$ is a sufficiently large constant and $C_2$ is sufficiently large as a function of $C_1$ and $\delta$, then the following conditions hold.  

With probability at least $1-1/25$ over the choice of hashings $\{\{H_{t, s}\}_{s\in [1:R_t]}\}_{t=1}^T$  Algorithm~\ref{alg:partition} terminates in $T=\frac1{1-\delta}\log_2 \log (k+1)+O(1)$ steps. When the algorithm terminates, the output partition $\{S_j\}_{j=1}^T$  is isolating as per Definition~\ref{lm:isolating-partition-construction}.}

We will use 
\begin{theorem}[Chernoff bound]\label{thm:chernoff-bounds}
Let $X_1,\ldots, X_n$ be independent Bernoulli random variables, let $\mu:=\expect[\sum_{i=1}^n X_i]$. Then 
for any $\eta>1$ one has 
$\prob[\sum_{i=1}^n X_i>(1+\eta)\mu]\leq e^{-\mu \eta/3}$.
\end{theorem}

The following basic technical claim is crucial to our analysis (the short proof is given in Appendix~\ref{app:basic}):
\begin{claim}\label{cl:max-expr}
For every $C_1, C_2> 0, \delta\in (0, 1)$ there exists $C_3$ such that for every $C_4\geq C_3$ one has 
$\frac1{C_4}2^{C_1 t}\cdot 2^{-C_2 2^{(1-\delta)t}+1}\leq 1$ all $t\geq 0$.
\end{claim}

Equipped with the technical claim above, we can now argue that Algorithm~\ref{alg:partition} constructs an isolating partition of any set $S\subseteq \nsq$ that satisfies $|S|\leq k$ with at least high constant probability and prove Lemma~\ref{lm:isolating-partition-construction}.

\begin{proofof}{Lemma~\ref{lm:isolating-partition-construction}}
The proof proceeds in four steps.
In {\bf Step (1)} we state a set of inductive claims that we will prove, then in {\bf Step (2)} argue that the inductive claims imply that Algorithm~\ref{alg:partition} terminates in $T=\frac1{1-\delta}\log_2 \log (k+1)+O(1)$ iterations, then in {\bf Step (3)} argue that the inductive claims imply that output partition is isolating and finally in {\bf Step (4)} prove the inductive claims (this step corresponds to the bulk of the proof). 

 {\bf Step (1)} Our argument proceeds inductively for $t=1,2,\ldots$, and we think of sampling the hashings $\{H_{t, s}\}_{s=1}^{R_t}$ independently at each step $t$.

For each $t\geq 1$ let $k_t:=|S_t^t|$. We show by induction on $t\geq 1$ that there exists a sequence of nested events ${\mathcal E}_1 \supseteq {\mathcal E}_2 \supseteq \ldots$ such that for all $t\geq 1$
\begin{description}
\item[(1)] event $\E_t$ depends only on the randomness up to time $t$;
\item[(2)] $\prob[{\mathcal E}_t]\geq 1-\frac{3}{100}\sum_{t'=1}^{t-1} \frac1{R_{t'}}$;
\item[(3)] conditioned on ${\mathcal E}_t$ one has  $k_t\leq k\cdot \frac{R_0}{R_{t-1}}2^{-2^{(1-\delta)(t-1)}+1}$. 
\end{description}

 {\bf Step (2)}  We now show that {\bf (3)} implies that the algorithm terminates in $T=\frac1{1-\delta}\log_2 \log (k+1)+O(1)$ steps. Indeed, by {\bf (3)} we have
$$
|S^T_T|\leq k\cdot \frac{R_0}{R_{T-1}} 2^{-2^{(1-\delta)(T-1)}+1}.
$$
Substituting $T=\frac1{1-\delta}(\log_2 \log_2 (k+1)+C)$, we get
$$
|S^T_T|\leq  k\cdot 2^{-2^{(1-\delta)(T-1)}+1}\leq k\cdot 2^{-\frac1{2} 2^{(1-\delta)T}+1}\leq k\cdot 2^{-\frac1{2} 2^C\cdot \log_2 (k+1) +1}\leq 4(k+1)^{1-2^{C-1}}<1
$$
as long as $C\geq 3$.

Also, {\bf (2)} implies that the algorithm terminates with probability at least 
\begin{equation*}
\begin{split}
\prob[\E_T]&\geq 1-\frac{3}{100}\sum_{t=1}^{T-1} \frac1{R_t}\\
&\geq 1-\frac{3}{100}\sum_{t=1}^{\infty} \frac1{C_1 2^t}\\
&\geq  1-\frac1{25}
\end{split}
\end{equation*}
where we used the assumption that $\{S_j\}_{j=1}^T$ satisfies property {\bf (1)} of an isolating partition (see Definition~\ref{lm:isolating-partition-construction}) as well as $R_t=C_1\cdot 2^t$ and $C_1$ is larger than an absolute constant.

 {\bf Step (3)}  We now that given the inductive claims from {\bf Step (1)}, the returned partition $S_1^T\cup \ldots\cup S_T^T$ satisfies the definition of a $\delta$-isolating partition (Definition~\ref{lm:isolating-partition-construction}).  We need to prove that 
\begin{enumerate}
\item $|S_t^T|\leq k\cdot \frac{R_0}{R_{t-1}}2^{-2^{(1-\delta)\cdot (t-1)}+1}$;
\item no element of $S_t^T$ is $R_t^{-3}$-crowded by $S_t^T$ under any of $\{H_{t, s}\}_{s=1}^{R_t}$;
\item no element of $S_t^T$ $R_t$-collides with a $\delta$-bad element for $S_t^T$ under any of $\{H_{t, s}\}_{s=1}^{R_t}$.
\end{enumerate}

To prove the first property, we note that the sizes of sets $S^{t'}_t$ are non-increasing in $t'\geq t$ for every $t$, as $S^{t'}_t\supseteq S^t_t$ (by line~7 of Algorithm~\ref{alg:partition}). Conditional on $\E_T$ we thus have
$$
|S^T_t|\leq |S^t_t|\leq k\cdot \frac{R_0}{R_{t-1}} 2^{-2^{(1-\delta)(t-1)}+1}
$$
for all $t\geq 1$, as required. 

For the second property, note that  no  element of $S_t^{t+1}$ is $R_t^{-3}$-crowded by $S_t^t$ under any $\{H_{t, s}\}_{s=1}^{R_t}$ by construction of $S_t^{t+1}$ (line~7 of Algorithm~\ref{alg:partition}). Since $S_t^{t+1}\subseteq S_t^t$, this means that no  element of $S_t^{t+1}$ is $R_t^{-3}$-crowded by $S_t^{t+1}$ under any $\{H_{t, s}\}_{s=1}^{R_t}$, and since $S_t^T\subseteq S_t^{t+1}$ this also means that no  element of $S_t^T$ is $R_t^{-3}$-crowded by $S_t^T$ under any $\{H_{t, s}\}_{s=1}^{R_t}$, so property 2 is satisfied.

For the third property, note that  no  element of $S_t^{t+1}$  $R_t$-collides with a $\delta$-bad element for $S_t^t$  under any $\{H_{t, s}\}_{s=1}^{R_t}$ by construction of $S_t^{t+1}$ (line~7 of Algorithm~\ref{alg:partition}). Since $S_t^T\subseteq S_t^{t+1}$, this means that no  element of $S_t^T$  $R_t$-collides with a $\delta$-bad element for $S_t^t$ under any $\{H_{t, s}\}_{s=1}^{R_t}$. Finally, note that since $S_t^T\subseteq S_t^t$, any element that is $\delta$-bad for $S^T_t$  is also  $\delta$-bad for $S^t_t$ by Definition~\ref{def:bad}. This shows that no  element of $S_t^T$  $R_t$-collides with a $\delta$-bad element for $S_t^T$ under any $\{H_{t, s}\}_{s=1}^{R_t}$ and establishes property 3 above. This completes the proof that the constructed partition $\{S^T_t\}$ is isolating.

{\bf Step (4)} In what follows we construct the events $\E_t, t=1,\ldots, T$ and prove properties {\bf (1)-(3)} above by induction on $t=1,\ldots, T$. The proof is by induction on $t$. 

\paragraph{Base:$t=1$} We let $S^1_1:=S$, so that the base is trivial (we let $\mathcal{E}_1$ be the trivial event that occurs with probability $1$). 

\paragraph{Inductive step: $t\to t+1$} Suppose that $|S_t^t|=k_t\leq k\cdot \frac{R_0}{R_{t-1}}2^{-2^{(1-\delta)(t-1)}+1}$.  We first bound the expected size of $U_t$ and $V_t$ conditional on ${\mathcal E}_t$,  and then put these bounds together to obtain a proof of the inductive step. 

\paragraph{Bounding the number of crowded elements in $S^t_t$ (size of $V_t$)}
 For each element $i\in [n]$ and every scale $q\geq 0$ we have, letting $H:=H_{t, s}, \pi:=\pi_{t, s}$ and $h:=h_{t, s}$ to simplify notation (recall that $h(i)=\text{round}((B/n)\pi(i))$; see section~\ref{sec:prelims-permutations}), 
\begin{equation}\label{eq:coll-bounds}
\begin{split}
\expect_H\left[\left|\pi(S^t_t\setminus \{i\})\cap \B(\pi(i), \frac{n}{B_t}2^q)\right|\right]\leq 4\cdot 2^q|S^t_t|/B_t\leq 4\cdot 2^qk_t/B_t,
\end{split}
\end{equation}
where we used the fact that $|S^t_t|\leq k_t$ by the inductive hypothesis, as well as Lemma~\ref{lemma:limitedindependence}. Thus by Markov's inequality for every $\lambda>0$
$$
\prob_H\left[\left|\pi(S^t_t\setminus \{i\})\cap \B\left(\pi(i), \frac{n}{B_t}2^q\right)\right|>\lambda \cdot 2^{2q}\right]\leq 4\lambda^{-1}2^{-q}k_t/B_t.
$$
By a union bound over all scales $q\geq 0$ (i.e. summing the rhs of the bound above over all scales $q\geq 0$) we conclude that 
\begin{equation}\label{eq:239htjeggg}
\prob_H[\text{$i$  is $\lambda$-crowded under hashing $H$}]\leq \sum_{q\geq 0} 4\lambda^{-1}2^{-q}k_t/B_t=8 \lambda^{-1}k_t/B_t.
\end{equation}

We thus have for every $i\in \nsq$ and a random hashing hashing $H=(\pi, B_t, G)$
\begin{equation}\label{eq:23gj4ggfFF}
\begin{split}
\prob_{H}[\text{$i$ is $R_t^{-3}$-crowded}]&\leq 8(R_t^3)\cdot k_t/B_t\text{~~~~~~(by~\eqref{eq:239htjeggg} with $\lambda=R_t^{-3}$)}\\
&\leq \frac{8}{C_2}(R_t^5)\cdot \frac{R_0}{R_{t-1}} 2^{-2^{(1-\delta)(t-1)}+1},
\end{split}
\end{equation}
where we used the bound $k_t\leq k\cdot \frac{R_0}{R_{t-1}} 2^{-2^{(1-\delta)(t-1)}+1}$ provided by the inductive hypothesis and the assumption that $B_t\geq C_2\cdot k/R_t^2$  by assumption {\bf p2} of the lemma to go from the first line to the second.

We thus have by a union bound over $R_t$ hashings $\{H_{t, s}\}_{s\in [1:R_t]}$, for every $i\in S_t^t$
\begin{equation}\label{eq:23gj4ggfFFihgeffg}
\begin{split}
\prob_{\{H_{t, s}\}_{s\in [1:R_t]}}[i\in V_t]&=\prob_{\{H_{t, s}\}_{s\in [1:R_t]}}[\text{$i$ is $R_t^{-3}$-crowded under at least one }H_{t, s}]\\
&\leq \frac{8}{C_2}(R_t^6)\cdot \frac{R_0}{R_{t-1}} 2^{-2^{(1-\delta)(t-1)}+1} \text{~~~~~(by a union bound applied to~\eqref{eq:23gj4ggfFF})}\\
&\leq \frac{16 C_1}{C_2}(R_t^5)\cdot 2^{-2^{(1-\delta)(t-1)}+1}\text{~~~~~~~(since $R_t=C_1\cdot 2^t$ by {\bf p1})}\\
\end{split}
\end{equation}
Using the upper bound on the size of $S^t_t$ given by the inductive hypothesis again, we obtain
\begin{equation}\label{eq:v-t}
\begin{split}
\expect_{\{H_{t, s}\}_{s\in [1:R_t]}}[|V_t|]&\leq \sum_{i\in S_t^t} \prob_{\{H_{t, s}\}_{s\in [1:R_t]}}[i\in V_t]\\
&\leq |S_t^t|\cdot  \prob_{\{H_{t, s}\}_{s\in [1:R_t]}}[i\in V_t]\text{~~~~~~~~~~~~(for any $i\in S_t^t$)}\\
&\leq \left(k\cdot \frac{R_0}{R_{t-1}} 2^{-2^{(1-\delta)(t-1)}+1}\right)\cdot  \prob_{\{H_{t, s}\}_{s\in [1:R_t]}}[i\in V_t]\text{~~~~~~~~~~~~(by the inductive hypothesis)}\\
&\leq \left(k\cdot \frac{R_0}{R_{t-1}} 2^{-2^{(1-\delta)(t-1)}+1}\right)\cdot \frac{16C_1}{C_2}\cdot R_t^5 2^{-2^{(1-\delta)(t-1)}+1}\text{~~~~~~~~~~(by \eqref{eq:23gj4ggfFFihgeffg})}\\
&\leq k\cdot \frac{32 C_1 R_0 R_t^5}{C_2\cdot R_{t-1}} \cdot 2^{-2^{(1-\delta)(t-1)+1}+1}\\
&\leq k\cdot \frac{64 C_1 R_0 R_t^4}{C_2} \cdot 2^{-2^{(1-\delta)t+\delta}+1}\\
&\leq k\cdot \frac{64 C_1 R_0 R_t^4}{C_2} \cdot 2^{-2^{(1-\delta)t}+1}\cdot 2^{-(2^\delta-1)2^{(1-\delta)t}}\\
&\leq \left(k\frac1{100 R_t^2}2^{-2^{(1-\delta)t}+1}\right) \cdot \xi_t\\
\end{split}
\end{equation}
where 
$$
\xi_t=\frac{6400 C_1 R_0 R_t^6}{C_2} \cdot 2^{-(2^\delta-1)2^{(1-\delta)t}}\leq \frac{6400 C_1^8}{C_2} 2^{6t}\cdot 2^{-\delta 2^{(1-\delta)t}},
$$
where we used the assumption that $R_t=C_1 2^t$ for a constant $C_1>0$, and the bound $e^x-1\geq x$ for all $x\geq 0$.

It remains to note that for every $\delta>0$, if $C_2$ is sufficiently large (depending on $C_1$ and $\delta$), we get that $\xi_t<1$ for all $t\geq 1$. Formally this follows by Claim~\ref{cl:max-expr}.

\paragraph{Bounding the number of bad elements ($\text{Bad}_t$).} Recall (Definition~\ref{def:bad}) that an element $a$ of $S$ is {\em bad for $S_t$} with respect to a partition $S=S_1\cup S_2\cup\ldots\cup S_T$ and hashings $\{\{H_{t, s}\}_{s=1}^{R_t}\}_{t=1}^T$ if $a$ participates in an $R_t$-collision with at least one element of $S_t$ under more than a $R_t^{-\delta}$ fraction of hashings $H_{t, 1},\ldots, H_{t, R_t}$. We now upper bound the probability that a given $i$ is bad. 

 For any $i\in \nsq$  the probability that $i$ $R_t$-collides with a given element $j\in S^t_t$ under a random hashing $H$ is upper bounded as follows. 
First recall that  that $\pi(i)=\sigma(i-q)$ for all $i\in \nsq$, so 
\begin{equation*}
\begin{split}
\prob_\pi[\text{$i$ and $j$ participate in an $R_t$-collision}]&=\prob_\pi[|\pi(i)-\pi(j)|_\circ\leq (n/B_t) R_t]\\
&=\prob_\sigma[|\sigma(i-j)|_\circ\leq (n/B_t) R_t]\leq 4 R_t/B_t,
\end{split}
\end{equation*}
where we used Lemma~\ref{lemma:limitedindependence} to obtain the last bound.

A union bound over all $j\in S^t_t$ then gives that for every $i\in S$ 
\begin{equation}\label{eq:prob-coll-i}
\prob_H[\text{$i$ $R_t$-collides with at least one element of $S^t_t$ under $H$}]\leq  4 R_t (k_t/B_t).
\end{equation} 

For each $s=1,\ldots, R_t$ let $X_s=1$ if $i$ $R_t$-collides with an element of $S^t_t$ under hashing $H_{t, s}$ and $X_s=0$ otherwise. We first bound $\expect[X_s]$, and then apply Chernoff bounds to $X:=\sum_{s=1}^{R_t} X_s$ to bound the number of bad elements in $S$  with respect to $S_t^t$ at step $t$. In order to bound the expected number of bad elements, it would be sufficient to bound $\prob[X>R_t^{1-\delta}]$. Instead, we will upper bound a slightly larger quantity that will be useful for upper bounding the expected number of elements that collide with a bad element (which is what we need to bound ultimately). Specifically, for any $s^*\in [1:R_t]$ we let 
$X_{-s^*}:=\sum_{s=1, s\neq s^*}^{R_t} X_s$. Note that $\prob[X>R_t^{1-\delta}]\leq \prob[X_{-s^*}\geq R_t^{1-\delta}]$ for any $s^*$, and it is the latter quantity that we bound now. We now have for every $s^*\in [1:R_t]$
\begin{equation*}
\begin{split}
\expect_{H_{t, s}}\left[X_{-s^*}\right]&\leq \sum_{s=1, s\neq s^*}^{R_t} \expect_{H_{t, s}}[X_s]\\
&\leq \sum_{s=1}^{R_t} 4 R_t\cdot  (k_t/B_t)\text{~~~~~(by~\eqref{eq:prob-coll-i} and definition of $X_s$)}\\
&\leq 4 R_t^2\cdot  (k_t/B_t)\\
\end{split}
\end{equation*}

By the inductive hypothesis we have $k_t\leq k\cdot \frac{R_0}{R_{t-1}}2^{-2^{(1-\delta)(t-1)}+1}$ and $B_t\geq C_2\cdot k/R_t^2$  by assumption {\bf p2} of the lemma. Substituting these bounds on the last line of the equation above, we obtain
\begin{equation*}
\begin{split}
&4 R_t^2\cdot \left[k\cdot \frac{R_0}{R_{t-1}}2^{-2^{(1-\delta)(t-1)}+1}\right]\cdot \left[C_2\cdot k/R_t^2\right]^{-1}\\
&=\frac1{10R_t}\left[\frac{40}{C_2}R_t^{5}\cdot \frac{R_0}{R_{t-1}}\cdot 2^{-2^{(1-\delta)(t-1)}+1}\right]\\
&\leq \frac1{10R_t}\left[\frac{40}{C_2}R_t^5 \cdot 2^{-2^{(1-\delta)(t-1)}+1}\right]\text{~~~~(since $R_0/R_{t-1}=2^{-(t-1)}\leq 1$ for all $t\geq 1$)}\\
&\leq \frac1{10 R_t} \text{~~~~~~(by Claim~\ref{cl:max-expr}, as long as $C_2$ is larger than a constant that may depend on $C_1$ and $\delta$)}.
\end{split}
\end{equation*}

Let $\mu_{-s^*}:=\expect\left[\sum_{s=1, s\neq s^*}^{R_t} X_s\right]$, and note that  by the bound on $\expect[X_s]$ above we have $\mu_{-s^*}\leq 1/10$ (we omit the subscript in $\mu_{-s^*}$ in what follows). Since the permutations $H_{t, s}$ were chosen independently, we have by Chernoff bounds (Theorem~\ref{thm:chernoff-bounds}) for any $\eta>1$
$\prob\left[X_{-s^*}\geq R_t^{1-\delta}\right]=\prob\left[X_{-s^*}\geq (1+\eta)\mu\right]$ with $\eta=R_t^{1-\delta}/\mu-1$.  Since $R_t^{1-\delta}>1$ by assumption of the lemma and $\mu\leq 1/10$, we have $R_t^{1-\delta}/\mu-1\geq R_t^{1-\delta}/(2\mu)$.
We thus have 
\begin{equation}\label{eq:prob-s}
\prob\left[X_{-s^*}\geq R_t^{1-\delta}\right]\leq e^{-R_t^{1-\delta}/6},
\end{equation} and by linearity of expectation 
\begin{equation}\label{eq:bb-t}
\expect_{H_{t, 1},\ldots, H_{t, R_t}}[|\text{Bad}_t|]\leq k\cdot e^{-R_t^{1-\delta}/6}.
\end{equation}

{\bf Bounding the number of elements $i\in S_t^t$ that $R_t$-collide with $\text{Bad}_t$.}  Consider $i\in S^t_t$.  For  fixed $s\in [1:R_t]$ let $Q_s(i)\subseteq S$ denote the set of elements that $i$ $R_t$-collides with under hashing $H_{t, s}$. We have by~\eqref{eq:coll-bounds}
$$
\expect_{H_{t, s}}[|Q_s(i)|]\leq \expect_{H_{t, s}}\left[\left|\pi(S\setminus \{i\})\cap \B(\pi(i), \frac{n}{B_t}\cdot R_t)\right|\right]\leq 4\cdot R_t (k/B_t)\leq  (4R_t^3/C_2)
$$
using assumption {\bf p2} of the lemma.

For every $j\in Q_s(i)$ one has that $j$ is bad only if $j$ collides with $S^t_t$ under $H_{t, s'}$ for at least $R_t^{1-\delta}$ values of $s'\in [1:R_t]\setminus \{s\}$. This probability is bounded by $\prob[X_{-s}\geq R_t^{1-\delta}]$, where $X_{-s}$ are as defined above. We thus have using ~\eqref{eq:prob-s} that $\prob[j\text{~is bad}|i\text{~collides with $j$ under $H_{t, s}$}]\leq e^{-R_t^{1-\delta}/6}$. Note that this is where we crucially use the fact that $X_{-s}$ does not depend on hashing $H_{t, s}$.
By a union bound over all $j\in Q_s(i)$ and all $s\in [1:R_t]$ that the probability that $i$ collides with a bad element is at most
\begin{equation}\label{eq:v-t-1}
\begin{split}
\sum_{s=1}^{R_t}\expect_{H_{t, s}}[|Q_s(i)|]\cdot e^{-R_t^{1-\delta}/6}&\leq \frac1{C_2}R_t\cdot (4\cdot R_t^3)\cdot e^{-R_t^{1-\delta}/6}\\
&\leq \frac{1}{C_2} 4R_t^4\cdot e^{-R_t^{1-\delta}/6}\\
&= e^{-R_t^{1-\delta}/12},\\
\end{split}
\end{equation}
where we used the fact that the last inequality holds for all $t\geq 1$ simultaneously as long as $C_2$ is larger than a constant that may depend on $C_1$ and $\delta$.  Summing over all elements in $S_t^t$, we get 
\begin{equation}\label{eq:u-t}
\expect_{H_{t, 1},\ldots, H_{t, R_t}}[|U_t|]\leq k_t e^{-R_t^{1-\delta}/12}\leq k e^{-R_t^{1-\delta}/12}.
\end{equation}

\paragraph{Putting it together.} Gathering bounds from~\eqref{eq:v-t},~\eqref{eq:bb-t} and ~\eqref{eq:u-t}, we get, using the fact that $S^{t+1}_{t+1}=\text{Bad}_t\cup U_t\cup V_t$ by Algorithm~\ref{alg:partition} (line~7),
\begin{equation}\label{eq:res-exp}
\begin{split}
\expect_{H_{t, 1},\ldots, H_{t, R_t}}\left[|S^{t+1}_{t+1}|\right]&=k\cdot \left(\frac1{100}\frac1{R_t^2} 2^{-2^{(1-\delta)t}+1}+e^{-R_t^{1-\delta}/6}+e^{-R_t^{1-\delta}/12}\right)\\
&=k\cdot \left(\frac1{100}\frac1{R_t^2} 2^{-2^{(1-\delta)t}+1}+e^{-C_1^{1-\delta} 2^{(1-\delta)t}/6}+e^{-C_1^{1-\delta} 2^{(1-\delta)t}/12}\right).\\
%&\leq 2k\cdot \left(\frac1{100}\frac1{R_t^2} 2^{-(3/2)^t+1}+2e^{-R_t^{1-\delta}/12}\right)\\
%&\leq 2k\cdot \frac{3}{100}\frac1{R_t^2} 2^{-(3/2)^t+1}
\end{split}
\end{equation}
We now show that for every $\delta\in (0, 1/2)$ if $C_1$ is larger than an absolute constant, the first term in parentheses on the last line above is at least as large as the other two for all $t\geq 1$. We first note that $e^{-C_1^{1-\delta} 2^{(1-\delta)t}/6}\leq e^{-C_1^{1-\delta} 2^{(1-\delta)t}/12}$ for all $C_1>0, \delta\in (0, 1/2), t\geq 1$, so it suffices to prove that 
\begin{equation}\label{eq:3terms}
\frac1{100}\frac1{R_t^2} 2^{-2^{(1-\delta)t}+1}\geq e^{-C_1^{1-\delta} 2^{(1-\delta)t}/12}
\end{equation}
for all $t\geq 1$ if $C_1$ is larger than an absolute constant. Intuitively, this is true because the rhs decays exponentially in $C_1$ for all $t\geq 0$, while the lhs decays only polynomially in $C_1$. More formally, taking the ratio of the two quantities, we get, assuming that $C_1\geq 12^2$,
\begin{equation*}
\begin{split}
e^{-C_1^{1-\delta} 2^{(1-\delta)t}/12}\cdot \left(\frac1{100}\frac1{R_t^2} 2^{-2^{(1-\delta)t}+1}\right)^{-1}&\leq e^{-\sqrt{C_1} 2^{(1-\delta)t}/12+(\ln 2) 2^{(1-\delta)t}}\cdot 100\cdot R_t^2\\
&\leq e^{-(\sqrt{C_1}/12-\ln 2)\cdot 2^{(1-\delta)t}}\cdot 100\cdot C_1^2 4^t\\
&\leq \exp\left(-(\sqrt{C_1}/12-\ln 2)\cdot 2^{(1-\delta)t}+(\ln 4) t+2\ln (100 C_1)\right)\
\end{split}
\end{equation*}
We now show that the exponent above is non-positive for all $t\geq 1$ as long as $C_1$ is larger than a constant. Indeed, taking the derivative of the exponent with respect to $t$, we get 
$$
\left(-(\sqrt{C_1}/12-\ln 2) 2^{(1-\delta) t}+(\ln 2) t+\ln (100C_1)\right)'= -(\sqrt{C_1}/12-\ln 2) (1-\delta) \ln 2\cdot 2^{(1-\delta)t}+\ln 2,
$$
which is nonpositive for all $t\geq 1$ as long as $C_1$ is larger than an absolute constant. This means that 
$$
\max_{t\geq 1} \left[-(\sqrt{C_1}/12-\ln 2) 2^{(1-\delta) t}+(\ln 2) t+\ln (100C_1)\right]\leq -(\sqrt{C_1}/12-\ln 2) 2^{1-\delta}+\ln 2+\ln (100C_1)\leq 0
$$
as long as $C_1$ is larger than an absolute constant (since $\sqrt{C_1}$ asympotitcally dominates $\ln C_1$). This establishes~\eqref{eq:3terms}.

Substituting this upper bound into~\eqref{eq:res-exp}, we get
$$
\expect_{H_{t, 1},\ldots, H_{t, R_t}}\left[|S^{t+1}_{t+1}|\right]\leq k\cdot \frac{3}{100}\frac1{R_t^2} 2^{-2^{(1-\delta)t}+1}\\
$$

By Markov's inequality applied to the last expression above we have 
\begin{equation}\label{eq:res-prob}
\begin{split}
\prob_{H_{t, 1},\ldots, H_{t, R_t}}\left[|S^{t+1}_{t+1}|> k\cdot \frac1{R_t} 2^{-2^{(1-\delta)t}+1}\right]\leq \frac{3}{100}\frac1{R_t}.
\end{split}
\end{equation}
Let ${\mathcal E}_t$ denote the intersection of ${\mathcal E_{t-1}}$ with the failure event in~\eqref{eq:res-prob}, we get, conditioned on ${\mathcal E}_t$ 
$$
|S^{t+1}_{t+1}|\leq k\cdot \frac1{R_t} 2^{-2^{(1-\delta)t}+1}.
$$
This completes the inductive step and the proof of the lemma.
\end{proofof}

\begin{proofof}{Lemma~\ref{lm:construct-partition-runtime}}
First note that there exists a (simple) efficient data structure for answering queries of the form `how many elements of a set $T$  hash within circular distance $\Delta$ of a point $x$ under hash function $\pi$?'.  Indeed, it suffices to cut the circle into two halves and for each of the halves construct a binary search tree on $T$, with each node annotated with the number of nodes in its subtree. Then each query about neighbors in the circular distance can be answered by answering at most two queries for the two data structure on the half-circles, for a total time of $O(\log |T|)$. We now show how to use this data structure to implement Algorithm~\ref{alg:partition}.

 For each $t$ at the beginning of the $t$-th iteration of Algorithm~\ref{alg:partition} for each $s=1,\ldots, R_t$ one constructs  two binary search trees on the permuted elements $\pi_{t, s}(S^t_t)$ as described above. This takes time $O(R_t |S_t^t|\log |S_t^t|)$. Equipped with this data structure, we can construct the set $\text{Bad}_t$ and the set $V_t$ in time $O(R_t\cdot S\log |S^t_t|)$.  Then construct a similar pair of augmented binary search trees for the set $\text{Bad}_t$ in time $O(|S|\log |S|)$. Using this data structure the set $U_t$ can be constructed in time $O(R_t |S|\log |S|)$. 
Summing over $t$ gives the final result.
\end{proofof}

%% file: estvals.tex
%!TEX root = ./ft-hd-wrap.tex 
\section{Properties of \textsc{EstimateValues}}\label{sec:app}
In this section we describe the procedure \textsc{EstimateValues} from~\cite{K16} (see Algorithm~\ref{alg:est}), which, given access to a signal $x$ in frequency domain (i.e. given $\wh{x}$), a partially recovered signal $\chi$ and a target list of locations $L\subseteq \nsq$,
estimates values of the elements in $L$, and outputs the elements that are above a threshold $\nu$ in absolute value. We need a slight strengthening of Lemma~9.1 from~\cite{K16}, which we state here. 

\begin{algorithm}
\caption{\textsc{EstimateValues}($\chi, L, \{(H_r, a_r, m(x, H_r, a_r)))\}_{r=1}^{r_{max}}$)}\label{alg:est} 
\begin{algorithmic}[1] 
\Procedure{EstimateValues}{$\chi, L, \{(H_r, a_r, m(x, H_r, a_r)))\}_{r=1}^{r_{max}}$}\Comment{$H_r$ are $(\pi_r, B, F)$-hashings}
\For{$r=1$ to $r_{max}$}\Comment{$\pi_r=(\sigma_r, q_r)$ for $r=1,\ldots, r_{max}$}
\State Compute $m_j(x-\chi, H_r, a_r)$ for $j\in [B]$ \Comment{Computation is done with polynomial precision, }
\State \Comment{using \textsc{HashToBins} as per Lemma~\ref{l:hashtobins}}
\EndFor
\For{$i\in L$}
\For{$r=1$ to $r_{max}$}\Comment{Note that $o_i(i)$ implicitly depends on $H_r$}
\State $w^r_i\gets  G^{-1}_{o_i(i)}m_{h_r(i)}(x-\chi, H_r, a_r) \omega^{-a_r \sigma_r i}$ \Comment{Estimate $(x-\chi)_i$ from each measurement}
\EndFor
\State $w_i\gets \text{median}\{w^r_i\}_{r=1}^ {r_{max}}$ \Comment{Median is taken coordinatewise}
\EndFor
\State \textbf{return} $w_L$
\EndProcedure 
\end{algorithmic}
\end{algorithm}

We will use 
\begin{definition}
For any $x\in \C^\nsq$ and any hashing $H=(\pi, B, F)$ define the vector $\mu^2_{H, \cdot}(x)\in \R^\nsq$ by letting for every $i\in \nsq$ 
$\mu^2_{H, i}(x):=\abs{G_{o_i(i)}^{-2}}\sum_{j \in \nsq\setminus \{i\}} \abs{x_j}^2 \abs{G_{o_i(j)}}^2$.
\end{definition}

The following properties of \textsc{HashToBins} will be using in the analysis of \textsc{EstimateValues}:
\begin{lemma}[Lemma 2.9 of~\cite{K16}]\label{lm:hashing}
There exists a constant $C>0$ such that for any integer $B\geq 1$, any $x, \chi\in \C^\nsq, x':=x-\chi$, if  $\sigma, a\in \nsq$, $\sigma$ odd,  are selected uniformly at random, the following conditions hold for every $q\in [n]$. 

Let $\pi=(\sigma, q)$, $H=(\pi, B, F)$, where $G$ is the filter with $B$ buckets and sharpness $F$ as per Definition~\ref{def:filterG}, and let
$u =\Call{HashToBins}{\hat x, \chi, (H, a)}$.  Then if $\fc\geq 2$, for any $i\in \nsq$
  \begin{description}
\item[(1)] For any $H$ one has $\max_{a\in \nsq} \abs{G_{o_i(i)}^{-1}\omega^{-a \sigma i}u_{h(i)} - x'_i}\leq G_{o_i(i)}^{-1}\cdot \sum_{j\in S\setminus \{i\}} G_{o_i(j)} |x'_j|$. Furthermore,
$\expect_H[G_{o_i(i)}^{-1}\cdot \sum_{j\in S\setminus \{i\}} G_{o_i(j)} |x'_j|]\leq C ||x'||_1/B+n^{-\Omega(c)}$;
  \item[(2)] $\expect_H[\mu^2_{H, i}(x')] \leq  C \norm{2}{x'}^2/B$,
\end{description}
Furthermore, 
\begin{description}
\item[(3)] for any hashing $H$, if $a$ is chosen uniformly at random from $\nsq$, one has 
  $$
  \expect_{a}[\abs{G_{o_i(i)}^{-1}\omega^{-a\sigma i}u_{h(i)} - x'_i}^2]\leq \mu^2_{H, i}(x')+n^{-\Omega(c)}.
  $$
\end{description}
Here $c>0$ is an absolute constant that can be chosen arbitrarily large at the expense of increasing runtime by a factor of $c$.
\end{lemma}
We note that Lemma~\ref{lm:hashing} was proved in~\cite{K16} for a slightly different choice of filter $G$ and for a uniformly random $q$, but the proof carries over directly to our setting.

{\noindent {\bf Lemma~\ref{lm:estimate-l1l2}} (Restated; bounds on estimation quality for Algorithm~\ref{alg:est}) \em
For every $x, \chi\in \C^n$,  every $L\subseteq \nsq$, every set $S\subseteq \nsq$ the following conditions hold for functions $e^{head}$ and $e^{tail}$ are defined with respect to $S$ (see ~\eqref{eq:eh-pi} and~\eqref{eq:et-pi}). If $r_{max}$ is larger than an absolute constant, then for every sequence $H_r, r=1,\ldots, r_{max}$ of $(\pi_r, B, F)$ hashings  the output $w$ of 
$$
\Call{EstimateValues}{\chi, L, \{(H_r, a_r, m(x, H_r, a_r)))\}_{r=1}^{r_{max}}}
$$
 satisfies, for each $i\in L$  
$$
|w_i- (x-\chi)_i|\leq 2\cdot\quant^{1/5}_r e^{head}_i(H_r, x-\chi)+ 2\cdot\quant^{1/5}_r e^{tail}_i(H_r, a_r, x-\chi)+n^{-\Omega(c)}.
$$
The sample complexity is bounded by $O(F B r_{max})$. The runtime is bounded by $O((F\cdot B\cdot \log n+||\chi||_0\log n+|L|)\cdot r_{max})$.
}

\begin{proof}
 We have by definition of the measurements $m_j$  (see Definition~\ref{def:measurements}) for every hashing $H$ and $a\in [n]$
  \begin{align*}
    m_{h(i)} &= \sum_{j \in [n]} G_{o_i(j)} (x - \chi)_{j} \omega^{a \sigma j},
  \end{align*}
so
 \begin{align*}
    G_{o_i(i)}^{-1} m_{h(i)}  \omega^{-a \sigma i} &=(x - \chi)_i+ G_{o_i(i)}^{-1} \sum_{j \in [n]\setminus \{i\}} G_{o_i(j)} (x - \chi)_{j} \omega^{a \sigma (j-i)}.
  \end{align*}  
  We thus have by triangle inequality, splitting the rhs into the contribution of the `head' elements (i.e., elements in $S$) and `tail' elements (i.e. elements in $[n]\setminus S$), that
  
  \begin{align*}
    |G_{o_i(i)}^{-1} m_{h(i)}  \omega^{-a\sigma i}-(x - \chi)_i| &\leq  \left|G_{o_i(i)}^{-1}\sum_{j \in [n]\setminus \{i\}} G_{o_i(j)} (x - \chi)_{j} \omega^{a \sigma (j-i)}\right|\\
    &\leq  \left|G_{o_i(i)}^{-1}\sum_{j \in S\setminus \{i\}} G_{o_i(j)} (x - \chi)_{j} \omega^{a\sigma (j-i)}\right|\\
    &+\left|G_{o_i(i)}^{-1}\sum_{j \in [n]\setminus (S\cup \{i\})} G_{o_i(j)} (x - \chi)_{j} \omega^{a\sigma (j-i)}\right|\\
    &\leq  G_{o_i(i)}^{-1} \sum_{j \in S\setminus \{i\}} G_{o_i(j)} |(x - \chi)_{j}| +G_{o_i(i)}^{-1} \left|\sum_{j \in [n]\setminus (S\cup \{i\})} G_{o_i(j)} (x - \chi)_{j} \omega^{a\sigma (j-i)}\right|\\
    &+|\Delta_{(n/B)\cdot h(i)}|    \\
    &=  e^{head}_i(x - \chi, H) +e^{tail}_i(x-\chi, H, a)+|\Delta_{(n/B)\cdot h(i)}|
  \end{align*}    
%\xxx{add error terms}
  
  We now use the bound above to obtain the conclusion of the lemma.
  Recall that for each $i\in L$ the final estimate $w_i$ is computed as a median of $w_i^r$'s along real and imaginary axes in line~9 of Algorithm~\ref{alg:est}. Let $r'\in [1:r_{max}]$ and $r''\in [1:r_{max}]$ be such that 
  $$
  w_i=\text{Re}(w_i^{r'})+\text{Im}(w_i^{r''})\cdot \mathbf{i}.
  $$
  
We have 
\begin{equation}\label{eq:i2bgh23}
  |\text{Re}(w_i^{r'}-(x-\chi)_i)|\leq |w_i^{r'}-(x-\chi)_i|\leq e^{head}_i(H_{r'}, x - \chi) +e^{tail}_i(H_{r'}, a_{r'}, x-\chi),
\end{equation}
and  since $r'$ is the result of taking the median of the list $\{\text{Re}(w_i^{r'})\}$, we have  
  $$
  |\text{Re}(w_i^{r'}-(x-\chi)_i)|\leq \quant^{1/5}_r e^{head}_i(H_r, x - \chi) +\quant^{1/5}_r e^{tail}_i(H_r, a_r, x - \chi).
  $$
  
  Indeed, at most a $2/5$ fraction of the error terms on the rhs of ~\eqref{eq:i2bgh23}, namely $e^{head}_i(H_{r'}, x - \chi) +e^{tail}_i(H_{r'}, a_{r'}, x-\chi)$, 
  are larger than $\quant^{1/5}_r e^{head}_i(H_r, x - \chi) +\quant^{1/5}_r e^{tail}_i(H_r, a_r, x-\chi)$. These error terms correspond to either the bottom or the top of the list $\{\text{Re}(w_i^{r'})\}$, and since $2/5<1/2$, the median estimate satisfies the upper bound above.
  
  A similar argument for the imaginary part shows that 
  
  $$
  |\text{Im}(w_i^{r'}-(x-\chi)_i)|\leq \quant^{1/5}_r e^{head}_i(H_r, x - \chi) +\quant^{1/5} e^{tail}_i(H_r, a_r, x-\chi).
  $$
  
Putting the two estimates together and using the bound $|a+b\cdot \mathbf{i}|\leq |a|+|b|$, we get for each $i\in L$
  $$
|w_i-(x-\chi)_i|\leq 2\cdot\quant^{1/5}_r e^{head}_i(H_r, x-\chi)+ 2\cdot\quant^{1/5}_r e^{tail}_i(H_r, a_r, x-\chi)+n^{-\Omega(c)}
  $$
as required.

The sample complexity follows by Lemma~\ref{l:hashtobins}.
 
The runtime analysis is as follows:
 \begin{itemize}
 \item Computing $m_j(x-\chi, H_r, a_r)$ for $j\in [B]$ and $r=1,\ldots, r_{max}$ takes $O((F B\log B+||\chi||_0 \log n)r_{max})$ time Lemma~\ref{l:hashtobins}.
 \item Computing estimates for each $i\in L$. This takes time $|L|\cdot r_{max}$ since median can be found in linear time.
 \end{itemize}

\end{proof}

\begin{theorem}[Chernoff bound]\label{thm:chernoff}
Let $X_1,\ldots, X_n$ be independent $0/1$ Bernoulli random variables with $\sum_{i=1}^n \expect[X_i]=\mu$. Then for any $\delta>0$ one has $\prob[\sum_{i=1}^n X_i>(1+\delta)\mu]<e^{(\delta -(1+\delta)\ln (1+\delta)) \mu}$.
\end{theorem}

We will use 

\begin{lemma}\label{lm:quant-exp}
Let $X_1,\ldots, X_n\geq 0$ be independent random variables with $\expect[X_i]\leq \mu$ for each $i=1,\ldots, n$. Then for any $\gamma\in (0, 1)$ if $Y\leq\quant^{\gamma} (X_1,\ldots, X_n)$,
then 
\begin{description}
\item[(1)] $\expect[\left|Y-4\mu/\gamma\right|_+]\leq (\mu/\gamma) \cdot 2^{-\Omega(\gamma n)}$;
\item[(2)] $\expect[\left|Y-4\mu/\gamma\right|^2_+]\leq (\mu/\gamma)^2 \cdot 2^{-\Omega(\gamma n)}$;
\item[(3)] $\prob[Y\geq 4\mu/\gamma]\leq 2^{-\Omega(\gamma n)}$;
\item[(4)] For every $t\geq 1$ one has 
$$
\prob[Y\geq t\mu/\gamma]\leq ( 0.99t/e)^{-0.99\gamma n}.
$$
\end{description}
\end{lemma}
\begin{proof}
%The proof is analogous to the proof of Lemma~\ref{lm:median}.

For any $t\geq 1$ by Markov's inequality $\prob[X_i>t \mu/\gamma]\leq \gamma/t$. Define indicator random variables $Z_i$ by letting $Z_i=1$ if $X_i>t \mu/\gamma$ and $Z_i=0$ otherwise. Note that 
$$
\expect[Z_i]\leq \gamma/t
$$
 for each $i$.
Then since $Y$ is bounded above by the $\gamma n$-th largest of $\{X_i\}_{i=1}^n$, we have $\prob[Y>t \mu/\gamma]\leq \prob[\sum_{i=1}^n Z_i\geq \gamma n]$. Let $\nu:=\sum_{i=1}^n \expect[Z_i]$.  We now apply the Chernoff bound (Theorem~\ref{thm:chernoff}) with $\delta=\gamma' n/\nu-1$, $\gamma'=0.99\gamma$, to the sequence $Z_i,i=1,\ldots, n$. Note that by our setting of $\delta$ we have $(1+\delta)\nu=\gamma' n$, so
\begin{equation*}
\begin{split}
\prob\left[\sum_{i=1}^n Z_i> \gamma' n\right]&\leq \exp\left((\gamma' n/\nu-1-(\gamma' n/\nu)\ln (\gamma' n/\nu))\nu\right)\\
&=\exp\left(\gamma' n-\nu-\gamma' n\ln (\gamma' n/\nu)\right)\\
&\leq \exp\left(\gamma' n(1-\ln (\gamma' n/\nu))\right)\\
&\leq \exp\left(\gamma' n(1-\ln ((\gamma'/\gamma) t))\right)~~~~~~\text{(since $\nu\leq n\gamma/t$)}\\
&= e^{\gamma' n} ((\gamma'/\gamma) t)^{-\gamma' n}.
\end{split}
\end{equation*}
We thus get 
\begin{equation}\label{eq:vi34g43qqg}
\prob\left[\sum_{i=1}^n Z_i\geq \gamma n\right]\leq \prob\left[\sum_{i=1}^n Z_i> \gamma' n\right]\leq  ( 0.99t/e)^{-0.99\gamma n}.
\end{equation}
This proves {\bf (4)}. Letting $t=4$ in the bound above proves {\bf (3)}.

For {\bf (1)} we have, as long as $n$ is sufficiently large (depending on $\gamma$), 
\begin{equation*}
\begin{split}
\expect[Y\cdot \mathbf{1}_{Y\geq 4\cdot \mu/\gamma}]&\leq \int_{4}^\infty t \mu \cdot \prob[Y\geq t\cdot \mu/\gamma]dt\\
&\leq \int_{4}^\infty t \mu (0.99 t/e)^{-0.99\gamma n}dt\text{~~~~~~~~~~ (by~\eqref{eq:vi34g43qqg})}\\
&\leq e^{-\gamma n/4}\int_{4}^\infty t \mu (0.99 t/e)^{-\gamma n/4}dt\\
&=O(\mu\cdot e^{-\gamma n/4}).
\end{split}
\end{equation*}

For {\bf (2)} we have, as long as $n$ is sufficiently large (depending on $\gamma$), 
\begin{equation*}
\begin{split}
\expect[Y\cdot \mathbf{1}_{Y\geq 4\cdot \mu/\gamma}]&\leq \int_{4}^\infty t^2 \mu^2 \cdot \prob[Y\geq t\cdot \mu/\gamma]dt\\
&\leq \int_{4}^\infty t^2 \mu^2 (0.99 t/e)^{-0.99\gamma n}dt\text{~~~~~~~~~~ (by~\eqref{eq:vi34g43qqg})}\\
&\leq e^{-\gamma n/4}\int_{4}^\infty t^2 \mu (0.99 t/e)^{- \gamma n/4}dt\\
&=O(\mu^2\cdot e^{-\gamma n/4}).
\end{split}
\end{equation*}
as required.

\end{proof}

We also have 
\begin{lemma}\label{lm:median-etail-ehead}
For every $x\in \C^n$, every $S\subseteq [n]$, every $i\in [n]$, every integer $r_{max}$ larger than an absolute constant, integers $B, F$ with $B$ a power of two and $F\geq 2$, the following conditions are satisfied for a sequence of random hashings $H_r=(\pi_r, B, F)$,  and random evaluation points $a_r$, $r=1,2,\ldots, r_{max}$. 

If 
$Z^{head}:=e^{head}_i(\{H_r\}, x)=\quant^{1/5}_r e^{head}_i(H_r, x)$ (as per ~\eqref{eq:eh})
and 
$Z^{tail}:=e^{tail}_i(\{H_r, a_r\}, x)=\quant^{1/5}_r e^{tail}_i(H_r, a_r, x)$ (as per ~\eqref{eq:et-pi-quant}), where $e^{head}$ and $e^{tail}$ are defined with respect to the set $S$, one has
 
\begin{description}
\item[(1)] $\expect_{\{H_r\}}\left[(Z^{head})^2\right]=O\left(\left(\frac1{B}||x_S||_1\right)^2\right)$;
\item[(2)] $\expect_{\{H_r, a_r\}}\left[(Z^{tail})^2\right]=O(||x_{[n]\setminus S}||_2^2/B$);
\item[(3)] $\prob_{\{H_r\}}\left[Z^{head}>O\left(\frac1{B}||x_S||_1\right)\right]=2^{-\Omega(r_{max})}$;
\item[(4)] $\prob_{\{H_r\}}\left[Z^{tail}>O(||x_{[n]\setminus S}||_2/\sqrt{B})\right]=2^{-\Omega(r_{max})}$.
\end{description}
\end{lemma}
\begin{proof}
We have $Z^{head}\leq |Z^{head}-40\expect[Z^{head}]|_++40\expect[Z^{head}]$, so, since $(a+b)^2\leq 2a^2+2b^2$ for all $a, b\in \R$,
$$
\expect\left[(Z^{head})^2\right]\leq 2\expect\left[(|Z^{head}-40\expect[Z^{head}]|_+)^2\right]+2 (40\expect[Z^{head}])^2.
$$
By Lemma~\ref{lm:quant-exp}, {\bf (2)} we have $\expect\left[(|Z^{head}-40\expect[Z^{head}]|_+)^2\right]=O((\expect[Z^{head}])^2)$ as long as $r_{max}$ is larger than an absolute constant, as assumed by the lemma. An application of Lemma~\ref{lm:hashing}, {\bf (1)} now gives the first bound. The proof of the second claim is analogous using Lemma~\ref{lm:hashing}, {\bf (2)}.

Claims {\bf (3)} and {\bf (4)} follow similarly using Lemma~\ref{lm:quant-exp}.
\end{proof}

\subsection{Properties of \textsc{HashToBins}}\label{sec:hash2bins}
\begin{algorithm}
\caption{Hashing using Fourier samples (analyzed in Lemma~\ref{l:hashtobins})}\label{alg:hash2bins} 
\begin{algorithmic}[1] 
\Procedure{HashToBins}{$\wh{x}, \chi, (H, a)$}\Comment{Hashing $H=(\pi, B, F)$, $a\in [n]$}
\State Compute $y'=\wh{G}\cdot P_{\sigma, a, q}(\hat x-\hat \chi')$, 
for some $\chi'$ with $\norm{\infty}{\wh{\chi}-\wh{\chi}'}<\norm{2}{\chi}\cdot n^{-c}$\Comment{Using Lemma~\ref{lem:semi_equi_std} with $\delta=n^{-2c}$, $c\geq 2$}
\State Compute $u_j = \sqrt{n}\F^{-1}(y')_{(n/B)\cdot j}$ for $j \in [B]$
\State {\bf return} $u$
\EndProcedure 
\end{algorithmic}
\end{algorithm}

The main lemma about the performance of \textsc{HashToBins} is

\noindent{{\bf Lemma~\ref{l:hashtobins}} (Restated) \em 
  \Call{HashToBins}{$\wh{x}, \chi, (H, a)$}, where $H=(\pi, B, F)$, computes
  $u\in \C^{B}$ such that for any $i \in [n]$, $u_{h(i)} = \Delta_{h(i)} + \sum_j G_{o_i(j)}(x - \chi)_j \omega^{a\sigma j}$,  where $G$ is the filter defined in section~\ref{sec:prelim}, and for all $i\in [n]$ we have that $\Delta_{h(i)}^2 \leq \norm{2}{\chi}^2
\cdot n^{-c}$ is a negligible error term (and $c>0$ is an absolute constant that governs the precision that semi-equispaced FFT, i.e. Lemma~\ref{lem:semi_equi_std}, is invoked with).  It takes $O(B \fc)$
  samples, and $O(F\cdot B\log B+\norm{0}{\chi} \log n)$ time.
}

\begin{proof}
  The {\bf first step} (line~2) in \textsc{HashToBins} is to compute
 \[
  y'=\wh{G} \cdot P_{\sigma, a, q}\wh{x - \chi'}=\wh{G} \cdot P_{\sigma, a, q}\wh{x - \chi}+\wh{G} \cdot P_{\sigma, a, q}\wh{\chi - \chi'},
  \]
  for an approximation $\wh{\chi}'$ to $\wh{\chi}$ obtained using Lemma~\ref{lem:semi_equi_std}, {\bf (b)}.  We now verify the runtime and precision guarantees. Recall that $\supp \wh{G}\subseteq [-O(FB), O(F B)]$ by Lemma~\ref{lem:filter_properties}. This means that it is sufficient to compute $\wh{\chi}_i$ on the set $S\subseteq [n]$ defined as  $S=\{i\in [n]: \sigma(i-a)\in [-O(FB), O(FB)]\}$.
  By~Lemma~\ref{lem:semi_equi_std}, {\bf (b)}, an approximation $\wh{\chi}'$ to
  $\wh{\chi}$ can be computed in $O(F\cdot B\log n)$ time such
  that
  \[
  |\wh{\chi}_i-\wh{\chi}'_i|<\norm{2}{\chi}\cdot n^{-2c}
  \]
  for all such $i$.  Since $\norm{1}{\wh{G}} \leq
  \sqrt{n}\norm{2}{\wh{G}} = \sqrt{n}\norm{2}{G} \leq n
  \norm{\infty}{G} \leq n$ and $\wh{G}$ is $0$ outside $S$, this
  implies that
  \begin{equation}\label{eq:gh-norm}
    \norm{2}{\wh{G}\cdot P_{\sigma, a, q}(\wh{\chi-\chi'})}\leq \norm{1}{\wh{G}}\max_{j \in S} \abs{\wh{\chi-\chi'}_i} \leq \norm{2}{\chi}\cdot n^{1-2c}.
  \end{equation}
  Define $\Delta$ by $\wh{\Delta}=\sqrt{n}\wh{G}\cdot P_{\sigma, a, q}(\wh{\chi-\chi'})$.  
  
  The {\bf second step} (line~3) in \textsc{HashToBins} is to compute $u \in
  \C^B$ such that for all $i$,
  \[
  u_{h(i)} =
  \sqrt{n}\F^{-1}(y')_{(n/B)\cdot h(i)}=\sqrt{n}\F^{-1}(y)_{(n/B)\cdot h(i)}+\Delta_{(n/B)\cdot h(i)},
  \]
  for $y = \wh{G} \cdot P_{\sigma, a, q}\wh{x - \chi}$.  This
  computation takes $O(\norm{0}{y'} + B \log B) =O( F B \log n)$ 
  time (alias $y'$ to length $B$ and compute a length $B$ FFT).  We have by the convolution theorem (see~\eqref{eq:convolution-thm}) that
  \begin{align*}
    u_{h(i)} &= \sqrt{n}\F^{-1}(\wh{G} \cdot P_{\sigma, a, q}\wh{(x - \chi)})_{(n/B)\cdot h(i)} + \Delta_{(n/B)\cdot h(i)}\\
    &= (G * \F^{-1}(P_{\sigma, a, q}\wh{(x - \chi)}))_{(n/B)\cdot h(i)}+\Delta_{(n/B)\cdot h(i)}\\
    &= \sum_{\pi(j) \in [n]} G_{(n/B)\cdot h(i)-\pi(j)} \F^{-1}(P_{\sigma, a, q}\wh{(x - \chi)})_{\pi(j)}+\Delta_{(n/B)\cdot h(i)}\\
    &= \sum_{j \in [n]} G_{o_i(j)} (x - \chi)_{j} \omega^{a\sigma j}+\Delta_{(n/B)\cdot h(i)}
  \end{align*}
  where the last step is the definition of $o_i(j)$ and Lemma~\ref{lm:perm}.
  
Finally, we note that 
\[
|\Delta_{(n/B)\cdot h(i)}|\leq \norm{2}{\Delta} =\norm{2}{\wh{\Delta}}=\sqrt{n}\norm{2}{\wh{G}\cdot P_{\sigma, a, q}(\wh{\chi-\chi'})}\leq \norm{2}{\chi} n^{3/2-2c}\leq \norm{2}{\chi} n^{-c},
\]
 where we used \eqref{eq:gh-norm} and the assumption that $c\geq 2$ in the last step. This completes the proof.
\end{proof}

The following lemma is analogous to Lemma 9.4 of~\cite{IKP}, but does not make the assumption that the number of repetitions involved in the quantile operation is a constant. The proof is essentially the same, but is given below for completeness.
\begin{lemma}\label{lm:noisiest-buckets}
For every $\gamma\in (0, 1)$, integers $m, n\geq 1$ such that $n> 4/\gamma$, every sequence $X^1,\ldots, X^n\in \R_+^m$ of random variables with non-negative entries such that $X^j\in \R_+^n$ are independent, $\expect[X^j_i]\leq \nu$ for every $i=1,\ldots, m, j=1,\ldots, n$ and $\nu>0$, the following conditions hold. If for every $i=1,\ldots, m$ 
$$
Y_i=\quant^{\gamma} (X^1,\ldots, X^n),
$$
then for every $U$ between $1$ and $m$
$$
\expect\left[\max_{Q\subseteq [m], |Q|\leq U} \sum_{i\in Q} Y_i\right]\leq U\cdot (4e\nu/\gamma)\cdot (m/U)^{2/(\gamma n)}
$$
\end{lemma}
Note that the lemma assumes that $X^j\in \R^n_+$ are independent, but allows for the coordinates of each $X^j$ to be arbitrarily correlated.
\begin{proof}
First fix $i\in \{1, 2,\ldots, m\}$. By Lemma~\ref{lm:quant-exp}, {\bf (4)} we have 
for every $t\geq 1$ 
\begin{equation*}
\prob[Y\geq t\nu/\gamma]\leq ( 0.99t/e)^{-0.99\gamma n}\leq (t/(2e))^{-\gamma n/2}.
\end{equation*}

We thus have for $t\geq 1$
\begin{equation}\label{eq:20jg2g2}
\expect[|\{i: Y_i\geq t\nu/\gamma\}|]\leq m\cdot (t/(2e))^{-\gamma n/2},
\end{equation}
and hence for every threshold $\theta>0$ one has
\begin{equation}\label{eq:max-bound}
\begin{split}
\expect\left[\max_{Q\subseteq [m], |Q|\leq U} \sum_{i\in Q} Y_i \right]&\leq \expect\left[\max_{Q\subseteq [m], |Q|=U} \sum_{i\in Q} Y_i \right]\text{~~~~~~~~~~~~~~~~~~~~~(since $Y_i\geq 0$ for all $i$)}\\
&=\expect\left[\int_0^\infty \min(U, |\{i: Y_i>\eta\}|) d\eta\right]\\
&\leq \int_0^\infty \min(U, \expect\left[|\{i: Y_i>\eta\}|\right]) d\eta\text{~~~~~(by convexity of $\min(U, x)$ as a function of $x$)}\\
&\leq U\cdot \theta+\int_\theta^\infty \min(U, \expect\left[|\{i: Y_i>\eta\}|\right])  d\eta \\
&\leq U\cdot \theta+\int_\theta^\infty \min(U, m\cdot (\eta \gamma/(2e \nu))^{-\gamma n/2})d\eta \text{~~~~~~~~~~(by~\eqref{eq:20jg2g2}, as long as $\theta\geq \nu/\gamma$)}\\
&\leq U\cdot \theta+\int_\theta^\infty m\cdot (\eta \gamma/(2e \nu))^{-\gamma n/2}d\eta \\
&\leq U\cdot \theta+m\cdot (\gamma/(2e\nu))^{-\gamma n/2}\int_\theta^\infty \eta^{-\gamma n/2}d\eta \\
&\leq U\cdot \theta+m\cdot (\gamma/(2e\nu))^{-\gamma n/2} \frac1{\gamma n/2-1} \theta^{-\gamma n/2+1}\\
&\leq U\cdot \theta+m\cdot (\theta\cdot  \gamma/(2e\nu))^{-\gamma n/2} \cdot \theta \text{~~~~~~~~~(since $\gamma n/2>2$ by assumption)}\
\end{split}
\end{equation}
 We now let $\theta=(2e\nu/\gamma)\cdot (m/U)^{2/(\gamma n)}>\nu/\gamma$, so that 
$$
m\cdot (\theta\cdot \gamma/(2e\nu))^{-\gamma n/2}=k,
$$
 and substituting into~\eqref{eq:max-bound}, we get
\begin{equation*}
\begin{split}
\expect\left[\max_{Q\subseteq [m], |Q|\leq U} \sum_{i\in Q} Y_i \right]&\leq 2U\cdot \theta=U\cdot (4e\nu/\gamma)\cdot (m/U)^{2/(\gamma n)},
\end{split}
\end{equation*}
as required.
\end{proof}

%% file: loc.tex
%!TEX root = ./ft-hd-wrap.tex 
\section{Signal location primitive and its analysis}\label{sec:location}

We reuse the location primitive from~\cite{K16} (\textsc{LocateSignal}, see Algorithm~\ref{alg:location}), but present it here with simplified notation adapted to the 1d setting (thus obviating the need for the $\star$ operation).
As in~\cite{K16}, we will use
\begin{definition}[Balanced set of points]\label{def:balance}
For an integer $\Delta\geq 2$  we say that a (multi)set $\mathcal{Z}\subseteq \nsq$ is {\em $\Delta$-balanced}  if for every $r=1,\ldots, \Delta-1$ at least $49/100$ fraction of elements in the set $\{\omega_{\Delta}^{r\cdot z}\}_{z\in \mathcal{Z}}$ belong to the left halfplane $\{u\in \C: \text{Re}(u)\leq 0\}$ in the complex plane, where $\omega_\Delta=e^{2\pi i/\Delta}$ is the $\Delta$-th root of unity.
\end{definition}

We will also need
\begin{claim}[Claim~2.14 of~\cite{K16}]\label{cl:balanced}
There exists a constant $C>0$ such that for any $\Delta$ a power of two, $\Delta=\log^{O(1)} n$,  and $n$ a power of $2$ the following holds if $\Delta<n$. If elements of a (multi)set $\A\subseteq \nsq\times \nsq$ of size $C\log\log n$ are chosen uniformly at random with replacement from $\nsq\times \nsq$, then with probability at least $1-1/\log^4 n$ one has that the set $\{\beta\}_{(\alpha, \beta)\in \A}$ is $\Delta$-balanced.
\end{claim}
Since we only use one value of $\Delta$ in the paper (see line~4 in Algorithm~\ref{alg:location}), we will usually say that a set is simply `balanced'  to denote the $\Delta$-balanced property for this value of $\Delta$.
Before we state the algorithm and give the analysis, we need to introduce notation for bounding the influence of tail noise on the location process. We do this in the next section.

\subsection{Analysis of \textsc{LocateSignal}}

\begin{algorithm}[H]
\caption{Location primitive: given a set of measurements corresponding to a single hash function, returns a list of elements in $\nsq$, one per each hash bucket}\label{alg:location}
\begin{algorithmic}[1]
\Procedure{LocateSignal}{$\chi, H, \{m(x, H, \alpha+\h\cdot \beta\}_{(\alpha, \beta)\in \A, \h\in \H}$}\Comment{$H=(\pi, B, F)$}
\State Let $x':=x-\chi$. Compute $\{m(x', H, \alpha+\h\cdot \beta)\}_{(\alpha, \beta)\in \A, \h\in \H}$ using \textsc{HashToBins}, as per Lemma~\ref{l:hashtobins}.
\State $L\gets \emptyset$
\State $\Delta\gets 2^{\lfloor \frac1{2}\log_2 \log_2 n\rfloor}$
\State $N\gets \Delta^{\lceil\log_\Delta n\rceil}$ \Comment{Extend $\wh{x}$ implicitly to $\C^N$ periodically}
\For{$j \in [B]$}\Comment{Loop over all hash buckets, indexed by $j\in [B]$}
\State ${\bf f}\gets 0$
\For {$g=1$ to $\log_\Delta N$} 
\State $\h\gets N\Delta^{-g}$ \Comment{Note that $\h\in \H$}
\State {\bf If}~there exists a unique $r\in [0:\Delta-1]$ such that  
\State~~~~~$\left|\omega_{\Delta}^{-r\cdot \beta}\cdot \omega^{-(N\cdot \Delta^{-g} {\bf f}\cdot \beta}\cdot \frac{m_j(x', H, \alpha+\h\cdot \beta)}{m_j(x', H, \alpha)}-1\right|<1/3$ for at least $3/5$ fraction of $(\alpha, \beta)\in \A$
\State {\bf then}  ${\bf f}\gets {\bf f}+\Delta^{g-1}\cdot r$
\EndFor
\State $L \gets L \cup \left\{\sigma^{-1}{\bf f}\cdot \frac{n}{N}\right\}$ \Comment{Add recovered element to output list}
\EndFor
\State \textbf{return} $L$
\EndProcedure 
\end{algorithmic}
\end{algorithm}

Equipped with the definitions above,  we now prove the following lemma, which yields sufficient conditions for recovery of elements $i\in S$ in \textsc{LocateSignal} in terms of $e^{head}$ and $e^{tail}$.
\begin{lemma}\label{lm:loc}
Let $H=(\pi, B, F)$ be a hashing, and let $\A\subseteq \nsq\times \nsq$. Then for every $S\subseteq \nsq$ and for every $x, \chi\in \C^{\nsq}$ and $x'=x-\chi$, the following conditions hold.
Let $L$ denote the output of 
$$
\Call{LocateSignal}{\chi, H, \{m(x, H, \alpha+\h\cdot \beta)\}_{(\alpha, \beta)\in \A, \h\in \H}}.
$$

Then for any $i\in S$ such that $|x'_i|>n^{-\Omega(c)}$, if there exists $r\in [1:r_{max}]$ such that 
\begin{enumerate}
\item  $e^{head}_{i}(H, x')<|x'_i|/20$;
\item  $e^{tail}_{i}(H, \{\alpha+\h \cdot \beta\}, x')< |x'_i|/20$ for all $\h\in \H$;
\item the set $\{\beta\}_{(\alpha, \beta)\in \A}$ is balanced (as per Definition~\ref{def:balance}),
\end{enumerate}
 then $i\in L$. The time taken by the invocation of \textsc{LocateSignal} is $O(FB\log^2 n+||\chi||_0\log^2 n)$.
\end{lemma}
\begin{proof}
Let $q=\sigma i$ for convenience. We show by induction on $g=1,\ldots, \log_{\Delta} N$ that after the $g$-th iteration of lines~9-12 of Algorithm~\ref{alg:location} we have that ${\bf f}$ coincides with ${\bf q}$ on the bottom $g\cdot \log_2 \Delta$ bits, i.e. ${\bf f}-{\bf q}= 0 \mod \Delta^g$ (note that we trivially have ${\bf f}< \Delta^g$ after iteration $g$).

The {\bf base} of the induction is trivial and is provided by $g=1$.  
We now show the {\bf inductive step}. Assume by the inductive hypothesis that ${\bf f}-{\bf q}= 0 \mod \Delta^{g-1}$, so that
${\bf q}={\bf f}+\Delta^{g-1}(r_0+\Delta r_1+\Delta^2 r_2+\ldots)$ for some sequence $r_0,r_1,\ldots$, $0\leq r_j<\Delta$. Thus,  $(r_0, r_1,\ldots)$ is the expansion of $({\bf q}-{\bf f})/\Delta^{g-1}$ base $\Delta$, and $r_0$ is the least significant digit. We now show that $r_0$ is the unique value of $r$ that satisfies the conditions of lines~10-11 of Algorithm~\ref{alg:location}.

First, we have by~\eqref{eq:uexp1} together with \eqref{eq:eh-pi} and ~\eqref{eq:et-pi} one has for each $(\alpha, \beta)\in \A$ and $\h\in \H$
\begin{equation*}
\begin{split}
\left|G_{o_i(i)}^{-1}m_{h(i)}(x', H, \alpha+\h\cdot \beta)-  x'_i \omega^{(\alpha+\h\cdot \beta) {\bf q}}\right|&\leq  e^{head}_i(H, x')+e^{tail}_i(H, \alpha+\h\cdot \beta, x)+n^{-\Omega(c)}.
\end{split}
\end{equation*}

Let $j:=h(i)$. We will show that $i$ is recovered from bucket $j$. The bounds above imply that 
\begin{equation}\label{eq:gergergre}
\begin{split}
\frac{m_j(x', H, \alpha+\h\cdot \beta)}{m_j(x', H, \alpha)}=\frac{x'_i \omega^{(\alpha+\h\cdot \beta) {\bf q}}+E'}{x'_i \omega^{\alpha {\bf q}}+E''}
\end{split}
\end{equation}
for some $E', E''$ satisfying $|E'|\leq e^{head}_i(H, x')+e^{tail}_i(H, \alpha+\h\cdot \beta, x)+n^{-\Omega(c)}$ and $|E''|\leq e^{head}_i(H, x')+e^{tail}_i(H, \alpha)+n^{-\Omega(c)}$. For all but $1/5$ fraction of $(\alpha, \beta)\in \A$ we have by definition of $e^{tail}$ (see~\eqref{eq:et-pi-a-h}) that {\bf both} 
\begin{equation}\label{eq:etail-bounds-eq-1}
e^{tail}_i(H, \alpha+\h\cdot \beta, x)\leq e^{tail}_i(H, \{\alpha+\h \cdot \beta\}, x)\leq |x'_i|/20
\end{equation}
 and 
\begin{equation}\label{eq:etail-bounds-eq-2}
 e^{tail}_i(H, \alpha, x)\leq e^{tail}_i(H, \{\alpha\}, x)\leq |x'_i|/20.
\end{equation}
In particular, we can rewrite ~\eqref{eq:gergergre} as
\begin{equation}\label{eq:gergergre-2}
\begin{split}
\frac{m_j(x', H, \alpha+\h\cdot \beta)}{m_j(x', H, \alpha)}&=\frac{x'_i \omega^{(\alpha+\h\cdot \beta) {\bf q}}+E'}{x'_i \omega^{\alpha {\bf q}}+E''}\\
&=\frac{\omega^{(\alpha+\h\cdot \beta) {\bf q}}}{\omega^{\alpha {\bf q}}}\cdot\xi\text{~~~where~~}\xi=\frac{1+\omega^{-(\alpha+\h\cdot \beta) {\bf q}}E'/x_i'}{1+\omega^{-(\alpha){\bf q}} E''/x_i'}\\
&=\omega^{(\alpha+\h\cdot \beta) {\bf q}-\alpha {\bf q}}\cdot\xi\\
&=\omega^{(\h\cdot \beta) {\bf q}}\cdot\xi.\\
\end{split}
\end{equation}

Let $\A^*\subseteq \A$ denote the set of values of $(\alpha, \beta)\in \A$ that satisfy the bounds~\eqref{eq:etail-bounds-eq-1} and~\eqref{eq:etail-bounds-eq-2} above.
We thus have for  $(\alpha, \beta)\in \A^*$, combining ~\eqref{eq:gergergre-2} with assumptions {\bf 1-2} of the lemma, that
\begin{equation}\label{eq:bound-1-oigb344tg32t}
|E'|/x_i'\leq  (2/20)+n^{-\Omega(c)}\leq 1/8\text{~~~and~~~~}|E''|/x_i'\leq  (2/20)+n^{-\Omega(c)}\leq 1/8
\end{equation}
for sufficiently large $n$, where $O(c)$ is the word precision of our semi-equispaced Fourier transform computation. Note that we used the assumption that $|x'_i|\geq n^{-\Omega(c)}$.

Writing $(\alpha, \beta)\in\nsq\times \nsq$, we have by~\eqref{eq:gergergre-2} that $\frac{m_j(x', H, \alpha+\h\cdot \beta)}{m_j(x', H, \alpha)}=\omega^{\h\cdot \beta {\bf q}}\cdot\xi$, and since $\h{\bf q}=n\Delta^{-g}{\bf q}$ when $\h=N \Delta^{-g}$ (as in line~8 of Algorithm~\ref{alg:location}), we get 
$$
\frac{m_j(x', H, \alpha+\h\cdot \beta)}{m_j(x', H, \alpha)}=\omega^{\h\cdot \beta {\bf q}}\cdot\xi=\omega^{n\Delta^{-g} \beta {\bf q}}\cdot\xi=\omega^{n\Delta^{-g} \beta {\bf q}}+\omega^{n\Delta^{-g} \beta {\bf q}}(\xi-1). 
$$
We analyze the first term now, and will show later that the second term is small. Since ${\bf q}={\bf f}+\Delta^{g-1}(r_0+\Delta r_1+\Delta^2 r_2+\ldots)$ by the inductive hypothesis, we have, substituting the first term above into the expression in line~10 of Algorithm~\ref{alg:location},
\begin{equation*}
\begin{split}
\omega_\Delta^{-r\cdot \beta}\cdot \omega^{-n\Delta^{-g}{\bf f}\cdot \beta}\cdot \omega^{n\Delta^{-g} \beta {\bf q}}&=\omega_\Delta^{-r\cdot \beta}\cdot \omega^{n\Delta^{-g}({\bf q}-{\bf f})\cdot \beta}\\
&=\omega_\Delta^{-r\cdot \beta}\cdot \omega^{n\Delta^{-g}(\Delta^{g-1}(r_0+\Delta r_1+\Delta^2 r_2+\ldots))\cdot \beta}\\
&=\omega_\Delta^{-r\cdot \beta}\cdot \omega^{(n/\Delta)\cdot (r_0+\Delta r_1+\Delta^2 r_2+\ldots)\cdot \beta}\\
&=\omega_\Delta^{-r\cdot \beta}\cdot \omega_{\Delta}^{r_0\cdot \beta}\\
&=\omega_\Delta^{(-r+r_0)\cdot \beta}.
\end{split}
\end{equation*}
We used the fact that $\omega^{n/\Delta}=e^{2\pi i (n/\Delta)/n}=e^{2\pi i/\Delta}=\omega_\Delta$ and $(\omega_{\Delta})^\Delta=1$. Thus, we have
\begin{equation}\label{eq:92hg34grggggdds}
\omega_{\Delta}^{-r\cdot \beta}\omega^{-(n2^{-g}{\bf f})\cdot \beta}\frac{m_j(x', H, \alpha+\h\cdot \beta)}{m_j(x', H, \alpha)}=\omega_\Delta^{(-r+r_0)\cdot \beta}+\omega_\Delta^{(-r+r_0)\cdot \beta}(\xi-1).
\end{equation}

We now consider two cases. First suppose that $r=r_0$. Then $\omega_\Delta^{(-r+r_0)\cdot \beta}=1$, and it remains to note that  by~\eqref{eq:bound-1-oigb344tg32t} we have $|\xi-1|\leq \frac{1+1/8}{1-1/8}-1\leq 2/7< 1/3$.
Thus every $(\alpha, \beta)\in \A^*$ passes the test in line~11 of Algorithm~\ref{alg:location}. Since $|\A^*|>(3/5)|\A|$ by the argument above, we have that $r_0$ passes the test in line~11. It remains to show that $r_0$ is the unique element in $0,\ldots, \Delta-1$ that passes this test.

Now suppose that $r\neq r_0$. Then by the assumption that $\{\beta\}_{(\alpha, \beta)\in \A}$ is balanced (assumption {\bf 3} of the lemma) at least $49/100$ fraction of $\omega_\Delta^{(-r+r_0)\cdot \beta}$ have negative real part.  This means that for at least $49/100$ of $(\alpha, \beta)\in \A$ we have using triangle inequality
\begin{equation*}
\begin{split}
\left|\left[\omega_\Delta^{(-r+r_0)\cdot \beta}+\omega_\Delta^{(-r+r_0)\cdot \beta}(\xi-1)\right]-1\right|&\geq \left|\omega_\Delta^{(-r+r_0)\cdot \beta}-1\right|-\left|\omega_\Delta^{(-r+r_0)\cdot \beta}(\xi-1)\right|\\
&\geq \left|\mathbf{i}-1\right|-1/3\\
&\geq \sqrt{2}-1/3> 1/3,
\end{split}
\end{equation*}
and hence the condition in line~11 of Algorithm~\ref{alg:location} is not satisfied for any $r\neq r_0$. This shows that location is successful and completes the proof of correctness.

{\bf Runtime.} We perform $|\A|\cdot |\H|=O(\log n)$ invocations of \textsc{HashToBins} in line~1 of the algorithm. Each invocation costs $O(FB\log B+||\chi||_0\log n)$ by Lemma~\ref{l:hashtobins}, for a total runtime of $O(FB\log B \log n+||\chi||_0\log^2 n)$ for line~1.

After this for each of $B$ buckets the algorithm performs decoding in blocks of $\log \Delta$ bits, amounting to $O(\log_\Delta n)$ iterations. The decoding of each block requires looping over $\Delta$ possibilities, and testing each against the evaluation points in $\A$. Since $|\A|=O(\log \log n)$. Thus the total runtime is $O(\log_\Delta n\cdot |\A|\cdot \Delta)=O(\Delta \log n)=O(\log^2 n)$, as $\log \Delta=\Theta(\log\log n)$ and $\Delta=O(\log n)$. The total runtime is thus $O(FB\log^2 n+||\chi||_0\log^2 n)$, as claimed.
\end{proof}

We also get an immediate corollary of Lemma~\ref{lm:loc}. 

\noindent{{\bf Lemma~\ref{cor:loc}} (Restated from section~\ref{sec:loc-tail-noise} \em
For any integer $r_{max}\geq 1$,  for any sequence of $r_{max}$ hashings $H_r=(\pi_r, B, R), r\in [1:r_{max}]$ and evaluation points $\A_r\subseteq \nsq\times \nsq$,  for every $S\subseteq \nsq$ and for every $x, \chi\in \C^{\nsq}, x':=x-\chi$, the following conditions hold.
If for each $r\in [1:r_{max}]$ $L_r\subseteq \nsq$ denotes the output of \Call{LocateSignal}{$\wh{x}, \chi, H_r, \{m(x, H_r, \alpha+\h\cdot \beta)\}_{(\alpha, \beta)\in \A_r, \h\in \H}$}, $L=\bigcup_{r=1}^{r_{max}} L_r$, and the sets $\{\beta\}_{(\alpha, \beta)\in \A_r}$ are balanced  $r\in [1:r_{max}]$, then
\begin{equation}
||x'_{S\setminus L}||_1\leq 20 ||e^{head}_S(\{H_r\}, x')||_1+20 ||e^{tail, \H}_S(\{H_r, \A_r\}, x)||_1+|S|\cdot n^{-\Omega(c)}.\tag{*}
\end{equation}
Furthermore, every element $i\in S$ such that 
\begin{equation}
|x'_i|>20 (e^{head}_i(\{H_r\}, x')+e^{tail, \H}_i(\{H_r, \A_r\}, x))+n^{-\Omega(c)}\tag{**}
\end{equation}
belongs to $L$.
}
\begin{proof}
Suppose that $i\in S$ fails to be located in any of the $R$ calls, and $|x'_i|\geq n^{-\Omega(c)}$. By Lemma~\ref{lm:loc} and the assumption that the sets $\{\beta\}_{(\alpha, \beta)\in \A_r}$ are balanced for all $r\in [1:r_{max}]$ this means that for at least one half of values $r\in [1:r_{max}]$  either {\bf (A)} $e^{head}_{i}(H_r, x')\geq |x'_i|/20$ or {\bf (B)} $e^{tail}_{i}(H_r, \{\alpha+\h\cdot \beta\}_{(\alpha, \beta)\in \A_r}, x)> |x'_i|/20$ for at least one $\h\in \H$. We consider these two cases separately.

\paragraph{Case (A).} In this case we have $e^{head}_{i}(H_s, x')\geq |x'_i|/20$ for at least one half of $r\in [1:r_{max}]$, so 
in particular $e^{head}_i(\{H_r\}, x')\geq \text{quant}^{1/5}_r e^{head}_{i}(H_r, x')\geq |x'_i|/20$.

\paragraph{Case (B).} Suppose that $e^{tail}_{i}(H_r, \{\alpha+\h\cdot \beta\}_{(\alpha, \beta)\in \A_r}, x)> |x'_i|/20$ for some $\h=\h(r)\in \H$ for at least one half of $r\in [1:r_{max}]$ (denote this set by $Q\subseteq [1:r_{max}]$). We then have
\begin{equation*}
\begin{split}
e^{tail, \H}_i(\{H_r, \A_r\}, x)&=\quant^{1/5}_{r\in [1:r_{max}]} e^{tail}_i(H_r, \A_r, x)\\
&=\quant^{1/5}_{r\in [1:r_{max}]} \left[40\mu_{H_r, i}(x)+\sum_{\h\in \H} \left|e^{tail}_i(H_r, \{\alpha+\h\cdot \beta\}_{(\alpha, \beta)\in \A_r}, x)-40\mu_{H_r, i}(x)\right|_+\right]\\
&\geq \min_{r\in Q} \left[40\mu_{H_r, i}(x)+\left|e^{tail}_i(H_r, \{\alpha+\h(r)\cdot \beta\}_{(\alpha, \beta)\in \A_r}, x)-40\mu_{H_r, i}(x)\right|_+\right]\\
&\geq \min_{r\in Q} e^{tail}_i(H_r, \{\alpha+\h(r)\cdot \beta\}_{(\alpha, \beta)\in \A_r}, x)\\
&\geq |x'_i|/20
\end{split}
\end{equation*}
 as required. This completes the proof of {\bf (*)} as well as {\bf (**)}.
\end{proof}

%% file: app-tail.tex
%!TEX root = ./ft-hd-wrap.tex 
\section{Proof of Lemma~\ref{eq:tail-bound} (tail noise error bounds)}\label{app:tail-bounds}

We will use 
\begin{lemma}[Lemma~6.6 of~\cite{K16}]\label{lm:loc-tail-small}
For any constant $C'>0$ there exists an absolute constant $C>0$ such that for any $x\in \C^n$, any integer $k\geq 1$ and $S\subseteq \nsq$ such that $||x_{\nsq\setminus S}||_\infty\leq C'||x_{\nsq\setminus [k]}||/\sqrt{k}$, if $B\geq 1$, then the following conditions hold for $e^{tail, \H}$ defined with respect to $S$.

If hashings $H_r=(\pi_r, B, F), F\geq 2$ and sets $\A_r, |\A_r|\geq c_{max}$ for $r=1,\ldots, r_{max}$ are chosen at random, then
for every $i\in \nsq$ one has 
$\expect_{\{(H_r, \A_r)\}}\left[e^{tail, \H}_i(\{H_r, \A_r\}, x)\right]\leq  C (40+|\H|2^{-\Omega(c_{max})})  ||x_{\nsq \setminus [k]}||_2/\sqrt{B}$.

\end{lemma}
Note that this lemma was stated in~\cite{K16} with slightly different notation ($e^{tail}$ instead of $e^{tail, \H}$).

\begin{proofof}{Lemma~\ref{eq:tail-bound}}
First recall that by Lemma~\ref{lm:loc-tail-small} for every $t\in [1:T]$,  $s\in [1:R_t]$ and $i\in S$ one has 
$$
\expect_{H_{t, s}, \A_{t, s}}\left[e^{tail, \H}_i(H_{t, s}, \A_{t, s}, x)\right]=\nu^2,
$$
where $\nu^2\leq C' ||x_{\nsq\setminus S}||_2/\sqrt{B_t}$ for an absolute constant $C'>0$ (we used the fact that $|\H|=O(\log N)$ and $|\A|=C''\log\log N$ for a sufficiently large absolute constant $C''$).

To upper bound $\expect_{\{H_{t, s}, \A_{t, s}\}}\left[||e^{tail, \H}_{S_t}(\{H_{t, s}, \A_{t, s}\}, x)||_1\right]$, we note that by conditioning on $\E_{partition}$ we have 
$|S_t|\leq 2k\frac{R_0}{R_{t-1}}2^{-2^{(1-\delta)(t-1)}+1}$.
Letting $U:=2k\frac{R_0}{R_{t-1}}2^{-2^{(1-\delta)(t-1)}+1}$ to simplify notation, we get that
\begin{equation}\label{eq:23ht932hh23ht23t3ffcc}
\begin{split}
\expect_{\{H_{t, s}, \A_{t, s}\}}\left[||e^{tail, \H}_{S_t}(H_{t, s}, \A_{t, s}, x)||_1\right]&\leq \expect\left[\max_{Q\subseteq S, |Q|\leq U} ||e^{tail, \H}_Q(H_{t, s}, \A_{t, s}, x)||_1\right]\\
\end{split}
\end{equation}
We now recall that by~\eqref{eq:et}
\begin{equation*}
e^{tail, \H}_i(\{H_{t, s}, \A_{t, s}\}, x):=\quant^{1/5}_{s=1,\ldots, R_t} e^{tail, \H}_i(H_{t, s} \A_{t, s}, x),
\end{equation*}
 and apply Lemma~\ref{lm:noisiest-buckets} with $\gamma=1/5$, $m=|S|, n=R_t$ and 
$$
X_i^s=e^{tail, \H}_i(H_{t, s} \A_{t, s}, x)\text{~~~for~}i\in S\text{~and~}s=1,\ldots, R_t,
$$
so that $\expect_{H_{t, s}, \A_{t, s}}[X_i^s]\leq \nu$ for each $i\in S$, $s=1,\ldots, R_t$. Note that $Y_i:=\quant^{1/5}_{s=1,\ldots, R_t} X_i^s=e^{tail}_i(\{H_{t, s}, z_{t, s}\}, x)$ is exactly the quantity that we are interested in. 
We thus have by Lemma~\ref{lm:noisiest-buckets}
\begin{equation}\label{eq:1tgwoprg11}
\begin{split}
\expect_{\{H_{t, s}, \A_{t, s}\}}\left[\max_{Q\subseteq S, |Q|\leq U} ||e^{tail, \H}_Q(H_{t, s}, \A_{t, s}, x)||_1\right]&=\expect_{\{H_{t, s}, \A_{t, s}\}}\left[\max_{Q\subseteq S, |Q|\leq U} \sum_{i\in Q}Y_i\right]\\
&\leq U\cdot (20e\nu)\cdot \left(|S|/U\right)^{10/R_t}
\end{split}
\end{equation}

Since $R_{t'}=C_1 2^{t'}$ for every $t'$,  $|S|=|S_0|\leq 2k$ and $U=2k\frac{R_0}{R_{t-1}}2^{-2^{(1-\delta)(t-1)}+1}=2k2^{-2^{(1-\delta)(t-1)}+1-(t-1)}$, we have
$$
\left(|S|/U\right)^{10/R_t}=2^{10(2^{(1-\delta)(t-1)}-1+(t-1))/(C_1 2^t)}\leq 2^{10(1+(t-1)/2^t)/C_1}\leq 2^{20/C_1}\leq 2
$$
for all $t\geq 0$ as long as $C_1>20$. Substituting the above into~\eqref{eq:1tgwoprg11}, we get 
\begin{equation*}
\expect\left[\max_{Q\subseteq S, |Q|\leq U} ||e^{tail}_Q(\{H_{t, s}, a_{t, s}\}_{s\in [1:R_{t}]}, x)||_1\right]\leq (40e) \cdot U\cdot \nu.
\end{equation*}

We thus get by combining the above with ~\eqref{eq:23ht932hh23ht23t3ffcc}
$$
\expect_{H_{t, s}, \A_{t, s}}\left[||e^{tail, \H}_{S_t}(H_{t, s}, \A_{t, s}, x)||_1\right]=(40 e C') U ||x_{\nsq\setminus [k]}||_2/\sqrt{B_t}.
$$

We now use assumption {\bf q2} of the lemma to upper bound 
\begin{equation}\label{eq:92hg423g}
\begin{split}
(40e C') U/\sqrt{B_t}&\leq (40e C')\left(2k\frac{R_0}{R_{t-1}}\cdot 2^{-2^{(1-\delta)(t-1)}+1}\right)/\sqrt{C_2 2k/ R_t^2}\\
&\leq \frac{\sqrt{k}}{R_{t-1}}\cdot \frac1{R_t}\cdot \left(\frac{(80e C') C_1}{\sqrt{2C_2}} R_t^2\cdot 2^{-2^{(1-\delta) (t-1)}+1}\right)\\
&= \frac{\sqrt{k}}{R_{t-1}}\cdot \frac1{R_t}\cdot \left(\frac{(80e C') C_1^2}{\sqrt{2C_2}} 2^{2t}\cdot 2^{-2^{(1-\delta)(t-1)}+1}\right)\\
&\leq \frac1{2}\frac{\sqrt{k}}{R_{t-1}}\cdot \frac1{R_t}
\end{split}
\end{equation}
as long as $C_2$ is sufficiently large as a function of $C'$ and $C_1$ (to ensure that $\frac{(80e C') C' C_1^2}{\sqrt{2C_2}} 2^{2t}2^{-2^{(1-\delta)(t-1)}}\leq  1$ for all $t\geq 1$; see Claim~\ref{cl:max-expr}).
Substituting this bound into the upper bound on the expectation above yields $\expect_{\{H_{t, s}, \A_{t, s}\}}\left[||e^{tail, \H}_{S_t}(\{H_{t, s}, \A_{t, s}\}, x)||_1\right]\leq \frac1{R_t}\cdot ||x_{\nsq\setminus [k]}||_2 \sqrt{k}/R_{t-1}$
for every $t\geq 1$. It now follows by Markov's inequality that for every $t\geq 1$
$$
\prob_{\{H_{t, s}\}_{s\in [1:R_t]}}\left[||e^{tail, \H}_{S_t}(\{H_{t, s}, \A_{t, s}\}, x)||_1> \frac1{200}||x_{\nsq\setminus [k]}||_2 \sqrt{k}/R_{t-1}\right]\leq 200/R_t.
$$
By a union bound over $t=1,\ldots, T$ we have
\begin{equation*}
\begin{split}
&\prob_{\{\{H_{t, s}\}_{s\in [1:R_t]}\}_{t=1}^T}\left[||e^{tail, \H}_{S_t}(\{H_{t, s}, \A_{t, s}\}, x)||_1\leq \frac1{200}||x_{\nsq\setminus [k]}||_2 \sqrt{k}/R_{t-1}\text{~for all~}t=1,\ldots, T\right]\\
&\geq 1-\sum_{t=1}^T 200/R_t\geq 1-\sum_{t=1}^T 200/(C_1 2^t) \geq 1-O(1/C_1),
\end{split}
\end{equation*}
which gives the result as long as $C_1$ is larger than a constant, as required. Letting $\E_{small-noise}$ denote the intersection of the success events above completes the proof.
\end{proofof}

\section{Semi-equispaced Fourier transform}\label{sec:semi_equi}

One of the steps of our algorithm is to take the Fourier transform of our current estimate of $x$, so that it can be subtracted off in frequency domain and we can work with the residual.  The \emph{semi-equispaced FFT} provides an efficient method for doing this, and is based on the application of the standard inverse FFT to a filtered and downsampled signal. The following guarantee, which we rely on, was given in \cite[Sec.~12]{IKP}:

\begin{lemma} \emph{\cite[Lemma 12.1, Cor.~12.2]{IKP}} \label{lem:semi_equi_std}
    {\bf (a)} Fix a power of two $n$ and a constant $\delta > 0$.  For every $x \in \C^n$  the procedure \textsc{SemiEquiFFT}($x,k,\delta$) returns a set of values $\{ \wh{y}_j\}_{|j| \le k/2}$ in time $O(\| x\|_0\log (1/\delta)+k \log k)$, satisfying
    $$ |\wh{y}_j - \wh{x}_j| \le \delta \|x\|_2 $$
    for every $j, |j|\leq k/2$.
    
    {\bf (b)} Given two additional parameters $\sigma,\Delta \in [n]$ with $\sigma$ odd, it is possible to compute a set of values $\{ \wh{y}_j\}$ for all $j$ equaling $\sigma j' + \Delta$ for some $j'$ with $|j'| \le k/2$, with the same running time and approximation guarantee.
\end{lemma}
\begin{remark}
We note that Corollary 12.2 is not stated in this form, but rather for the special case when the sparsity of the signal that we are working with is comparable with the length of the interval. The more general bounds stated above follow immediately from their proof.
\end{remark}

We will also use 
\begin{lemma} \emph{\cite[Lemma 12.3]{IKP}} \label{lem:semi_equi_std_inv}
   Fix a power of two $n$ and a constant $\delta > 0$.  For every integer $k>1$, every $S\subseteq [n]$, $|S|=k$, every $\wh{x} \in \C^n$ such that $\supp (\wh{x})\subseteq [-k, k]$ it is possible to compute 
a set of values $\{ {y}_j\}_{j\in S}$ in time $O(k\log (n/\delta))$ satisfying
    $$ |{y}_j - {x}_j| \le \delta \|x\|_2. $$
    
\end{lemma}

%% file: app.tex
%!TEX root = ./ft-hd-wrap.tex 
\iflong{

\section{Basic theorems}\label{app:basic}
\subsection{Basic identities involving the Fourier transform}

Recall that we use the following normalization of the Fourier transform (as per~\eqref{eq:dft-forward}):
\begin{equation}\label{eq:dft-forward-app}
\hat x_f=\frac1{\sqrt{n}}\sum_{i\in \nsq}  \omega^{-if}x_i \text{~~and~~}x_{j}=\frac1{\sqrt{n}}\sum_{f\in \nsq}  \omega^{jf}\hat x_f
\end{equation}
We also use  $\F$ and  $\F^{-1}$ to denote the forward and inverse Fourier transforms respectively. Covolution is denoted by  $(x*y)_i=\sum_{j\in \nsq} x_{i-j} y_j$.
With this normalization of the Fourier transform the convolution theorem takes form 
\begin{equation}\label{eq:convolution-thm}
\begin{split}
({\mathcal F}^{-1}(\wh{x}\cdot \wh{y}))_i&=\frac1{\sqrt{n}}\sum_{f\in \nsq}  \omega^{if} \wh{x}_f \cdot \wh{y}_f\\
&=\frac1{\sqrt{n}}\sum_{f\in \nsq}  \omega^{if} \frac1{n}\sum_{i', i''\in \nsq} x_{i'}y_{i''} \omega^{-f (i'+i'')}\\
&=\frac1{\sqrt{n}} \sum_{i', i''\in \nsq} x_{i'} y_{i''} \cdot \frac1{n}\sum_{f\in \nsq}\omega^{f (i'+i''-i)}\\
&=\frac1{\sqrt{n}} \sum_{f'\in \nsq} x_{f'} y_{i-i'}\\
\end{split}
\end{equation}

%Parseval's equality takes form 
%\begin{equation*}
%\begin{split}
%||x||_2^2=||\wh{x}||_2^2
%\end{split}
%\end{equation*}

\subsection{Proof of Claim~\ref{cl:max-expr}}

We restate the claim here for convenience of the reader:

{\noindent {\bf Claim~\ref{cl:max-expr}} \em
For every $C_1, C_2> 0, \delta\in (0, 1)$ there exists $C_3$ such that for every $C_4\geq C_3$ one has 
$\frac1{C_4}2^{C_1 t}\cdot 2^{-C_2 2^{(1-\delta)t}+1}\leq 1$ all $t\geq 0$.
}

\begin{proof}
One has $2^{C_1 t}\cdot 2^{-C_22^{(1-\delta)t+1}}=2^{C_1 t-C_2 2^{(1-\delta)t}+1}$, so it suffices to note that since $\delta<1$ by assumption of the claim,
\begin{equation*}
\max_{t\geq 0} (2C_1 t-C_2^{(1-\delta)t})
\end{equation*}
is a constant(that may depend on $C_1, C_2$ and $\delta$. Thus, the claim follows for sufficiently large $C_3$ (as a function of $C_1, C_2, \delta$).

\end{proof}

}